\documentclass[acmsmall,screen,dvipsnames,x11names,svgnames]{acmart} 
\renewcommand\footnotetextcopyrightpermission[1]{}

\usepackage{booktabs}   
\usepackage{subcaption} 

\usepackage{microtype}
\usepackage{adjustbox}
\usepackage{bbding}

\usepackage{color}
\usepackage{xcolor}

\usepackage{cancel}
\usepackage{csquotes}

\usepackage{faktor}

\usepackage{mdframed}

\usepackage{stmaryrd}
\usepackage{tensor}
\usepackage{xfrac}
\usepackage{xspace}
\usepackage{tabularx}
\usepackage{semantic}
\usepackage{nicefrac}
\usepackage{needspace}

\input{insbox}

\usepackage{tcolorbox}
\tcbuselibrary{theorems}

\usepackage{tikz}
	\usetikzlibrary{fit,backgrounds,decorations.pathmorphing,decorations.fractals,shapes,calc,patterns,arrows.meta,arrows,decorations.pathmorphing,positioning,fit,trees,shapes,shadows,automata,calc,decorations.pathreplacing}
\pgfdeclaredecoration{penciline}{initial}{
    \state{initial}[width=+\pgfdecoratedinputsegmentremainingdistance,
    auto corner on length=1mm,]{
        \pgfpathcurveto%
        {
            \pgfqpoint{\pgfdecoratedinputsegmentremainingdistance}
                      {\pgfdecorationsegmentamplitude}
        }
        {
        \pgfmathrand
        \pgfpointadd{\pgfqpoint{\pgfdecoratedinputsegmentremainingdistance}{0pt}}
                    {\pgfqpoint{-\pgfdecorationsegmentaspect
                     \pgfdecoratedinputsegmentremainingdistance}%
                               {\pgfmathresult\pgfdecorationsegmentamplitude}
                    }
        }
        {
        \pgfpointadd{\pgfpointdecoratedinputsegmentlast}{\pgfpoint{1pt}{1pt}}
        }
    }
    \state{final}{}
}
\pgfdeclaredecoration{free hand}{start}
{
  \state{start}[width = +0pt,
                next state=step,
                persistent precomputation = \pgfdecoratepathhascornerstrue]{}
  \state{step}[auto end on length    = 3pt,
               auto corner on length = 3pt,               
               width=+2pt]
  {
    \pgfpathlineto{
      \pgfpointadd
      {\pgfpoint{2pt}{0pt}}
      {\pgfpoint{rand*0.15pt}{rand*0.15pt}}
    }
  }
  \state{final}
  {}
}

\usepackage[textwidth=1.25cm]{todonotes}
\usepackage{mathtools}

\usepackage[nameinlink,noabbrev]{cleveref}
\usepackage{esrelation}
\usepackage{listings}

\DeclareMathAlphabet{\mathpzc}{OT1}{pzc}{m}{it}

\newcommand{\sk}{\mathfrak{s}} 
\newcommand{\hh}{\mathfrak{h}} 
\newcommand{\skhh}{(\sk,\hh)} 
\newcommand{\Stacks}{\mathsf{Stacks}}
\newcommand{\Heaps}{\mathsf{Heaps}}
\newcommand{\emptyheap}{\hh_\emptyset}
\newcommand{\singleheap}[2]{\{ #1 \,\mapsto\, #2\}}
\newcommand{\loc}{\ell}
\newcommand{\val}{\vala}
\newcommand{\vala}{v}
\newcommand{\valb}{w}
\newcommand{\Loc}{L}
\newcommand{\disjoint}{\mathrel{\bot}}
\newcommand{\States}{\mathsf{States}}
\newcommand{\To}{\rightarrow}
\newcommand{\true}{\mathsf{true}}
\newcommand{\false}{\mathsf{false}}
\newcommand{\mydot}{\text{{\Large\textbf{.}}~}}
\newcommand{\guard}{\varphi}

\newcommand{\qif}{\quad\textnormal{if}~}
\newcommand{\tif}{~\textnormal{if}~}

\newcommand{\qiff}{\quad\textnormal{iff}\quad}
\newcommand{\qqiff}{\qquad\textnormal{iff}\qquad}

\newcommand{\qand}{\quad\textnormal{and}\quad}
\newcommand{\qqand}{\qquad\textnormal{and}\qquad}

\newcommand{\qimplies}{\quad\textnormal{implies}\quad}
\newcommand{\qqimplies}{\qquad\textnormal{implies}\qquad}

\newcommand{\ppreceq}{~{}\preceq{}~}
\newcommand{\ssucceq}{~{}\succeq{}~}
\newcommand{\eeq}{~{}={}~}
\newcommand{\defeq}{\coloneqq}
\newcommand{\ddefeq}{~{}\coloneqq{}~}

\newcommand{\qeq}{\quad{}={}\quad}
\newcommand{\qqeq}{\qquad{}={}\qquad}
\newcommand{\lleq}{~{}\leq{}~}

\newcommand{\mmodels}{~{}\models{}~}

\newcommand{\morespace}[1]{~{}#1{}~}
\newcommand{\ttimes}{\morespace{\times}}

\newcommand{\qedtriangle}{\hfill $\triangle$}

\newcommand{\RSL}{\sfsymbol{RSL}\xspace}
\newcommand{\SL}{\sfsymbol{SL}\xspace}

\newcommand{\slemp}{\sfsymbol{\textbf{emp}}}

\newcommand{\iemp}{\iiverson{\slemp}}
\newcommand{\slsingleton}[2]{{#1 \mapsto #2}}
\newcommand{\slvalidpointer}[1]{{#1 \mapsto \,{-}\,}}

\newcommand{\singleton}[2]{\iverson{#1 \mapsto #2}}

\newcommand{\isingleton}[2]{\iiverson{#1 \mapsto #2}}

\newcommand{\ivalidpointer}[1]{\iiverson{#1 \mapsto \,{-}\,}}

\newcommand{\sepcon}{\mathbin{{\star}}} 
\newcommand{\sepadd}{\mathbin{{\oplus}}}
\newcommand{\ssepadd}{\morespace{\sepadd}}
\newcommand{\sepimp}{\mathbin{\text{\raisebox{-0.1ex}{$\boldsymbol{{-}\hspace{-.55ex}{-}}$}}\hspace{-1ex}\text{\raisebox{0.13ex}{\rotatebox{-17}{$\star$}}}}}
\newcommand{\sepmon}{\mathbin{\text{\raisebox{-0.0ex}{$\boldsymbol{{-}\hspace{-.55ex}{-}}$}}\hspace{-0.3ex}\text{$\ominus$}}}

\newcommand{\emprun}[1]{\ensuremath{\underline{#1}}}
\newcommand{\pureemp}[1]{\emprun{\iiverson{#1}}}

\newcommand{\sfsymbol}[1]{\textsf{\upshape {#1}}}
\newcommand{\ttsymbol}[1]{\texttt{\upshape {#1}}}

\newcommand{\ertsymbol}{\sfsymbol{ert}}
\newcommand{\boldertsymbol}{\textbf{\sfsymbol{ert}}}
\newcommand{\ert}[2]{\ertsymbol \, \left\llbracket {#1} \right\rrbracket \, \left( {#2} \right)}
\newcommand{\ertC}[1]{\ertsymbol \, \left\llbracket {#1} \right\rrbracket}

\newcommand{\SKIP}{\ttsymbol{skip}}

\newcommand{\TICKsymbol}{\ensuremath{\textnormal{\texttt{tick}}}}
\newcommand{\TICK}[1]{\ensuremath{\TICKsymbol\left({#1}\right)}}

\newcommand{\AssignSymbol}{\mathrel{\textnormal{\texttt{:=}}}}
\newcommand{\ASSIGN}[2]{\ensuremath{#1 \AssignSymbol #2}}
\newcommand{\UNIFASSIGN}[3]{\ensuremath{\ASSIGN{#1}{\ttsymbol{Unif} \left(#2,#3\right)}}}

\newcommand{\ALLOC}[2]{\ensuremath{{#1} \AssignSymbol \mathtt{alloc}\left( #2 \right)}}
\newcommand{\ADDRESS}[1]{\langle#1\rangle}
\newcommand{\dereference}[1]{\ADDRESS{#1}}
\newcommand{\HASSIGN}[2]{\ensuremath{\dereference{#1} \AssignSymbol #2}}
\newcommand{\ASSIGNH}[2]{\ensuremath{#1 \AssignSymbol \dereference{#2}}}

\newcommand{\FREE}[1]{\ensuremath{\mathtt{free}(#1)}}

\newcommand{\COMPOSE}[2]{\ensuremath{{#1}{\,\fatsemi}~ {#2}}}
\newcommand{\PCHOICE}[3]{\ensuremath{\left\{\, {#1} \,\right\}\mathrel{\left[{#2}\right]}\left\{\, {#3} \,\right\}}}

\newcommand{\IFSYMBOL}{\ensuremath{\textnormal{\texttt{if}}}}
\newcommand{\IF}[1]{\ensuremath{\IFSYMBOL\,\left(\, {#1} \,\right)\,\{}}
\newcommand{\ELSESYMBOL}{\ensuremath{\textnormal{\texttt{else}}}}
\newcommand{\ELSE}{\ensuremath{\}\,\ELSESYMBOL\,\{}}
\newcommand{\ITE}[3]{\ensuremath{\IFSYMBOL\,\left(\, {#1} \,\right)\,\left\{\, {#2} \,\right\}\,\ELSESYMBOL\,\left\{\, {#3} \,\right\}}}

\newcommand{\WHILESYMBOL}{\ensuremath{\textnormal{\texttt{while}}}}
\newcommand{\WHILE}[1]{\ensuremath{\WHILESYMBOL \, \left(\, {#1} \,\right)\left\{\right.}}

\newcommand{\WHILEDO}[2]{\ensuremath{\WHILESYMBOL \, \left(\, {#1} \,\right)\left\{\, {#2} \,\right\}}}

\newcommand{\QSL}{\textnormal{\sfsymbol{QSL}}\xspace}

\newcommand{\hpgcl}{\textnormal{\sfsymbol{hpGCL}}\xspace}   
\newcommand{\boldhpgcl}{\textnormal{\textbf{\sfsymbol{hpGCL}}}\xspace}   
\newcommand{\Vars}{\ensuremath{\mathsf{Vars}}\xspace}   
\newcommand{\Vals}{\ensuremath{\mathsf{Vals}}\xspace}    
\newcommand{\Locs}{\ensuremath{\mathsf{{Loc}}\xspace}}    

\newcommand{\Nats}{\ensuremath{\mathbb{N}}\xspace}

\newcommand{\Rats}{\ensuremath{\mathbb{Q}}\xspace}
\newcommand{\Reals}{\ensuremath{\mathbb{R}}\xspace}
\newcommand{\PosReals}{\mathbb{R}_{\geq 0}}

\newcommand{\PosRealsInf}{\mathbb{R}_{\geq 0}^\infty}

\newcommand{\Confs}{\textsf{Conf}}
\newcommand{\E}{\mathbb{E}}

\newcommand{\A}{\mathbb{A}}
\newcommand{\Api}[1][\potent]{\A_{#1}}
\newcommand{\T}{\mathbb{T}}

\newcommand{\dom}[1]{\sfsymbol{dom}\left({#1}\right)}
\newcommand{\iverson}[1]{\left[ {#1} \right]}
\newcommand{\iiverson}[1]{\ensuremath{\Lbag #1 \Rbag}}

\newcommand{\subst}[2]{\left[ {#1} \middle/ {#2}\right]}
\newcommand{\statesubst}[2]{\left[ {#1} \leftarrow {#2}\right]}

\newcommand{\monus}{
	\mathbin{
		\vphantom{+}
		\text{
			\mathsurround=0pt 
			\ooalign{
				\noalign{\kern-.5ex}
				\hidewidth$\smash{\cdot}$\hidewidth\cr 
				\noalign{\kern.5ex}
				$-$\cr 
			}%
		}%
	}%
}

\newcommand{\charfun}[1]{\Phi_{#1}}
\newcommand{\aertcharfun}[1]{\Psi_{#1}}

\newcommand{\modC}[1]{\sfsymbol{Mod}\left( #1 \right)}

\newcommand{\ms}{\sigma}
\newcommand{\ma}{a}
\newcommand{\msched}{\mathfrak{S}}

\newcommand{\MDP}{\mathcal{M}}
\newcommand{\MS}{\mathcal{S}}
\newcommand{\MA}{\textsf{Act}}
\newcommand{\MP}{\textsf{prob}}
\newcommand{\MI}{\ms_{\text{init}}}
\renewcommand{\MR}{\textsf{rew}}

\newcommand{\MSCHED}{\textsf{Sched}}

\newcommand{\mstep}[2]{\xrightarrow[#1]{#2}}
\newcommand{\qmstep}[2]{\quad\mstep{#1}{#2}\quad}

\newcommand{\conf}[3]{#1,\, #2,\, #3}
\newcommand{\term}{\mathpzc{term}}
\newcommand{\sink}{\mathpzc{sink}}
\newcommand{\fail}{\mathpzc{fault}}
\newcommand{\tconf}[2]{\conf{\term}{#1}{#2}}

\newcommand{\MPOP}{\mathpzc{prob}}
\newcommand{\MROP}{\mathpzc{rew}}
\newcommand{\MDPOP}[1]{\mathpzc{M}\llbracket #1 \rrbracket}

\newcommand{\mpath}{\pi}
\newcommand{\mpaths}[2]{\textsf{Paths}_{=#1}(#2)}

\newcommand{\rt}{\ensuremath{{\textnormal{\textsf{rt}}}}}
\newcommand{\rtt}[2]{\ensuremath{\rt\llbracket #1 \rrbracket(#2)}}
\newcommand{\oprt}{\ensuremath{{\textnormal{\textsf{oprt}}}}}
\newcommand{\oprtt}[2]{\ensuremath{{\oprt\llbracket #1 \rrbracket(#2)}}}
\newcommand{\eert}{\ensuremath{{\textnormal{\textsf{eert}}}}}
\newcommand{\eertt}[2]{\ensuremath{{\eert\llbracket #1 \rrbracket(#2)}}}

\newcommand{\ExpRew}[1]{\textsf{ExpRew}\left(#1\right)}

\newcounter{blacktrianglelefteq}

\newcounter{leftsliceeq}

\newcommand{\pplus}{~{}+{}~}
\newcommand{\rtmin}{\ensuremath{\sqcap}}

\newcommand{\mminus}{~{}-{}~}

\newcommand{\qmid}{\quad{}|{}\quad}

\newcommand{\rta}{\ensuremath{f}}
\newcommand{\rtb}{\ensuremath{g}}
\newcommand{\rtc}{\ensuremath{u}}
\newcommand{\inv}{\ensuremath{I}}

\newcommand{\arta}{\ensuremath{X}}
\newcommand{\artb}{\ensuremath{Y}}
\newcommand{\artc}{\ensuremath{Z}}
\newcommand{\ainv}{\ensuremath{J}}

\newcommand{\preda}{\ensuremath{\varphi}} 
\newcommand{\predb}{\ensuremath{\psi}}

\newcommand{\setcomp}[2]{\left\{\, {#1} ~\middle|~ {#2} \,\right\}}

\definecolor{webgreen}{rgb}{0,.5,0}

\colorlet{myblue}{teal}

\definecolor{mydarkorange}{RGB}{220,80,0}
\colorlet{cemphcolor}{mydarkorange}

\newcommand{\underdot}[1]{%
    \tikz[baseline=(todotted.base)]{
\node[inner sep=1pt,outer sep=0pt] (todotted) {#1};
\draw[densely dotted] (todotted.south west) -- (todotted.south east);
    }%
}%

\newcommand{\cloze}[1]{\underdot{\hphantom{$#1$}}}

\newcommand{\ruletext}[1]{\textsf{#1}}

\newcommand{\annocolor}[1]{\textcolor{DodgerBlue3}{#1}}
\newcommand{\aannocolor}[1]{\textcolor{webgreen}{#1}}
\newcommand{\annotate}[1]{{\annocolor{\!\!{\fatslash}\!\!{\fatslash}~~\vphantom{G'} {#1}}}}

\newcommand{\starannotate}[1]{{\annocolor{{\talloblong}\!{\talloblong}\:\vphantom{G'} {#1}}}}

\newcommand{\phiannotate}[1]{{\annocolor{\!\!\hspace{-.15ex}{}^{\annocolor{\Phi}}\!\!\!{\fatslash}\!\!{\fatslash}~~\vphantom{G'} {#1}}}}

\newcommand{\psiannotate}[1]{{\aannocolor{\!\!\hspace{-.15ex}{}^{\aannocolor{\Psi}}\!\!\!{\fatslash}\!\!{\fatslash}~~\vphantom{G'} {#1}}}}

\newcommand{\lkpannotate}[1]{\annocolor{\!\!\hspace{-1.3ex}{}^{\annocolor{{\text{\tiny $\ruletext{lkp}$}}}}{\!\!\!{\fatslash}\!\!{\fatslash}~~\vphantom{G'} {#1}}}}

\newcommand{\succeqannotate}[1]{\annocolor{\!\!\hspace{-.25ex}{}^{\annocolor{{\succeq}}}{\!\!\!{\fatslash}\!\!{\fatslash}~~\vphantom{G'} {#1}}}}

\newcommand{\frameannotate}[2]{{\annocolor{\!\!\hspace{-2.25ex}{}^{\annocolor{\text{\tiny $\ruletext{frame}$}}}_{#2}\!\!\!{\fatslash}\!\!{\fatslash}~~\vphantom{G'} {#1}}}}

\newcommand{\unframeannotate}[2]{{\annocolor{\!\!\hspace{-4.0ex}{}^{\annocolor{\text{\tiny $\ruletext{unframe}$}}}_{#2}\!\!\!{\fatslash}\!\!{\fatslash}~~\vphantom{G'} {#1}}}}

\tcbset{colframe=blue,colback=white,left=0mm,right=0mm,top=0mm,bottom=0mm,boxsep=1mm,arc=0mm,boxrule=0.5pt}

\newcommand{\auxvarannotate}[1]{{\annocolor{\!\!\hspace{-0.90ex}{}^{\annocolor{\text{\tiny $\ruletext{aux}$}}}\!\!\!{\fatslash}\!\!{\fatslash}~~\vphantom{G'} {#1}}}}

\newcommand{\ertannotate}[1]{{\annocolor{\!\!\hspace{-0.5ex}{}^{\annocolor{\text{\tiny $\ertsymbol$}}}\!\!\!{\fatslash}\!\!{\fatslash}~~\vphantom{G'} {#1}}}}

\newcommand{\phantannotate}[1]{{\annocolor{\vphantom{\!\!{\fatslash}\!\!{\fatslash}~~}\vphantom{G'} {#1}}}}

\newcommand{\aaertannotate}[1]{{\aannocolor{\!\!\hspace{-1.3ex}{}^{\aannocolor{\text{\tiny $\aertsymbol$}}}\!\!\!{\fatslash}\!\!{\fatslash}~~\vphantom{G'} {#1}}}}

\newcommand{\aannotate}[1]{{\aannocolor{\!\!{\fatslash}\!\!{\fatslash}~~\vphantom{G'} {#1}}}}

\newcommand{\asucceqannotate}[1]{\aannocolor{\!\!\hspace{-.25ex}{}^{\aannocolor{{\succeq}}}{\!\!\!{\fatslash}\!\!{\fatslash}~~\vphantom{G'} {#1}}}}

\newcommand{\aframeannotate}[1]{{\aannocolor{\!\!\hspace{-2.25ex}{}^{\aannocolor{\text{\tiny $\ruletext{frame}$}}}\!\!\!{\fatslash}\!\!{\fatslash}~~\vphantom{G'} {#1}}}}

\newcommand{\aunframeannotate}[1]{{\aannocolor{\!\!\hspace{-4.0ex}{}^{\aannocolor{\text{\tiny $\ruletext{unframe}$}}}\!\!\!{\fatslash}\!\!{\fatslash}~~\vphantom{G'} {#1}}}}

\newcommand{\amutannotate}[1]{\aannocolor{\!\!\hspace{-1.3ex}{}^{\aannocolor{{\text{\tiny $\ruletext{mut}$}}}}{\!\!\!{\fatslash}\!\!{\fatslash}~~\vphantom{G'} {#1}}}}

\newcommand{\astarannotate}[1]{{\aannocolor{{\talloblong}\!{\talloblong}\:\vphantom{G'} {#1}}}}

\newcommand{\aphantannotate}[1]{{\aannocolor{\vphantom{\!\!{\fatslash}\!\!{\fatslash}~~}\vphantom{G'} {#1}}}}

\newcommand{\alkpannotate}[1]{\aannocolor{\!\!\hspace{-1.3ex}{}^{\aannocolor{{\text{\tiny $\ruletext{lkp}$}}}}{\!\!\!{\fatslash}\!\!{\fatslash}~~\vphantom{G'} {#1}}}}

\newcommand{\aauxvarannotate}[1]{{\aannocolor{\!\!\hspace{-0.90ex}{}^{\aannocolor{\text{\tiny $\ruletext{aux}$}}}\!\!\!{\fatslash}\!\!{\fatslash}~~\vphantom{G'} {#1}}}}

\newcounter{computationarrowsone}
\newcounter{computationarrowstwo}

\newcounter{sarrow}

\newcommand{\lfp}{\ensuremath{\textnormal{\sfsymbol{lfp}}~}}
\newcommand{\gfp}{\ensuremath{\textnormal{\sfsymbol{gfp}}~}}

\makeatletter
\newlength{\negph@wd}
\DeclareRobustCommand{\negphantom}[1]{%
  \ifmmode
    \mathpalette\negph@math{#1}%
  \else
    \negph@do{#1}%
  \fi
}
\newcommand{\negph@math}[2]{\negph@do{$\m@th#1#2$}}
\newcommand{\negph@do}[1]{%
  \settowidth{\negph@wd}{#1}%
  \hspace*{-\negph@wd}%
}
\makeatother

\newcommand{\aertsymbol}{\ensuremath{{\textnormal{\textsf{aert}}}}}
\newcommand{\aert}[3]{\aertOne{#1}{#2}{#3}} 
\newcommand{\aertOne}[3]{\aertsymbol_{#1}  \left\llbracket{#2} \right\rrbracket \left({#3}\right)}

\newcommand{\boldaertsymbol}{\ensuremath{\boldsymbol{{\textnormal{\textbf{\textsf{aert}}}}}}}

\newcommand{\potent}{\ensuremath{\pi}}

\newcommand{\Sup}{\reflectbox{\textnormal{\textsf{\fontfamily{phv}\selectfont S}}}\hspace{.2ex}}
\newcommand{\Inf}{\raisebox{.6\depth}{\rotatebox{-30}{\textnormal{\textsf{\fontfamily{phv}\selectfont \reflectbox{J}}}}\hspace{-.1ex}}}

\newcommand{\ie}{i.e.\xspace}

\newcommand{\eg}{e.\,g.\xspace}
\newcommand{\Eg}{E.\,g.\xspace}

\newcommand{\procfont}[1]{\textit{#1}}

\newcommand{\listtraverse}[1]{C}

\newcommand{\listinsert}[1]{\ensuremath{\procfont{Insert(#1)}}}
\newcommand{\arraycopy}[3]{\ensuremath{\procfont{ArrayCopy(#1,\, #2,\, #3)}}}
\newcommand{\arraydelete}[2]{\ensuremath{\procfont{DeleteArray(#1,\, #2)}}}

\newcommand{\ahead}{A}
\newcommand{\aoff}{o}
\newcommand{\asize}{s}

\newcommand{\expfont}[1]{\textsf{#1}}
\newcommand{\elist}[2]{\expfont{list}\left( #1,\, #2 \right)}
\newcommand{\esize}[2]{\expfont{listsize}\left( #1,\, #2 \right)}

\newcommand{\earray}[2]{\expfont{array}\left(#1,\, #2\right)}

\newcommand{\lh}{H}
\newcommand{\ls}{len}
\newcommand{\lpre}{pre}
\newcommand{\lend}{end}
\newcommand{\lsucc}{succ}
\newcommand{\dll}[5]{\expfont{dll}\left(#1,\, #2, \, #3, \, #4, \, #5 \right)}

\newcommand{\procsample}{\procfont{Sample}}
\newcommand{\procrank}{\procfont{Rank}}
\newcommand{\procremove}[1]{\procfont{remove} \left( #1 \right)}
\newcommand{\procdelete}[1]{\procfont{delete} \left( #1 \right)}
\newcommand{\procadd}[1]{\procfont{add} \left( #1 \right)}
\newcommand{\procinsert}[1]{\procfont{insert} \left( #1 \right)}
\newcommand{\procfindany}{\procfont{FindAny}}
\newcommand{\varrank}{rank}
\newcommand{\varany}{any}
\newcommand{\varremove}{x}
\newcommand{\varadd}{y}

\newcommand{\varvalrem}{v_x}
\newcommand{\varvalany}{v_{any}}

\newcommand{\varcur}{c}
\newcommand{\varvalcur}{v_c}
\newcommand{\varsamplen}{\ell}
\newcommand{\varlistelem}{y}

\newcommand{\poww}[1]{\textsf{pow2}\left( #1 \right)}
\newcommand{\nextpoww}[1]{\textsf{nextpow2}\left( #1 \right)}

\newcommand{\procrandlistinsert}[1]{\procfont{RandInsert}\left( #1\right)}
\newcommand{\procrandlistinsertcopy}[2]{\procfont{RandInsert}_{#1}\left( #2\right)}
\newcommand{\procbalancedlistinsert}[1]{\procfont{BalancedInsert}\left( #1\right)}

\AtEndPreamble{%
\theoremstyle{acmdefinition}
\newtheorem{remark}[theorem]{Remark}}

\begin{document}

\title{A Calculus for Amortized Expected Runtimes}         

\author[Batz]{Kevin Batz}
\authornote{Batz, Katoen, and Verscht are supported by the ERC AdG 787914 FRAPPANT.}
\orcid{0000-0001-8705-2564}             
\affiliation{
	\institution{RWTH Aachen University}           
	\country{Germany}                   
}
\email{kevin.batz@cs.rwth-aachen.de}         

\author[Kaminski]{Benjamin Lucien Kaminski}
\orcid{0000-0001-5185-2324}             
\affiliation{%
	\institution{Saarland University, Saarland Informatics Campus}
	\country{Germany}
}
\affiliation{%
	\institution{University College London}
	\country{United Kingdom}
}
\email{kaminski@cs.uni-saarland.de}          

\author[Katoen]{Joost-Pieter Katoen}
\authornotemark[1]
\orcid{0000-0002-6143-1926}             
\affiliation{
  \institution{RWTH Aachen University}           
  \country{Germany}                   
}
\email{katoen@cs.rwth-aachen.de}         

\author[Matheja]{Christoph Matheja}
\orcid{0000-0001-9151-0441}             
\affiliation{
	\institution{Technical University of Denmark}           
	\country{Denmark}                   
}
\email{chmat@dtu.dk}         

\author[Verscht]{Lena Verscht}
\authornotemark[1]
\orcid{0000-0001-6823-7918}             
\affiliation{
  \institution{RWTH Aachen University}           
  \country{Germany}                   
}
\email{lena.verscht@rwth-aachen.de}          

\begin{abstract}
	We develop a weakest-precondition-style calculus à la Dijkstra for reasoning about \emph{amortized expected runtimes of randomized algorithms with access to dynamic memory} --- the $\aertsymbol$ calculus.
	Our calculus is truly quantitative, \ie instead of Boolean valued predicates, it manipulates real-valued functions.
	
	En route to the $\aertsymbol$ calculus, we study the $\ertsymbol$ calculus for reasoning about \emph{expected runtimes} of 
	Kaminski et al.\ [2018]
	 extended by capabilities for handling dynamic memory, thus enabling \emph{compositional} and \emph{local} reasoning about \emph{randomized data structures}.
	This extension employs \emph{runtime separation logic}, which has been foreshadowed by 
	Matheja [2020]
	and then implemented in Isabelle/HOL by 
	Haslbeck [2021].
	In addition to Haslbeck's results, we further prove soundness of the so-extended $\ertsymbol$ calculus with respect to an operational Markov decision process model featuring countably-branching nondeterminism, provide extensive intuitive explanations, and provide proof rules enabling separation logic-style verification for upper bounds on expected runtimes.
	Finally, we build the so-called \emph{potential method} for amortized analysis into the $\ertsymbol$ calculus, thus obtaining the $\aertsymbol$ calculus. Soundness of the $\aertsymbol$ calculus is obtained from the soundness of the $\ertsymbol$ calculus and some probabilistic form of telescoping.
	
	Since one needs to be able to handle \emph{changes in potential} which can in principle be both positive or negative, the $\aertsymbol$ calculus needs to be --- essentially --- capable of handling certain signed random variables.
	A particularly pleasing feature of our solution is that, unlike \eg 
	Kozen [1985], we obtain a loop rule for our signed random variables, and furthermore, unlike \eg
	 Kaminski and Katoen [2017], the $\aertsymbol$ calculus makes do without the need for involved technical machinery keeping track of the integrability of the random variables.

	Finally, we present case studies, including a formal analysis of a randomized delete-insert-find-any set data structure
	[Brodal et al.\ 1996], which yields a constant expected runtime per operation, whereas no deterministic algorithm can achieve this.
\end{abstract}



\allowdisplaybreaks[0]

\maketitle

\section{Introduction}
\label{s:intro}

Amortized analysis~\cite{tarjan1985amortized} is a well-established method to analyze the runtime complexity of algorithms, in particular of those who manipulate dynamic data structures such as dynamically-sized lists, self-balancing trees, and so forth.
The essence of amortized analysis is to average the runtime of a single operation $\mathit{Op}$ over a long sequence of $\mathit{Op}$'s.
Why is this useful?
Suppose $\mathit{Op}$ has \emph{large worst-case runtime} and \emph{small \enquote{normal-case} runtime}.
A worst-case analysis of $\mathit{Op}$ would tell us that $\mathit{Op}$ performs poorly.
However, when executing a \emph{long sequence} of consecutive $\mathit{Op}$'s, it may be that the worst case inevitably occurs only very \emph{seldomly}.
The large runtimes of the small number of worst cases thus amortize over the large number of small runtimes of the normal cases.
\emph{On average}, the runtime of a single execution of $\mathit{Op}$ is thus actually small.
Amortized analysis hence yields results that are more realistic than worst-case analyses.
Notice that amortized analysis is \emph{not} the same as so-called \emph{average case analysis}.
The latter assumes a \emph{probability distribution over all possible inputs} to $\mathit{Op}$ and averages $\mathit{Op}$'s runtime over this distribution.

One popular technique for amortized analysis is the \emph{potential method} (aka \emph{physicist’s method}) \cite{tarjan1985amortized,leiserson2001introduction}.
We will introduce this method in some detail in \Cref{s:aert} and plot out how to make it probabilistic.
For that, we will develop a calculus for reasoning about the \emph{amortized runtime
complexity of randomized algorithms}. 
Various  randomized algorithms use dynamic data structures such as randomized meldable heaps, randomized splay trees and randomized search trees.
An amortized analysis gives a detailed account of the expected runtime of a randomized algorithm and extends (read: refines) existing runtime analysis techniques for probabilistic programs~\cite{DBLP:conf/lics/AvanziniLG19,DBLP:conf/pldi/NgoC018,kaminski2018weakest}.
For instance, an amortized analysis of the complexity of the randomized delete-insert-find-any set data structure~\cite{DBLP:journals/njc/BrodalCR96} yields a constant expected runtime per operation, whereas no deterministic algorithm can achieve this.
The aim of this paper is to develop a \emph{systematic, calculational method for carrying out an amortized runtime analysis of randomized algorithms on source code level}.
Our method is in the spirit of weakest-precondition style reasoning.
That is to say, we present a syntax-oriented technique to determine the amortized expected time of randomized algorithms by applying backward reasoning.
Our technique yields amortized upper bounds on the expected runtime complexity.

Our starting point is the $\ertsymbol$-calculus for reasoning about expected runtimes of probabilistic pointer programs by \citet{haslbeckPhd}, which extends the $\ertsymbol$-calculus of~\citet{kaminski2018weakest} by the principles of separation logic~\cite{DBLP:conf/lics/Reynolds02,DBLP:conf/popl/IshtiaqO01}.
The main challenge here is that classical separation logic connectives do not admit a frame rule --- the key to compositional and local \mbox{reasoning ---} for runtime \emph{over}-approximations.
Based on a suggestion in \cite[Chapter 9]{christophPhd}, \citet{haslbeckPhd} investigated runtime analogues to these separating connectives, separating sum and (its adjoint) separating monus, thus obtaining a real-valued \enquote{logic} --- \emph{runtime separation logic} (\textsf{RSL}) --- upon which the $\ertsymbol$-calculus is built.

 \citet{haslbeckPhd} has mechanized the $\ertsymbol$-calculus in Isabelle/HOL and has proven various properties such as the validity of the frame rule. In addition to Haslbeck's results, we further prove \emph{soundness} of $\ertsymbol$ by establishing a strong correspondence to a simple operational cost model defined in terms of Markov decision processes~(MDPs)~\cite{puterman2005markov}.
This resembles the approach adopted for quantitative separation logic~($\QSL$)~\cite{QSLpopl}, a version of separation logic to reason about the correctness (not the runtime) of probabilistic pointer programs.
The treatment of rewards in the operational model is rather different, however, as time may elapse at arbitrary steps in a program execution.
This contrasts the situation for $\QSL$ where rewards corresponding to post\enquote{conditions} are only collected in states indicating successful program termination.
The proof principle to establish the correspondence between \textsf{RSL} and the operational MDP interpretation is new and relies on Bellman equations~\cite{puterman2005markov} and Blackwell’s theorem~\cite{blackwell1967positive}.

In a second step, we extend $\ertsymbol$ to $\aertsymbol$ --- a calculus for reasoning about \emph{amortized expected runtimes using the potential method}.
We show that 
$\aertsymbol$ recovers what is essentially the \emph{telescoping} property of the classical potential method in the probabilistic setting. That is, for probabilistic program $C$ and potential function $\potent$:
$$
\overbrace{\aert{\potent}{C}{0}}^{\mathclap{\mbox{\footnotesize amortized expected runtime of $C$}}} 
\qquad\qqeq\qquad
\underbrace{\ert{C}{\potent}~{}-{}~ \potent}_{\mathclap{\mbox{\footnotesize expected runtime of $C$ + expected change of the potential caused by $C$}}} 
$$
This result enables us to derive a frame rule for local reasoning about amortized expected runtimes. 
As indicated above, an integral part of the potential method is reasoning about differences in potential when performing an operation.
Such differences can potentially become negative.
This seems rather innocent, but technically it is not.
Existing weakest precondition reasoning rules for probabilistic programs restrict random variables to be non-negative (at least for loops)~\cite{DBLP:series/mcs/McIverM05,DBLP:journals/jcss/Kozen85}.
This is for a good reason as it avoids issues with integrability of expected values.
Extensions which can handle signed random variables~\cite{DBLP:conf/lics/KaminskiK17} are technically involved.
(Indeed a na\"{i}ve approach would have been to extend the classical $\ertsymbol$ calculus with \textsf{RSL} and with these signed random variables.)
This paper shows that this complicated machinery is not needed to treat amortization. 
Another interesting result is that our framework recovers a classical result from amortized complexity analysis over sequences of programs (\Cref{thm:aert_sound}).

We illustrate our amortized runtime calculus on a few examples such as a randomized dynamic list as well as an analysis of the insert-delete-find-any set data structure from~\cite{DBLP:journals/njc/BrodalCR96}.
The latter example is of interest as it only has a constant amortized runtime per operation under randomization.
To the best of our knowledge, our analysis is the first such analysis using the potential method and on source code level.

To summarize, the main contributions of this paper are:
\begin{itemize}
\item 
a compositional, weakest-precondition-style calculus to reason about amortized expected runtimes of randomized algorithms that features local reasoning, 
\item invariant-based reasoning for loops and proof rules enabling the separation logic-style verification of such runtimes,
\item
soundness of our methods by providing a close correspondence to an operational model based on countably-branching Markov decision processes, and 
\item 
a source-code-level analysis based on the potential method for amortized complexity on the insert-delete-find-any data structure~\cite{DBLP:journals/njc/BrodalCR96}.
\end{itemize}%
\paragraph{Structure of the paper.} \Cref{s:hpgcl} introduces 
our model programming language, where \Cref{s:op-semantics} defines its operational semantics.
\Cref{s:rsl} studies and explains runtime separation logic. 
We present the $\ertsymbol$ calculus for expected runtimes in \Cref{s:ert}, where \Cref{s:op} features its soundness proof. 
In \Cref{s:aert} we present the $\aertsymbol$ calculus for reasoning about \emph{amortized} expected runtimes. We consider related work in \Cref{s:related}. 
Proofs and details on the case studies are found in the appendix.

\section{Probabilistic Pointer Programs}
\label{s:hpgcl}

We will employ an imperative model language \`{a} la Dijkstra's guarded commands language adapted from~\cite{QSLpopl} with three main features:
(1) \emph{probabilistic choices}, (2) a \emph{customizable runtime model}, and (3) statements for \emph{accessing and manipulating dynamic memory}.

\subsection{Program States}
Program states have two components: 
(1) a \emph{stack} assigning values to program variables and 
\mbox{(2) a \emph{heap}} modeling the dynamic memory which stores values at (dynamically) allocated locations.%

\paragraph{Stacks}
A stack $\sk$ is a mapping from \emph{variables} taken from a finite set $\Vars$ to values taken from a set \Vals; in our case values are natural numbers, \ie $\Vals = \Nats$.
Hence, the set of \emph{stacks} is given by%
\begin{align*}
	{\Stacks} \eeq \setcomp{\sk}{\sk \colon \Vars \To \Vals}~.
\end{align*}%
\paragraph{Heaps}
A heap $\hh$ maps finitely many \emph{memory locations} taken from the set $\Locs = \Nats_{{>} 0}$ to values; the value $0$ is \emph{not} a valid location and represents the null pointer. Hence, the set of \emph{heaps} is given by%
%
\begin{align*}
	{\Heaps} \eeq \setcomp{\hh}{\hh \colon \Loc \To \Vals,\quad \Loc \subseteq \Locs,\quad 
	\Loc~\text{finite} }~.
\end{align*}%
For a given heap $\hh \colon \Loc \To \Vals$, we denote by $\dom{\hh}$ its \emph{domain}, \ie
$\dom{\hh} = \Loc$.
We write $\hh_1 \disjoint \hh_2$ to indicate that the domains of heaps $\hh_1$ and $\hh_2$ are disjoint, \ie
\begin{align*}
	h_1 \disjoint h_2 \qiff \dom{h_1} \cap \dom{h_2} = \emptyset~.
\end{align*}%
For heaps $\hh_1$ and $\hh_2$ with disjoint domains, \ie $\hh_1 \disjoint \hh_2$, 
their \emph{union} $\hh_1 \sepcon \hh_2$ is given by
\begin{align*}
	\hh_1 \sepcon \hh_2\colon \quad \dom{\hh_1} \mathrel{\cup} \dom{\hh_2} \To \Vals, \quad \loc \mapsto \begin{cases}\hh_1(\loc), & \textnormal{if } \loc \in \dom{\hh_1} \\ \hh_2(\loc), & \textnormal{if } \loc \in \dom{\hh_2}~. \end{cases}
\end{align*}%
If the domains of $\hh_1$ and $\hh_2$ are not disjoint, $\hh_1 \sepcon \hh_2$ is undefined.

We denote by $\emptyheap$ the \emph{empty heap} with $\dom{\emptyheap} = \emptyset$.
Moreover, $\singleheap{\loc}{\val}$ denotes the heap $\hh$ that consists of a single memory location $\loc$ which stores value $\val$, \ie $\dom{\hh} = \{ \loc \}$ and $\hh(\loc) = \val$.
Note that $\hh \sepcon \emptyheap = \emptyheap \sepcon \hh= \hh$ for any heap $\hh$, whereas $\singleheap{\loc}{\val} \sepcon \singleheap{\loc}{\valb}$ is always undefined.
%
\paragraph{Program states}
The set $\States$ of \emph{program states} consists of all stack-heap pairs, \ie
%
\begin{align*}
	{\States} \eeq \setcomp{\skhh}{\sk \in \Stacks,~ \hh \in \Heaps}~.
\end{align*}%
Given a program state $\skhh$, we denote by \textbf{$\sk(e)$} the evaluation of an arithmetic expression $e$ in stack~$\sk$, i.e.\ the value that is obtained by evaluating expression $e$ after replacing every occurrence of every variable $x$ in $e$ by its assigned value $\sk(x)$.
Analogously, we denote by \textbf{$\sk \models \varphi$} that the boolean expression $\varphi$ evaluates to true in stack $\sk$.
We require that both arithmetic and boolean expressions are \emph{pure}, meaning that they only depend on variables in $\Vars$ and \emph{not} on the heap.
Evaluating an expression thus never causes any side effects, such as dereferencing an unallocated location. 

We write \textbf{$\sk\statesubst{x}{\val}$} for stack $\sk$ in which the value of variable \mbox{$x$ has been updated to $\val \in \Vals$, \ie}\footnote{We use $\lambda$--expressions to construct functions: $\lambda \xi \mydot \epsilon$ stands for the function that, when applied to an argument $\alpha$, evaluates to $\epsilon$ in which every occurrence of $\xi$ is replaced by $\alpha$.}%
\begin{align*}
	\sk\statesubst{x}{v} \eeq \lambda\, y\mydot \begin{cases}
		v, & \textnormal{if } y = x\\
		\sk(y), & \textnormal{if } y \neq x~.
	\end{cases}
\end{align*}%

Likewise, 
we write \textbf{$\hh\statesubst{\loc}{\val}$} for the heap $\hh$ in which the value stored at location~$\loc$ has been updated to $\val$.
Formally, if $\loc$ is allocated in $\hh$, \ie if $\loc \in \dom{\hh}$ (otherwise $\hh\statesubst{\loc}{v}$ is undefined),%
%
%
\begin{align*}
	\hh\statesubst{\loc}{v} \eeq \lambda\, \loc'\mydot \begin{cases}
		\val, & \textnormal{if } \loc' = \loc \\
		\hh(\loc'), & \textnormal{if } \loc' \neq \loc~.
	\end{cases}
\end{align*}%
\subsection{The Heap-manipulating Probabilistic Guarded Command Language}
\label{s:hpgcl-syntax}
We now present the syntax of our model programming language and briefly discuss its intuitive semantics and runtime model; a formal semantics and runtime model are provided in \Cref{s:op-semantics}. 

\paragraph{Syntax}
Programs written in the \emph{heap-manipulating probabilistic guarded command language}, denoted $\hpgcl$, are given by the context-free grammar%
%
\begin{center}
\begin{adjustbox}{max width=\linewidth}
\begin{minipage}{0.999\linewidth}
\begin{align*}
	\begin{array}{rlrlr}
		C~\morespace{\longrightarrow}~	& 		\TICK{e}				&\textnormal{(time consumption)}		&\qmid \ASSIGN{x}{e} 		&\textnormal{(variable assignment)} 		\\
									& \qmid 	\ALLOC{x}{e}			&\textnormal{(memory allocation)}		&\qmid \HASSIGN{e}{e'} 		&\textnormal{(heap mutation)} 			\\
									& \qmid 	\ASSIGNH{x}{e} 		&\textnormal{(heap lookup)}			&\qmid \FREE{e} 			&\textnormal{(memory deallocation)} 	\\
									& \qmid 	\PCHOICE{C}{p}{C}		&\textnormal{(probabilistic choice)}		&\qmid \ITE{\varphi}{C}{C}	 	& \textnormal{(conditional choice)} 		\\
									& \qmid 	\COMPOSE{C}{C}		&\textnormal{(sequential composition)}	&\qmid \WHILEDO{\varphi}{C} 	& \textnormal{(while loop)} 			\\
	\end{array}
\end{align*}
\end{minipage}
\end{adjustbox}
\end{center}%
\vspace{.5em}
where $x \in \Vars$, $e, e_1, \ldots, e_n$ are arithmetic expressions over variables that evaluate to values in $\Vals = \Nats$,
and $\varphi$ is a Boolean expression over variables. 
Moreover, $p$ is an arithmetic expression over variables that evaluates to a rational probability, \ie $\sk(p) \in [0,1] \cap \Rats$ holds for all stacks $\sk$.

\paragraph{Intuitive semantics}
Assignments, sequential composition, conditionals, and loops are standard. 
The probabilistic choice $\PCHOICE{C_1}{p}{C_2}$ executes $C_1$ with probability $\sk(p)$ and $C_2$ with probability $1-\sk(p)$.
$\TICK{e}$ does not affect the program state but takes $e$ units of time; see below.

The remaining statements access or manipulate the dynamic memory.
$\ALLOC{x}{e}$ allocates a block of $e$ consecutive, previously unallocated, and nondeterministically chosen memory locations, initializes their contents with zero,\footnote{Similarly to \texttt{C}'s \texttt{calloc}.} and assigns to $x$ the first of those locations; attempting to allocate an empty block, \eg via $\ALLOC{x}{0}$, only assigns a \emph{nondeterministic} value to $x$ but does not affect the heap.
Since we have an infinite reservoir of locations, \emph{memory allocation never fails}. 

The mutation statement $\HASSIGN{e}{e'}$ changes the value at location $e$ to $e'$.
Mutation \emph{can} fail:
If location $e$ is not allocated, we encounter \emph{undefined behavior} (most likely a crash) due to a \emph{memory fault}.
The lookup statement $\ASSIGNH{x}{e}$ assigns the value stored at location $e$ to variable $x$ if location $e$ is allocated and otherwise causes a memory fault.
Finally, the deallocation statement $\FREE{e}$ disposes of location $e$ if allocated and causes a memory fault otherwise.

\paragraph{Runtime model}
Our ultimate goal is to reason about (amortized) expected runtimes of $\hpgcl$ programs. To deal with a variety of runtime models, we do \emph{not} assign particular runtimes to individual statements. Rather, we model runtime using $\TICKsymbol$ statements; 
executing $\TICK{e}$ takes $e$ units of time (where $e$ is evaluated in the current program state). 
All other $\hpgcl$ statements have a runtime of zero with one exception: whenever we encounter a memory fault, this constitutes for us \emph{undefined behavior} --- anything can happen, including non-termination. 
Hence, our runtime model for memory faults is that they have unbounded, \ie infinite, runtime.

\paragraph{Modified variables}
We denote by $\modC{C} \subseteq \Vars$ the set of variables that are potentially \emph{modified} by program $C$, \ie occur on the left-hand side of a variable assignment $\ASSIGN{x}{e}$ in $C$.

\subsection{Formal Operational Semantics}
\label{s:op-semantics}

We give operational semantics to programs by 
(1) defining a \emph{small-step execution relation} $\mstep{}{}$ describing how (and how probable) statements manipulate program states, 
(2) constructing a \emph{Markov decision process} (MDP) based on $\mstep{}{}$, and
(3) introducing a \emph{reward function modeling runtimes}.
Expected runtimes of program executions will then be the expected rewards of a \mbox{corresponding MDP}.%

\paragraph{Configurations}
The set of \emph{program configurations} is given by
\begin{align*}
	\Confs 
	\qeq 
	\bigl(
		\hpgcl 
		\cup  
		\{\term,\, \fail\}
	\bigr)
	\ttimes 
	\Stacks \ttimes \Heaps 
	\quad{}\cup{}\quad
	\{\sink\}~.
\end{align*}
A configuration is an \hpgcl program $C$, or $\term$ indicating fault-free termination, or $\fail$, indicating a memory fault, together with a program state $\skhh$. 
For technical reasons, we also add a $\sink$ configuration which we enter after program termination.

\paragraph{Execution relation}
The steps of our operational semantics are given by an execution relation
\begin{align*}
  \mstep{}{}~ \quad\subseteq\quad \Confs ~\times~ \textsf{Prob} ~\times~ \Vals ~\times~ \Confs~,
\end{align*}
where $\textsf{Prob}$ is the set of transition probabilities\footnote{Formally, we set $\textsf{Prob} = ([0,1] \cap \Rats) \cup \{1-\sfrac{1}{2}\}$, where $1-\sfrac{1}{2} \neq \sfrac{1}{2}$ is a \emph{formal value} that allows us to represent two distinguishable steps from a configuration $c$ both to the same configuration $c'$, each with probability $\sfrac{1}{2}$.} and $\Vals$ are the allocation values which are chosen nondeterministically; if the step is not an allocation, we default to $0$.
Hence, $c \mstep{p}{\val} c'$ (denoting $(c, p, \val, c') \in {\mstep{}{}}$) indicates that our program performs \emph{one step} from $c$ to $c'$ with probability~$p$ while allocation value~$\val$ has been chosen.
To avoid cluttering, we omit $p$ if $p=1$ and $v$ if $v = 0$.
\mbox{$\mstep{}{}$ is} given by the inference rules in \Cref{fig:op:rules} which match the intuitive semantics of \Cref{s:hpgcl-syntax}. 
\Eg, the rule for $\ALLOC{x}{e}$ chooses a location $\val$ from $e$ consecutive unallocated locations. 
These locations are added to the heap with their content initialized to $0$. 
In particular, allocation never fails (steps into $(\fail,\ldots)$) and causes infinite branching over all such memory locations $\val$.%
\begin{figure}[t]
\small
\begin{align*}
& \conf{\TICK{e}}{\sk}{\hh} \qmstep{}{} \tconf{\sk}{\hh}
  \hspace{3em}
  \tconf{\sk}{\hh} \qmstep{}{} \sink 
  \hspace{3em} 
  \sink \mstep{}{} \sink
  \hspace{3em} 
  \conf{\fail}{\sk}{\hh} \mstep{}{} \sink
\\
& \conf{\PCHOICE{C_1}{p}{C_2}}{\sk}{\hh} \qmstep{\sk(p)}{} \conf{C_1}{\sk}{\hh}
  \hspace{5em}
  \conf{\PCHOICE{C_1}{p}{C_2}}{\sk}{\hh} \qmstep{1-\sk(p)}{} \conf{C_2}{\sk}{\hh}
\\
& \conf{\ASSIGN{x}{e}}{\sk}{\hh} \qmstep{}{} \tconf{\sk\statesubst{x}{\sk(e)}}{\hh}
  \hspace{2em}
  \conf{\ITE{\varphi}{C_1}{C_2}}{\sk}{\hh} \qmstep{}{} 
  \begin{cases}
    \conf{C_1}{\sk}{\hh} & \tif \sk \models \varphi \\
    \conf{C_2}{\sk}{\hh} & \tif \sk \not\models \varphi
  \end{cases}
\\
& \conf{\ALLOC{x}{e}}{\sk}{\hh} \qmstep{}{\val} 
  \begin{cases}
    \tconf{\sk\statesubst{x}{\val}}{\hh \sepcon \hh' & \tif \sk(e) = n > 0~\text{and}~\hh \disjoint \hh' \\
        & \qand \hh' = \singleheap{\val}{0} \sepcon \ldots \sepcon \singleheap{\val+n-1}{0}} \\
    \tconf{\sk\statesubst{x}{\val}}{\hh} & \tif \sk(e) = 0~\text{and}~\val \in \Vals
  \end{cases}
\\
& \conf{\ASSIGNH{x}{e}}{\sk}{\hh} \qmstep{}{}
  \begin{cases}
    \tconf{\sk\statesubst{x}{\hh(v)}}{\hh} & \qif \sk(e) = v \in \dom{\hh} \\
    \conf{\fail}{\sk}{\hh} & \qif \sk(e) \notin \dom{\hh}
  \end{cases}
\\
& \conf{\HASSIGN{e}{e'}}{\sk}{\hh} \qmstep{}{}
  \begin{cases}
    \tconf{\sk}{\hh\statesubst{v}{\sk(e')}} & \qif \sk(e) = v \in \dom{\hh} \\
    \conf{\fail}{\sk}{\hh} & \qif \sk(e) \notin \dom{\hh}
  \end{cases}
\\
& \conf{\FREE{e}}{\sk}{\hh} \qmstep{}{}
  \begin{cases}
    \tconf{\sk}{\hh'} & \qif \hh = \hh' \sepcon \{ \sk(e) \mapsto v \} ~\text{for some}~v \in \Vals \\
    \conf{\fail}{\sk}{\hh} & \quad\text{otherwise}
  \end{cases}
\\
& \conf{\WHILEDO{\varphi}{C}}{\sk}{\hh} \qmstep{}{} 
  \begin{cases}
    \conf{\COMPOSE{C}{\WHILEDO{\varphi}{C}}}{\sk}{\hh} & \qif \sk \models \varphi \\
    \conf{\term}{\sk}{\hh} & \qif \sk \not\models \varphi
  \end{cases}
\\
&
  \inference{
    \conf{C_1}{\sk}{\hh}
    \qmstep{p}{\val}
    \conf{C_1'}{\sk'}{\hh'}
  }{
    \conf{\COMPOSE{C_1}{C_2}}{\sk}{\hh}
    \qmstep{p}{\val}
    \conf{\COMPOSE{C_1'}{C_2}}{\sk'}{\hh'}
  }
  \hspace{0.3em}
  \inference{
    \conf{C_1}{\sk}{\hh}
    \qmstep{p}{\val}
    \tconf{\sk'}{\hh'}
  }{
    \conf{\COMPOSE{C_1}{C_2}}{\sk}{\hh}
    \qmstep{p}{\val}
    \conf{C_2}{\sk'}{\hh'}
  }
  \hspace{0.3em}
  \inference{
    \conf{C_1}{\sk}{\hh}
    \qmstep{p}{\val}
    \conf{\fail}{\sk}{\hh}
  }{
    \conf{\COMPOSE{C_1}{C_2}}{\sk}{\hh}
    \qmstep{p}{\val}
    \conf{\fail}{\sk}{\hh}
  }
\end{align*}
\caption{The inference rules determining the execution relation $\mstep{}{}$ of our operational semantics.
Here, $\smash{c \mstep{p}{\val} c'}$ indicates a step from $c$ to $c'$ with probability $p$ in which allocation value $\val$ has been chosen.
To avoid clutter, we omit $p$ and $\val$  if they equal their default values $p = 1$ and $\val = 0$. 
Hence, $c \mstep{}{} c'$ means $\smash{c \mstep{1}{0} c'}$.
}
\label{fig:op:rules}
\end{figure}%

\paragraph{Markov Decision Processes}
We formalize the expected runtime of programs as expected rewards of MDPs.
While we broadly adhere to~\citet[Chapter 10]{DBLP:books/daglib/0020348}, we consider MDPs with infinitely many states, infinite branching (actions), and (non-discounted) rewards. 
A thorough discussion of such MDPs is found in \cite[Chapter 7]{puterman2005markov}.
Intuitively, an MDP is a transition system that assigns to every state one or more\footnote{due to nondeterminism} probability distributions (distinguished by an action) over successor states. Moreover, whenever we \emph{leave} a state, we collect a reward.%
\begin{definition}[Markov Decision Process]
	A \emph{Markov decision process} $\MDP$ is a tuple 
	\begin{align*} \MDP \eeq \left( \MS, \MA, \MP, \MI, \MR \right)~, \end{align*}
	where 
	$\MS$ is a countable set of \emph{states},
	$\MA$ is a countable set of \emph{actions},
	$\MP\colon \MS \times \MA \times \MS \to [0,1]$ is the \emph{transition probability function}\footnote{i.e., (1) for all states $\ms \in \MS$ and actions $\ma \in \MA$, $\text{dist}(\ms, \ma) = \sum_{\ms' \in \MS} \MP(\ms,\ma,\ms') \in \{0, 1\}$ and (2) for all $\ms \in \MS$ there exists an action $\ma \in \MA$ such that $\text{dist}(\ms, \ma) = 1$. We call the actions $\ma$ with $\text{dist}(\ms, \ma) = 1$ the \emph{enabled actions} of state $\ms$.}, 
	$\MI \in \MS$ is the \emph{initial state}, and 
	$\MR\colon \MS \to \PosRealsInf$ assigns to every state a reward that is collected when leaving a state.\qedtriangle%
\end{definition}%
\noindent{}%
Consider until further notice a fixed MDP $\MDP = \left( \MS, \MA, \MP, \MI, \MR \right)$.
Our goal is to determine the maximal expected reward collected over all possible paths that start in $\MI$.
For that, we first resolve the nondeterminism, arising from multiple actions being enabled, by a \emph{scheduler} $\msched\colon \MS^{+} \to \MA$
	which chooses an action for every history of states. 
	We denote by $\MSCHED$ the set of all schedulers.%

Once a scheduler $\msched$ is fixed, a \emph{path} is a sequence $\ms_0 \ms_{1} \ms_{2} \ldots$ of states starting in $\ms_0 = \MI$ such that there is non-$0$ probability (under the distribution determined by $\msched$) of moving from $\ms_{n}$ to $\ms_{n+1}$.
Formally, the set $\mpaths{n}{\msched}$ of \emph{paths of length $n \in \Nats$ induced by scheduler $\msched$} 
is given by
\begin{align*}
  \mpaths{n}{\msched} \eeq 
   \left\{~
     \ms_0 \ldots \ms_{n-1}
     ~~\middle|~~
     \ms_0 = \MI,~~ \forall i \in \{1, \ldots, n-1\}\colon \MP(\ms_{i-1}, \msched(\ms_0 \ldots \ms_{i-1}), \ms_i) > 0
  ~\right\}~.
\end{align*}
The total probability of taking a path $\ms_0 \ldots \ms_{n-1}$ and the total reward collected along that path are then obtained by multiplying transition probabilities and summing up rewards, that is,  
\begin{align*}
 \MP(\ms_0 \ldots \ms_{n-1}) \eeq \prod_{i=1}^{n-1} \MP(\ms_{i-1}, \msched(\ms_0 \ldots \ms_{i-1}), \ms_i)
 \qand
 \MR(\ms_0 \ldots \ms_{n-1}) \eeq \sum_{i=0}^{n-2} \MR(\ms_{i})~~.
\end{align*}
The \emph{total expected reward} $\ExpRew{\MDP}$ of MDP $\MDP$ is then the maximal\footnote{since the reward is used for modeling the runtime} (over all schedulers $\msched$ and path lengths $n$) 
accumulated reward collected along all paths of length $n$ and induced by scheduler~$\msched$  weighted by each path's probability.
Formally, $\ExpRew{\MDP}$ is given by
\begin{align*}
  	\ExpRew{\MDP} \qeq 
  	\sup_{\msched \in \MSCHED} ~ \sup_{n \in \Nats} ~\sum_{\mpath \in \mpaths{n}{\msched}} \MP(\mpath) \cdot \MR(\mpath)~.
\end{align*}
\paragraph{The Operational Markov Decision Process}
We will now construct an MDP%
\begin{align*} 
	  \MDPOP{C,\sk,\hh} \qeq \left(\, \Confs,~ \Vals,~ \MPOP,~ (\conf{C}{\sk}{\hh}),~ \MROP \,\right) 
\end{align*}%
whose expected reward captures the expected runtime of executing program $C$ on state $\skhh$.
We call $\MDPOP{C,\sk,\hh}$ the \emph{operational MDP} of $C$ and $\skhh$.
This MDP's set of states are the program configurations $\Confs$, its action set are the allocation values $\Vals$, and its initial state is the configuration $(C, \sk, \hh)$.
The transition probability function $\MPOP$ is obtained from the execution relation 
$\mstep{}{}$ in \Cref{fig:op:rules} by accumulating the probability of all steps from $c$ to $c'$ \mbox{which choose the same value $\val$, \ie}%
\begin{align*}
    \MPOP\colon \Confs \times \Vals \times \Confs,
    \qquad
    (c,\val,c') \mapsto \sum_{p\colon c ~\mstep{p}{\val}~ c'} p~.
    \label{eq:mpop}
\end{align*}%
By construction of our execution relation, $\MPOP$ is a well-defined transition probability function:
\begin{enumerate}
	\item for all $c \in \Confs$ and $\val \in \Vals$, we have $\sum_{c' \in \Confs} \MPOP(c,\val,c') \in \{0, 1\}$ \quad and
	\item for all $c \in \Confs$, there exists $\val \in \Vals$ such that $\sum_{c' \in \Confs} \MPOP(c,\val,c') \eeq 1$.
\end{enumerate}
(2) follows from the fact that in $\mstep{}{}$ every configuration has a successor ($\sink \mstep{}{} \sink$ if necessary).

Finally, the reward function $\MROP$ reflects our runtime model, where collected reward corresponds to accumulated runtime: we collect reward $e$ whenever we execute $\TICK{e}$ and reward $\infty$ whenever we encounter a memory fault. 
Executing any other $\hpgcl$ statement consumes no runtime and thus reward 0 is collected. 
Hence, the reward function $\MROP$ is given by
\begin{align*}
  \MROP\colon \Confs \to \PosRealsInf~,
  \qquad\qquad 
  c \mapsto
  \begin{cases}
    \sk(e),        & \qif c = (\conf{\TICK{e}}{\sk}{\hh})~\text{or}~c = (\conf{\TICK{e}\fatsemi C}{\sk}{\hh})\\
    \infty,        & \qif c = (\conf{\fail}{\sk}{\hh}) \\
    0,             & \quad \text{otherwise}.
  \end{cases}
\end{align*}
Put together, we define the \emph{expected runtime} of $\hpgcl$ program $C$ on initial program state $\skhh$ as the expected reward $\ExpRew{\MDPOP{C,\sk,\hh}}$ of the operational MDP of $C$ and $\skhh$.

\section{Runtime Separation Logic}
\label{s:rsl}

We will now study \emph{runtime separation logic} (\RSL), a real-valued \enquote{logic} in the spirit of \cite{QSLpopl}, suitable for use in reasoning about upper bounds on expected runtimes of randomized algorithms that manipulate dynamic data structures.
Its key ingredient are two separating connectives, $\sepadd$ and $\sepmon$, which replace the standard separation logic connectives $\sepcon$ and $\sepimp$.
Though rediscovered independently by us, \RSL has been proposed for future investigation by \citet{christophPhd} and then investigated by \citet{haslbeckPhd}, with almost-exclusive focus on its meta-theory.%

\subsection{Runtimes}
Classical program verification employs \emph{logical predicates} which evaluate to $\true$ or $\false$ for reasoning about program correctness.
Our goal is to reason about a program's \emph{expected runtime}, \ie the average (possibly unbounded) number of time units it takes to execute the program.
To this end, we use genuine \emph{quantities}, which map states to \emph{numbers} instead of truth values.%
\begin{definition}[Runtimes]
	The set of \emph{runtimes} is given by%
	\begin{align*}
		\T \qeq \setcomp{\rta}{\rta \colon \States \to \PosRealsInf}~.
	\end{align*}
	We use metavariables $\rta, \rtb, \rtc$, and variations for runtimes. 
	Together with the order%
	\begin{align*}
		\rta \ppreceq \rtb \qqiff \forall\,(\sk, \hh) \in \States\colon \quad \rta(\sk, \hh) \lleq \rtb(\sk, \hh)
	\end{align*}
	the set of runtimes forms a complete lattice.
	
	We call a runtime $\rta$ \emph{finite} if $\forall\, (\sk,\hh) \in \States\colon \rta(\sk,\hh) < \infty$ and denote this by $f \prec \infty$.
	\qedtriangle
\end{definition}%
\noindent{}%
%
%
%
For any constant $r \in \PosRealsInf$, we simply write $r$ for the \emph{constant runtime} $\lambda (\sk,\hh)\mydot r$.
Similarly, we write $x$ for the runtime $\lambda (\sk,\hh)\mydot \sk(x) $.
We write \textsf{size} for the runtime corresponding to the number of allocated memory locations on the heap, \ie $\textsf{size} = \lambda(\sk,\hh)\mydot|\dom{\hh}|$.

\subsection{Truth vs.~Runtimes}

It is not overly helpful, in general, to think of runtimes as many-valued truth values.
However, if one subscribes to the fairly widespread conception that \emph{truth is something desirable} and \emph{falsehood is something undesirable}, then one could make the following analogy:
In the world of runtimes \mbox{--- usually ---} \emph{a runtime of \,$0$ is something desirable}, \emph{higher numbers are less and less desirable}, and \emph{$\infty$~is something undesirable}.
In that sense, one can well draw a connection between the undesirable $\false$ and $\infty$ on the one hand, and less well a connection between the desirable $\true$ and all finite runtimes.
This analogy is to some extent also reflected in our runtime model of the heap-manipulating \hpgcl constructs:
Memory faults (certainly undesirable!) cause infinite runtime.

\subsection{Gatekeeper Brackets}

We can turn Boolean predicates $\preda\colon \States \to \{\true, \false\}$ into runtimes in a way that preserves the above analogy:
The \emph{gatekeeper bracket} \textbf{$\iiverson{\preda}$} (reminiscent of \emph{Iverson brackets}) of a predicate $\preda$ is defined as the function
\begin{align*}
\iiverson{\preda} \colon\quad \States \To \{0,\, \infty\},\quad \iiverson{\preda}(\sk, \hh) \eeq \begin{cases}
0, & \textnormal{if $(\sk,\hh) \models \preda$}\\
\infty, & \textnormal{if $(\sk,\hh) \not\models \preda$}.
\end{cases}
\end{align*}
$\iiverson{\preda}$ can be understood as a gatekeeper which checks whether the documents we present to them (\ie the current state $(\sk, \hh)$) are in accordance with their internal guidelines (\ie the predicate $\preda$).
In the desirable case that our documents check out (\ie if $(\sk,\hh) \models \varphi$), they will let us pass with zero further delay: $\iiverson{\preda}(\sk,\hh) = 0$.
In the undesirable case that our documents do not check out \mbox{(\ie if $(\sk,\hh) \not\models \varphi$)}, the gatekeeper will hold us up for infinitely long: $\iiverson{\preda}(\sk,\hh) = \infty$.

\subsection{Separation Logic Atoms}
We will often specify memory-safety constraints using the standard (Boolean) separation logic atoms collected in \Cref{fig:sl-predicates} (cf., \cite{DBLP:conf/popl/IshtiaqO01,DBLP:conf/lics/Reynolds02}).
The \emph{empty heap predicate} $\slemp$ specifies that \emph{no} memory location is allocated; 
the predicate $\slvalidpointer{e}$ specifies that exactly one location, namely $e$, is allocated on the heap, 
and the \emph{points-to predicate} $\slsingleton{e}{e'}$ specifies also that precisely location $e$ is allocated on the heap and moreover that it stores content $e'$.
For a separation logic atom $\alpha$, we write $\iiverson{\alpha}$ to obtain its gatekeeper bracket.
\begin{table}[t]
	\renewcommand{\arraystretch}{1.5}
	\begin{adjustbox}{max width=\textwidth}
	\begin{tabular}{l@{\hspace{5em}}l}
	\hline \hline
	$\preda$ & $\skhh \models \preda$ \qiff $\ldots$ \\
	\hline
	$\slemp$ & $\dom{\hh} = \emptyset$ \\
	$\slvalidpointer{e}$ & $\dom{\hh} = \{ \sk(e) \}$ \\
	$\slsingleton{e}{e'}$ & $\dom{\hh} = \{ \sk(e) \}$ \qand $\hh(\sk(e)) = \sk(e')$ \\
	\hline 
	\hline
	\end{tabular}
	\end{adjustbox}
	\caption{Standard separation logic atoms and their semantics.}
	\label{fig:sl-predicates}
\end{table}%

\subsection{Standard Connectives on Runtimes}
\label{sec:std-connectives}

\subsubsection{Addition}
\label{sec:std-connectives-add}

We define standard mathematical operations (addition, multiplication, minimum, etc.) on runtimes pointwise, \eg $\rta + \rtb = \lambda(\sk,\hh)\mydot \rta(\sk,\hh) \pplus \rtb(\sk,\hh)$.
\emph{Addition} ($+$) aggregates undesirableness, namely the runtime $\rta$ \emph{and} the runtime $\rtb$.
In that sense, addition can be thought of as the runtime analogue to \emph{logical conjunction} aka \emph{logical and} ($\wedge$) which also aggregates undesirableness, namely \emph{falsehood}:
If either $a$ or $b$ are $\false$, $a \wedge b$ aggregates this falsehood and becomes itself $\false$.
This interpretation is also compatible when using gatekeepers: $\iiverson{\preda \wedge \psi} = \iiverson{\preda} + \iiverson{\psi}$.

\subsubsection{Monus}
\label{sec:std-connectives-mon}

The (pointwise) \emph{monus} operation $\rtb \monus \rta = \max\{\rtb - \rta, 0\}$ can be read as first carving out the runtime $\rta$ from the runtime $\rtb$ and then measuring only the remaining runtime.
Monus is the adjoint of addition, satisfying $\rtb \preceq \rtc + \rta$ iff $\rtb \monus \rta \preceq \rtc$.
Consequently, monus is the runtime analogue to logical implication ($\rightarrow$) in that $\alpha \monus \beta$ corresponds to $\beta \rightarrow \alpha$.
Keeping in mind that non-zero runtime is undesired, implication also carves out undesirableness, namely falsehood:
\mbox{For $a \rightarrow b$}, whatever falsehood $a$ carries is carved out from $b$.
Indeed, if $a$ is $\false$, then we carve out \emph{all} the falsehood from $b$ (since we are in a Boolean realm).
Thus, there cannot remain any falsehood and $a \rightarrow b$ is $\true$ in this case.
Dually, if $a$ is $\true$, then there is no falsehood to be carved out from $b$.
Thus, there remains only whichever falsehood $b$ already carries and $a \rightarrow b$ is just $b$ in this case.

Compatibility for gatekeepers is given by $\iiverson{\preda \rightarrow \predb} = \iiverson{\preda} \monus \iiverson{\predb}$ when using the convention $\infty \monus \infty = 0$.
Even $\infty \monus \infty = 0$ fits with our intuition of carving out undesirableness: 
$\infty$ is the most undesired, and $\infty \monus \infty$ would thus carve out \emph{all} undesirableness out of the most undesired.
What remains is no undesirableness whatsoever: 0.%

\subsubsection{Minimum}
\label{sec:std-connectives-min}

The \emph{minimum} of two runtimes, denoted $\rta \rtmin \rtb$, is the runtime analogue of logical disjunction ($\vee$). 
Applied to arbitrary runtimes, we can read $\rta \rtmin \rtb$ as a preference for smaller, \ie more desirable, runtimes, thus reflecting that we ultimately wish to reason about as tight as possible upper bounds. 
In particular, we prefer a finite runtime over an infinite one indicating undesired behavior. 
Analogously, $a \vee b$ prefers the more desirable (more true) truth value.
Compatibility for gatekeepers is given by $\iiverson{\preda \vee \predb} = \iiverson{\preda} \rtmin \iiverson{\predb}$.

\subsubsection{Multiplication}
\label{sec:std-connectives-mul}

We typically use runtime multiplication $\rta \cdot \rtb$ in two restricted forms: 
Firstly, we write $p \cdot \rta$ for the runtime $\rta$ scaled by some probability $p \in [0,1]$. 
Throughout this paper, we adapt the convention that $0 \cdot \infty = \infty \cdot 0 = 0$.
Secondly, we write $\iverson{\preda} \cdot \rta$ to specify a \emph{conditional runtime}~$\rta$ that only amounts to $f$ if the predicate $\preda$ holds and otherwise to $0$.
Here, the \emph{Iverson bracket} $\iverson{\preda}$ (defined as $\iverson{\preda}(\sk, \hh) = 1$ if $(\sk,\hh) \models \preda$ and $\iverson{\preda}(\sk, \hh) = 0$ otherwise) acts as a logical guard that \enquote{activates} the runtime $\rta$ if and only if $\preda$ holds. 
Conditional runtimes $\iverson{\preda} \cdot \rta$ can also be expressed with gatekeeper brackets, since %
	$\iverson{\preda} \cdot \rta = \iiverson{\neg \preda} \rtmin \rta = \rta \monus \iiverson{\preda}$. 
%

\subsection{Separating Connectives on Runtimes}

To enable \emph{local reasoning} about expected runtimes of randomized \emph{heap-manipulating} programs, we derive quantitative versions of separation logic's characteristic connectives --- the separating conjunction $\preda \sepcon \psi$ and the separating implication $\preda \sepimp \psi$.
We will obtain them from our runtime analogues for conjunction and implication, namely addition $\rta + \rtb$ and monus $\rtb \monus \rta$.

\subsubsection{Separating Addition}

Recall from \cite{DBLP:conf/popl/IshtiaqO01} that%
\begin{align*}
	(\sk, \hh) \mmodels \preda \sepcon \predb \qqiff
	\exists \, \hh_1, \hh_2 \textnormal{ with } \hh \eeq \hh_1 \sepcon \hh_2\colon \quad (\sk, \hh_1) \mmodels \preda ~\text{ and }~ (\sk, \hh_2) \mmodels \predb~,
\end{align*}
\ie the \emph{separating conjunction} $\preda \sepcon \psi$ is $\true$ for a state $(\sk,\hh)$ if, among all partitionings of the heap~$\hh$ into $\hh_1$ and $\hh_2$, there exists one such that $\preda$ is $\true$ for $(\sk,\hh_1)$ \emph{and} $\psi$ is true for $(\sk,\hh_2)$. 
Notice that the \enquote{and} aggregates \emph{undesirableness} (falsehood, cf.~\Cref{sec:std-connectives-add}), whereas the $\exists$~quantifier, by choosing a heap partitioning, optimizes globally for the \emph{most desirable} outcome (truth). 

Towards connecting runtimes $\rta$ and $\rtb$ in a similar fashion aggregating as little undesirableness as possible, it is natural to replace the falsehood aggregator \enquote{and} by its runtime analogue~$+$.
As for the $\exists$ quantifier which governs the choice of partitioning, it is natural to replace this with a $\min$, so that we aggregate as little runtime as possible.
This leads us to the following definition:%
\begin{definition}[Quantitative Separating Addition~\textnormal{\cite{christophPhd,haslbeckPhd}}]	
\label{def:sepadd}
	The \emph{quantitative separating addition} $\rta \sepadd \rtb$ of two runtimes $\rta,\rtb \in \T$ is defined as%
	\begin{align*}
	\rta \ssepadd \rtb \qeq \lambda (\sk, \hh)\mydot \min_{\hh_1, \hh_2} \setcomp{\rta (\sk, \hh_1) + \rtb(\sk, \hh_2)}{\hh = \hh_1 \sepcon \hh_2}~.\tag*{\qedtriangle}
	\end{align*}
\end{definition}%
\begin{example}
We typically use $\rta \sepadd \rtb$ to cut off parts of the heap (specified by $\rta$) before evaluating the runtime in $\rtb$.
For example, to evaluate $\isingleton{7}{3} \ssepadd \textsf{size}$, we first attempt to cut off the single memory location $7$ such that $\isingleton{7}{3}$ evaluates to $0$ and then measure the number of locations in the remaining heap.
As a truly quantitative example, $\textsf{size} \sepadd \textsf{size} = \textsf{size}$.%
\qedtriangle%
\end{example}

\subsubsection{Separating Monus}

Recall from \cite{DBLP:conf/popl/IshtiaqO01} that%
\begin{align*}
	(\sk, \hh) \mmodels \preda \sepimp \predb \qqiff
    \forall \, \hh' \textnormal{ with } \hh' \disjoint \hh\colon \quad (\sk, \hh') \models \preda ~~\text{ implies }~~ (\sk, \hh \sepcon \hh') \models \predb, 
\end{align*}
\ie the state $(\sk,\hh)$ satisfies the \emph{separating implication} $\preda \sepimp \psi$
iff for every well-defined heap extension $\hh'$ of $\hh$ (\ie $\hh' \disjoint \hh$) specified by $\preda$ (\ie $(\sk,\hh') \models \preda$), the combined state $(\sk,\hh \sepcon \hh')$ satisfies $\psi$.
In other words, $\psi$ must hold for the worst (for satisfying $\psi$) heap extensions admitted by $\preda$. 

As we saw in \Cref{sec:std-connectives-mon}, the \enquote{implies} carves out the \emph{undesirableness} (falsehood) of its left operand from its right one, whereas the $\forall$~quantifier, by considering \emph{every} heap extension, optimizes globally for the \emph{least desirable} outcome (falsehood). 
Towards connecting runtimes $\rta$ and $\rtb$ in a similar fashion carving out as little undesirableness as possible, it is natural to replace \enquote{implies} by its runtime analogue~$\monus$.
As for the $\forall$ quantifier which optimizes for falsehood, it is natural to replace this with a quantitative analogue that also optimizes for most undesirable: $\sup$,\footnote{Not a $\max$ as there are potentially infinitely many extensions $\hh'$.} so that we aggregate as much runtime as possible.
This leads us to the following definition:%
\begin{definition}[Quantitative Separating Monus~\textnormal{\cite{haslbeckPhd}}]	
\label{def:sepmon}
	 The \emph{quantitative separating monus} $\rta \sepmon \rtb$ of two runtimes $\rta, \rtb \in \T$ is defined as%
	\begin{align*}
	\rta \sepmon \rtb \qeq \lambda (\sk, \hh)\mydot \smash{\sup_{\hh'}}~ \setcomp{\rtb(\sk, \hh \sepcon \hh') \monus \rta(\sk,\hh')}{\hh' \disjoint \hh}~,
	\end{align*}
	where $\infty \monus \infty = 0$.\qedtriangle
\end{definition}%
%
%
%
%
%
\begin{example}
	We typically use $\rta \sepmon \rtb$ to extend the heap before evaluating $\rtb$ on the extended heap. 
	For example, to evaluate $\isingleton{7}{3} \sepmon \textsf{size}$, we first extend the heap $\hh$ by $\singleheap{7}{3}$ and then count the number of allocated locations. 
	If location $7$ is not allocated in $\hh$, the result is $|\dom{\hh}|+1$; otherwise, $\isingleton{7}{3}$ is $\infty$ for every heap extension and the overall result is $\textnormal{\enquote{something}} \monus \infty = 0$.%
	\qedtriangle%
\end{example}%

\subsubsection{Properties of \,$\sepadd$ and $\sepmon$}

\citet[Chapter 4]{haslbeckPhd} showed that the separating addition~$\sepadd$ and the separating monus~$\sepmon$ enjoy most desirable properties of the classical separating connectives collected by \citet{DBLP:conf/lics/Reynolds02}.
In particular, $\sepadd$ and $\sepmon$ are \emph{adjoint}, \ie for all runtimes $\rta,\rtb,\rtc \in \T$, 
	\begin{align*}
	   \rtc \ppreceq \rta \sepadd \rtb \qquad \text{iff} \qquad \rtb \sepmon \rtc \ppreceq \rta~.
	\end{align*}%
Adjointness immediately yields the \emph{modus ponens} property: subtracting and adding the same runtime $\rta$ from and to a runtime $\rtb$ overapproximates $\rtb$, \ie $\rtb \ppreceq \rta \sepadd (\rta \sepmon \rtb)$.	
Moreover, $(\T, \, \sepadd,\, \iemp)$ forms a commutative monoid, \ie $\sepadd$ is associative ($\rta \sepadd (\rtb \sepadd \rtc) = (\rta \sepadd \rtb) \sepadd \rtc$), $\iemp$ is the neutral element ($\rta \sepadd \iemp = \iemp \sepadd \rta = \rta$), and $\sepadd$ is commutative ($\rta \sepadd \rtb = \rtb \sepadd \rta$).
Many properties of standard addition naturally carry over to the separating addition. For example, $\sepadd$ is monotone and multiplication distributes over $\sepadd$, \ie $p \cdot (f \sepadd g) = p\cdot f \sepadd p\cdot g$.
We remark that standard addition and separating addition are \emph{not}\footnote{as opposed to our claim in \cite{AERTpopl}} sub-distributive, \ie $\rta + (\rtb \sepadd \rtc) \ssucceq (\rta \sepadd \rtb) \pplus (\rta \sepadd \rtc)$ does \emph{not} hold in general (choose $\rta=1$ and $\rtb=\rtc=0$).
	

\subsection{Runtime Specifications}

\paragraph{Pure runtimes}
A runtime $\rta$ is \emph{pure} if it does not depend on the heap, \ie $\rta(\sk,\hh) = \rta(\sk,\hh')$ for all $\hh,\hh' \in \Heaps$.
Examples of pure runtimes include $5$, $x + y$, and $\iiverson{x = y}$, but not $\textsf{size}$ or $\iemp$.

In a separating addition $\rta \sepadd \rtb$ where $\rta$ is pure, heap portions that would increase $\rtb$ will always be evaluated in $\rta$ (thus not at all), since $\sepadd$ tries to minimize the overall runtime. 
In particular, we have $\rta \preceq \rta \sepadd \rtb$ and $0 \sepadd \rtb  \preceq \rtb$, where $0 \sepadd \rtb$ is the runtime analogue to the smallest intuitionistic extension $\true \sepcon \preda$ of predicate $\preda$ (cf. \cite{DBLP:conf/lics/Reynolds02}).
Notice that $\rta \preceq \rta \sepadd \rtb$ does not hold for arbitrary $\rta$: 
\Eg, $ \textsf{size}  \not\preceq  \textsf{size} \sepadd \ivalidpointer{x} $, since the LHS can become 1 unit bigger than the RHS.

To explicitly prohibit such effects, we denote by $\emprun{\rta}$ the runtime $\rta$ ($\rta$ arbitrary) that is required to be evaluated in the empty heap (otherwise it is $\infty$), \ie we require $\rta$ \emph{and} the empty heap:
\begin{align*}
   \emprun{\rta} ~{}\coloneqq{}~\rta \pplus \iemp~.
\end{align*}%
If, additionally, $\rta$ is pure, then $\rta + \rtb = \emprun{\rta} \sepadd \rtb$ holds for all runtimes $\rtb$.

\begin{example}
	Consider the runtime $\rta$ given by
	\begin{align*}
	  \rta \qeq \isingleton{a}{b} \sepadd \isingleton{b}{c} \sepadd \isingleton{c}{d} \sepadd \emprun{\iiverson{a = d}} \ssepadd \emprun{42}~.
	\end{align*}%
	It evaluates to $42$ for every state whose heap contains a circle $a \mapsto b \mapsto c \mapsto a$ and nothing else; otherwise, it evaluates to $\infty$. We can think of $\rta$ as having two components: the gatekeeper brackets impose safety constraints to avoid undefined behavior (which would lead to $\infty$) and $\emprun{42}$ represents the time units consumed if all safety constraints are met.\qedtriangle
\end{example}
\paragraph{Quantifiers}
In the Boolean case, $\exists$ optimizes for the most desirable (truth), whereas $\forall$ optimizes for the least desirable (falsehood).
In \RSL, smaller runtimes are more desirable than larger ones.
The \RSL analogue to $\exists x\colon \preda$ is thus an \emph{infimum}, denoted by $\Inf x\colon \rta$, which picks a value for $x$ to minimize runtime $\rta$.
The \RSL analogue to $\forall x\colon \preda$ is a \emph{supremum}, denoted by $\Sup x\colon \rta$, which picks a for $x$ to maximize runtime $\rta$. 
To formally define our runtime quantifiers $\Inf$ and $\Sup$, we denote by
\begin{align*}
  \rta\subst{x}{e} \qeq \lambda(\sk,\hh)\mydot \rta(\sk\statesubst{x}{\sk(e)}, \hh)
\end{align*}%
the \enquote{syntactic} replacement of every \enquote{occurrence} of variable $x$ in $\rta$ by expression $e$.
We then define
\begin{align*}
   \Inf x\colon \rta &\qeq \lambda(\sk,\hh)\mydot \inf \setcomp{\rta\subst{x}{v}(\sk,\hh)}{v \in \Vals} \qqand\\
    \Sup x\colon \rta &\qeq \lambda(\sk,\hh)\mydot \sup \setcomp{\rta\subst{x}{v}(\sk,\hh)}{v \in \Vals}~.
\end{align*}%
Further details on these quantifiers are found in \cite{QSLpopl,DBLP:journals/pacmpl/BatzKKM21}.
\noindent
To specify runtimes over data structures of arbitrary sizes, we also define runtime variants of \emph{iterating separating conjunctions} and \emph{inductive predicate definitions} (cf. \cite{DBLP:conf/lics/Reynolds02}).

\paragraph{Separating sums} 

To specify runtimes evaluated in variably-sized contiguous memory blocks, we use the iterative separating addition, called \emph{separating sum} for short, given by
\begin{align*}
  \bigoplus_{i=e}^{e'} \rta \qeq \lambda(\sk,\hh)\mydot
  \begin{cases}
    \left(\rta\subst{i}{\sk(e)} \sepadd \bigoplus_{i=e+1}^{e'} \rta\right) (\sk,\hh) & \qif \sk(e) \leq \sk(e')
    \\
  	\iemp(\sk,\hh) & \qif \sk(e) > \sk(e')~,
  \end{cases}
\end{align*}%
where $\rta \in \T$ and $e,e'$ are arithmetic expressions evaluating to values in $\Nats$.
\begin{example}
Consider the following runtime $\rta$ over the variables $x,y$, and $i$:
\begin{align*}
  \rta \qeq \bigoplus_{i=1}^{y} \left(~\Inf z\colon \isingleton{x + i - 1}{z} \ssepadd \emprun{(z \cdot z)}~\right)~.
\end{align*}%
$\rta$ \emph{specifies} --- read: \emph{is not $\infty$ iff} --- that the heap is an array of length $y$ and evaluates to
the sum of squares of the values stored in the array.
The $\isingleton{x + i - 1}{z}$'s ensure the array structure and we use \RSL's $\Inf$ quantifier to (\enquote{existentially}) refer to each location $x+i-1$'s content $z$.\qedtriangle%
\end{example}

\paragraph{Coinductive runtime definitions}
We specify runtimes of linked data structures using coinductive definitions, \ie \emph{greatest} fixed points $\gfp \Psi$ of recursive runtime equations of the form 
\begin{align*}
  \rta \eeq \Psi(\rta)   \qquad \text{where} \quad  
  \Psi\colon \T \to \T~\text{monotone}~.
\end{align*}%
%
Given a coinductive definition $\rta = \Psi(f)$, we just write $\rta$ to refer to its solution $\gfp \Psi$.
For example, a runtime $\elist{e}{e'}$ specifiying that the heap is a singly-linked list segment from $e$ to $e'$ is given by
\begin{align*}
  \elist{e}{e'} \qeq\quad 
  \underbrace{\pureemp{e=e'}}_{\text{empty list}} 
  \qquad\underbrace{{\rtmin}}_{\mathclap{\substack{\textnormal{\enquote{or}}\\\textnormal{cf.~\Cref{sec:std-connectives-min}}}}}\qquad 
  \underbrace{
    \Inf z \colon \isingleton{e}{z} \sepadd \elist{z}{e'}
  }_{\text{lists of length $\geq$ 1}}.
\end{align*}%
In words, a list specifies either the empty list such that $e = e'$  \emph{or} a non-empty list in which $e$ points to some location $z$ that is the head of a list segment to $e'$.
We can easily extend this definition to obtain the \emph{size} of the list from $e$ to $e'$, or $\infty$ if the heap is not such a list by
\begin{align*}
  \esize{e}{e'} \qeq
  \pureemp{e=e'} 
  \quad\rtmin\quad
  (\emprun{1} \sepadd \Inf z \colon \isingleton{e}{z} \sepadd \esize{z}{e'})~.
\end{align*}%

\section{The Expected Runtime Calculus for $\boldhpgcl$}
\label{s:ert}

\begin{table}[t]
	\renewcommand{\arraystretch}{1.5}
	\begin{adjustbox}{max width=\textwidth}
		\begin{tabular}{@{\hspace{1em}}l@{\hspace{2em}}l@{\hspace{1em}}}
			\hline\hline
			$\boldsymbol{C}$			& $\boldsymbol{\textbf{\textsf{ert}}\,\left \llbracket C\right\rrbracket  \left(\rta \right)}$\\
			\hline\hline
			$\TICK{e}$				& $e + \rta$ \quad or equivalently \quad $\emprun{e} \sepadd \rta$\\
			$\ASSIGN{x}{e}$			& $\rta\subst{x}{e}$ \\
			%
			%
			%
			$\ALLOC{x}{e}$
			& $\Sup\, v \colon \Bigl(\bigoplus_{i=1}^{e}\, \isingleton{v+i-1}{0}\Bigr) \sepmon \rta\subst{x}{v}$ \\
			%
			$\ASSIGNH{x}{e}$			& $\Inf\, v \colon \isingleton{e}{v} \sepadd \Bigl(\isingleton{e}{v} \sepmon \rta\subst{x}{v}   \Bigr)$ \\
			$\HASSIGN{e}{e'}$			& $\ivalidpointer{e} \sepadd \Bigl(\isingleton{e}{e'} \sepmon \rta \Bigr)$ \\	
			$\FREE{e}$			& $\ivalidpointer{e} \sepadd \rta$ \\	
			$\PCHOICE{C_1}{p}{C_2}$		& $p \cdot \ert{C_1}{\rta} \pplus (1- p) \cdot \ert{C_2}{\rta}$ \\
			$\ITE{\varphi}{C_1}{C_2}$		& $\iverson{\varphi} \cdot \ert{C_1}{\rta} \pplus\iverson{\neg\varphi} \cdot \ert{C_2}{\rta}$ \\
			$\COMPOSE{C_1}{C_2}$		& $\ert{C_1}{\ert{C_2}{\rta}}$ \\
			$\WHILEDO{\varphi}{C'}$		& $\lfp \rtb\mydot ~ \iverson{\neg \varphi} \cdot \rta \pplus \iverson{\varphi} \cdot \ert{C'}{\rtb}$\\[1ex]
			\hline\hline
		\end{tabular}
	\end{adjustbox}
	\renewcommand{\arraystretch}{1}
	\vspace{1ex}
	\caption{Rules for the $\ertsymbol$--transformer. Here $v$ is a fresh variable not occuring in $e$ or $\rta$.
	}	
	\label{table:ert}
\end{table}%
We now extend the expected runtime calculus of \citet{kaminski2018weakest} by \RSL, thus enabling capabilities for local reasoning about expected runtimes of programs that access and mutate dynamic memory.
This is inspired by the quantitative separation logic of \citet{QSLpopl}.

The backward-moving \emph{expected runtime transformer}
%
	%
	\begin{align*}
	\ertsymbol \colon \hpgcl \to (\T \to \T)
	\end{align*}
	%
%
is defined by induction on $\hpgcl$ in \Cref{table:ert}. The transformer is defined in such a way that 
\begin{align*}
\ert{C}{0}(\sk,\hh)
 \eeq
\begin{cases}
   \text{``expected number of ticks} &  \\[-.3em]
   \text{\quad when executing $C$ on $(\sk,\hh)$"}~, & \text{if $C$ memory-safe on $(\sk,\hh)$} \\
   \infty~, &\text{if $C$ not memory-safe on $(\sk,\hh)$.}
\end{cases}
\end{align*}
--- a fact we will make formally precise w.r.t. our operational MDP semantics in \Cref{s:op}. More generally, to achieve compositionality, $\ert{C}{\rta}(\sk,\hh)$ is a runtime that gives us the expected number of ticks  it takes to first execute the program $C$ on $(\sk,\hh)$ and then let time $\rta$ pass, or $\infty$ in case $C$ is not memory safe on $(\sk,\hh)$. We refer to $\rta$ as the \emph{postruntime}. Let us go over the rules. 

\paragraph{Time consumption} 
How long does it take to execute $\TICK{e}$ and then let time $\rta$ pass? $e + \rta$.
Since $e$ is always pure, this is equivalent to $\emprun{e} \sepadd \rta$, which can be more handy for local reasoning.

\paragraph{Assignment, Sequential composition, Conditional and probabilistic choice, While loop}

All these cases have been treated in detail in \cite[Section 7.3, p.163--166]{Kam19}.
The only difference is that in \cite{Kam19} each basic instruction consumes 1 unit of time whereas we have here (by means of $\TICKsymbol$) a more fine-grained runtime model.

In order to be somewhat self-contained, however, let us quickly go over the case for assignments and probabilistic choice:
How long does it take to execute $\ASSIGN{x}{e}$ (on initial state $(\sk, \hh)$) and then let time $\rta$ pass?
Executing $\ASSIGN{x}{e}$ itself takes no time.
But we need to evaluate $\rta$ in the state that is reached \emph{after} the assignment, \ie the final state $(\sk\statesubst{x}{\sk(e)}, \hh)$.
This is precisely $\rta\subst{x}{e}$ but evaluated in the initial state $(\sk, \hh)$.
For the probabilistic choice $\PCHOICE{C_1}{p}{C_2}$, we simply take the weighted average of the expected time it takes to either execute $C_1$ or $C_2$ and then let time $\rta$ pass.%

\paragraph{Allocation} 

Again, $\ALLOC{x}{e}$ itself takes no time, but we need to measure $\rta$ in a state where the heap has been extended by $e$ contiguous memory locations, all initialized to store value 0.
$\bigoplus_{i=1}^{e}\, \isingleton{v+i-1}{0}$ describes precisely such an extension. By $\sepmon$ we impose this extension on $\rta$.
What is left is to handle the nondeterminism arising from the memory allocator's choice of the first new location $\val$.
As we do upper-bound (worst-case) reasoning, we resolve this nondeterminism via a maximizing $\Sup$ and measure $\rta$ in the extended heap and in a stack where $x$ has been updated to $\val$.%

\paragraph{Lookup} For $\ASSIGNH{x}{e}$, we first ensure via $\isingleton{e}{v} \sepadd \cloze{f}$ that $e$ is indeed allocated.
(If not, the whole term becomes $\infty$, indicating a memory fault.)
In connection with the $\Inf$ quantifier, we moreover select the value $\val$ that $e$ points to.
The $\sepadd$ has now carved the memory location $\singleheap{e}{\val}$ out from the heap.
Since we did not want to manipulate the heap, we reinsert $\singleheap{e}{\val} $ via $\isingleton{e}{\val} \sepmon \cloze{f}$.
What is left is to measure $\rta$ but in a stack where $x$ has been \mbox{updated to $\val$}.

\paragraph{Mutation}
For $\HASSIGN{e}{e'}$, we first ensure via $\ivalidpointer{e} \sepadd \cloze{f}$ that $e$ is actually allocated, but care not about its stored value since we are about to overwrite it.
The $\sepadd$ has now carved location $e$ out from the heap.
In order to overwrite the value stored at $e$ with $e'$, we insert $\singleheap{e}{e'}$ into the heap via $\isingleton{e}{e'} \sepmon \cloze{f}$.
This insertion of location $e$ cannot fail (become $\infty$) because we have previously carved out precisely location $e$ (unless $e$ was not allocated in the first place, in which case the whole term becomes $\infty$ anyway).
What is left is to measure $\rta$ in the so-manipulated heap.

\paragraph{Deallocation}

For $\FREE{e}$, we need to measure $\rta$ in a heap where the location $e$ has been carved out.
As demonstrated numerous times previously, such carving out is achieved by $\ivalidpointer{e} \sepmon \cloze{f}$.%

%
%
\begin{theorem}[Healthiness Properties of $\ertsymbol$]
	\label{thm:ert-health}
	Let $C$ be a program; $F = \{\rta_1 \preceq \rta_2 \preceq \ldots\}$ be an $\omega$-chain of runtimes; $\rta, \rtb$ be runtimes; and $\rtc \in \T$ be a constant runtime.
	Then the following hold:
	\begin{enumerate}
		\item \emph{$\omega$-continuity}: \qquad $\ert{C}{\sup F} \eeq \sup \ert{C}{F}$\vspace{.5em}
		
		\item \emph{Monotonicity:} \qquad $\rta \ppreceq \rtb$ \qimplies $\ert{C}{\rta} \ppreceq \ert{C}{\rtb}$ \vspace{.5em}
		
		\item \emph{Sub-additivity:}\qquad $\ert{C}{\rta + \rtb}
		   \ppreceq
		   \ert{C}{\rta} + \ert{C}{\rtb}$ \vspace{.5em}

		 \item \emph{Constant propagation:} \qquad
		 %
		   $ \ert{C}{\rtc + \rta} \ppreceq \rtc + \ert{C}{\rta}$
	\end{enumerate}
	
\end{theorem}%
%
%
%
%
\begin{remark}[(Non-)\,$\omega$-continuity of $\ertsymbol$]
	$\omega$-continuity is actually somewhat \emph{unexpected}.
	The weakest preexpectation transformer of \citet{QSLpopl} who deals with determining (minimal) expected values of essentially the same $\hpgcl$ is indeed \emph{not} $\omega$-continuous due to the unbounded nondeterminism arising from memory allocation.
	Probably guided by this result, \citet[Section~4.2, p.46]{haslbeckPhd} \emph{claims} that the $\ertsymbol$ transformer is also \emph{not} $\omega$-continuous.
	However, since $\ertsymbol$ is a \emph{maximizer} and suprema commute, $\ertsymbol$ \emph{does} enjoy the beneficial property of $\omega$-continuity.
\end{remark}%
%
%
%


\needspace{2\baselineskip}
\subsection{Local Reasoning for Expected Runtimes}

To enable local reasoning, the $\ertsymbol$ calculus features a \emph{frame rule for establishing upper bounds}:%
\begin{theorem}[Frame Rule for \RSL \textnormal{\citet{haslbeckPhd}}]
	\label{thm:ert_frame}
	For every $C \in \hpgcl$ and runtimes $\rta,\rtb$,%
	\begin{align*}
	     \modC{C}\cap\Vars(\rtb)\eeq\emptyset
	     \qqimplies 
	     \ert{C}{\rta \sepadd \rtb} \ppreceq \ert{C}{\rta} \ssepadd \rtb ~.
	\end{align*}
\end{theorem}%
%
%
\noindent%
We call $\rtb$ in the above theorem the \emph{frame}. Combining the above frame rule with the \RSL analogues of \citet{DBLP:conf/lics/Reynolds02}'s local rules for heap mutation, lookup, memory allocation, and auxiliary variable elimination enables \SL-style source-code level proofs for upper-bounding expected runtimes:%
%
%
\begin{theorem}[Local Rules for \RSL following \textnormal{~\citet{DBLP:conf/lics/Reynolds02}}]
	\label{thm:local_rules}
	Let $C \in \hpgcl$. Then:
	
	\begin{enumerate}
		\item \textnormal{(mut):} \qquad $\ert{\HASSIGN{e}{e'}}{\isingleton{e}{e'}} \ppreceq \ivalidpointer{e}$
		\item \textnormal{(lkp):} \qquad 
		%
		$\ert{\ASSIGNH{x}{e}}{\pureemp{x=z} \sepadd \isingleton{e\subst{x}{y}}{z}} \ppreceq \pureemp{x=y} \sepadd \isingleton{e}{z}$,
		%
		\item \textnormal{(alc):} if $x$ does not occur in $e$, then
		\[\ert{\ALLOC{x}{e}}{\bigoplus_{i=1}^{e} \isingleton{x+i-1}{0}} \ppreceq \iemp~.
		\]
		\item \textnormal{(aux):} For all $\rta, \rtb \in \T$ and all $y \in \Vars$ not occurring in $C$, 
		\begin{align*}
		\ert{C}{\rta} \ppreceq \rtb \qqimplies \ert{C}{\Inf y\colon \rta} \ppreceq \Inf y \colon \rtb~.
		\end{align*}
	\end{enumerate}
\end{theorem}
%

\subsection{Invariant-based Reasoning for Loops}

Recall that $\ert{\mathit{loop}}{\rta}$ for $\mathit{loop}=\WHILEDO{\varphi}{\mathit{body}}$ is defined as 
\begin{align*}
\lfp \rtb\mydot ~ \underbrace{\iverson{\neg \varphi} \cdot \rta \pplus \iverson{\varphi} \cdot \ert{\mathit{body}}{\rtb}}_{ {}\coloneqq \charfun{\rta}(\rtb)}~,
\end{align*}
where we call $\charfun{\rta}$ the \emph{$\ertsymbol$-characteristic functional} of $\mathit{loop}$ w.r.t.\ postruntime $f$. 
For upper-bounding expected runtimes of loops --- given as least fixed points --- we have an invariant-based proof rule:%
\begin{theorem}[Park Induction for \RSL]
	\label{thm:ert_invariants}
	Let $\mathit{loop} = \WHILEDO{\guard}{\mathit{body}}$ and $\inv,\rta\in \T$. 
	Then%
	\begin{align*}
	   \charfun{\rta}(I) \ppreceq I \qimplies
	   \ert{\mathit{loop}}{\rta} \ppreceq I~.
	\end{align*}%
\end{theorem}%
%
%
\noindent%
We call such $\inv$ an \emph{$\ertsymbol$-invariant}. 
Park induction can simplify the verification of loop runtimes significantly: 
to obtain an upper bound on the expected runtime of the \emph{entire} loop, it essentially suffices to upper-bound $\ert{\mathit{body}}{\inv}$, \ie the expected runtime of \emph{one} (arbitrary) loop \emph{iteration}. 
Notice that the above proof rule is \emph{complete} since $\lfp \rtb\mydot \charfun{\rta}(\rtb)$ is necessarily an $\ertsymbol$-invariant.%

\subsection{Example: The Lagging List Traversal}
\label{label:ex_lagging_list}
We demonstrate the applicability of the $\ertsymbol$ calculus. Consider the program $\listtraverse{\varlistelem}$:
\begin{align*}
& \WHILE{\varlistelem \neq 0} 
%
\TICK{1}\fatsemi 
\PCHOICE{\ASSIGNH{\varlistelem}{\varlistelem}}{\sfrac{1}{2}}{\SKIP}
%
 ~\}
\end{align*}
$\listtraverse{\varlistelem}$ traverses a null-terminated list segment beginning at $\varlistelem$. Every iteration costs $1$ unit of time. 
The program then either traverses the next edge of the list (left branch), or forgets to do so (right branch), each with probability $\sfrac{1}{2}$. 
Using the frame rule, the local rules for lookup and auxiliary variable elimination, and monotonicity of $\ertsymbol$ we prove that
\[
  \inv \ddefeq 2 \cdot \Inf \lh \colon \elist{\lh}{\varlistelem} \sepadd \esize{\varlistelem}{0}
\]
is an $\ertsymbol$-invariant of $\listtraverse{\varlistelem}$ w.r.t.\ postruntime $0$, which, by \Cref{thm:ert_invariants}, implies that $\inv$ upper-bounds the expected runtime of $\listtraverse{\varlistelem}$, where $\elist{\lh}{\varlistelem}$ is only needed for strengthening the loop invariant.
Hence, when executed on a heap consisting of a null-terminated list containing the element $\varlistelem$, the program $\listtraverse{\varlistelem}$ is \emph{memory-safe} and takes, in expectation, a runtime of at most $2$ times the size of the list segment beginning at $\varlistelem$. 

\subsection{Soundness of the $\boldertsymbol$ calculus}
\label{s:op}

We will now show that the $\ertsymbol$ calculus is sound in that it characterizes a program's expected execution time defined in terms of expected rewards of operational MDPs in \Cref{s:op-semantics}.
Formally, we show that, for every $\hpgcl$ program $C$ and initial state $(\sk,\hh)$, we have 
\begin{align*}
  \ert{C}{0}(\sk,\hh) \qeq \ExpRew{\MDPOP{C,\sk,\hh}}~,
\end{align*}
where the post-runtime $0$ indicates that no time is consumed after termination of $C$.

We will prove a more general claim for arbitrary post-runtimes $\rta \in \T$ 
instead of the fixed post-runtime $0$. 
To account for $\rta$ in our operational semantics, we first extend the reward function $\MROP$ of our operational MDPs $\MDPOP{C,\sk,\hh}$  such that we collect a reward of $\rta(\sk',\hh')$ whenever $C$ successfully terminates as indicated by an execution step from configuration $(\conf{\term}{\sk'}{\hh'})$ to $\sink$:
\begin{align*}
  \MROP\colon \quad \Confs \to \PosRealsInf,
  \qquad
  c \mapsto
  \begin{cases}
    \sk(e),        & \qif c = (\conf{\TICK{e}}{\sk}{\hh}) \\
    \rta(\sk,\hh), & \qif c = (\conf{\term}{\sk}{\hh}) \\
    \infty,        & \qif c = (\conf{\fail}{\sk}{\hh}) \\
    0,             & \quad \text{otherwise}.
  \end{cases}
\end{align*}
We denote by $\MDPOP{C,\rta,\sk,\hh}$ the operational MDP introduced in \Cref{s:op-semantics} but with the above reward function.
Our expected runtime calculus is then sound in the following sense:%
\begin{theorem}[Soundness of $\ertsymbol$]\label{thm:ert-sound-op}
  For all $C \in \hpgcl$, runtimes $\rta \in \T$, and program states $(\sk,\hh)$,
  \begin{align*}
    \ert{C}{\rta}(\sk,\hh) \qeq 
    \ExpRew{\MDPOP{C,\rta,\sk,\hh}}~.
  \end{align*}
\end{theorem}%
%
\noindent%
Our soundness theorem clarifies what the expected runtime calculus actually computes: 
as long as \emph{no} program execution of $C$ on initial state $(\sk,\hh)$ leads to a memory fault, $\ert{C}{0}(\sk,\hh)$ is the expected execution time measured in units of time consumed by $\TICK{e}$ statements; 
if \emph{some} program execution of $C$ on initial state $(\sk,\hh)$ \emph{does} lead to a memory fault, we cannot give a finite bound on the expected runtime and have to conclude $\ert{C}{0}(\sk,\hh) = \infty$.
Towards a proof of our soundness theorem, a \emph{runtime transformer} is a function of the  form
\begin{align*}
	    \rt\colon \hpgcl \cup \{ \term, \fail \} \to (\T \to \T)
\end{align*}
that maps the terminated programs to the postruntime, \ie $\rtt{\term}{\rta} = \rta$, and memory faults to an infinite runtime, \ie $\rtt{\fail}{\rta} = \infty$.

We then lift the partial ordering $\preceq$ on runtimes in $\T$ to an order on runtime transformers, \ie
	\begin{align*}
	    \rt \lleq \rt' \qqiff \forall\, C \in \hpgcl \cup \{ \term, \fail \} ~ \forall\, \rta \in \T\colon \quad \rtt{C}{\rta} \ppreceq \rt'\llbracket C \rrbracket(\rta)~.
	\end{align*}
Clearly, we can naturally extend the expected runtime calculus $\ertsymbol\colon \hpgcl \to (\T \to \T)$ to a runtime transformer. We denote the resulting \emph{extended expected runtime calculus} by $\eert$, \ie
 \begin{align*}
   \eert \qeq \lambda C ~ \lambda \rta\mydot~
   \begin{cases}
     \rta, & \qif C \eeq \term \\
     \infty, & \qif C \eeq \fail \\
     \ert{C}{\rta}, & \quad\text{otherwise}. 	
   \end{cases}
 \end{align*}  	
Analogously, our operational semantics induces a runtime transformer $\oprt$  that maps every program, runtime, and state to the expected reward of the corresponding operational MDP, \ie
 \begin{align*}
   \oprt \qeq \lambda C ~ \lambda \rta ~ \lambda (\sk,\hh) \mydot ~ \ExpRew{\MDPOP{C,\rta,\sk,\hh}}~.
 \end{align*} 
 To prove our soundness theorem, it then suffices to show $\eert \leq \oprt$ and $\oprt \leq \eert$.
 
 We will leverage that our runtime transformers $\eert$ and $\oprt$ satisfy the well-established optimality equations for MDPs, also known as Bellman equations (cf. \cite{puterman2005markov}).
Formally, a runtime transformer $\rt$ is \emph{Bellman compliant} if and only if for all $C \in \hpgcl$, $\rta \in \T$, and $(\sk,\hh)$,\footnote{Textbooks typically restrict the supremum to the set of actions that are \emph{enabled} in the given configuration. To simplify notation, we take the supremum over all actions $\MA = \Vals$ and agree on the convention that $\sum_{\emptyset} \ldots = 0$. }
 	    \begin{align*}
 	        \rtt{C}{\rta}(\sk,\hh) \qeq \MROP(\conf{C}{\sk}{\hh}) \pplus \sup_{\ma \in \MA} \sum_{(\conf{C}{\sk}{\hh}) \mstep{p}{\ma} (\conf{C'}{\sk'}{\hh'})} p \cdot \rtt{C'}{\rta}(\sk',\hh')~.
 	    \end{align*}
In words, the runtime computed by a Bellman compliant runtime transformer is the reward collected for leaving the current configuration $(\conf{C}{\sk}{\hh})$ plus
the runtimes of all direct successor configurations $(\conf{C'}{\sk'}{\hh'})$ weighted by the probability moving to configuration $(\conf{C'}{\sk'}{\hh'})$; if there is a nondeterministic choice between different actions (\ie we can choose a distribution over successor configurations), then we take one that maximizes the overall runtime.

A long established result on MDPs (even those with countable state spaces and actions) is that their expected rewards satisfy the Bellman equations~\cite[Theorem 7.1.3]{puterman2005markov}. Hence:%
\begin{lemma}\label{thm:op:bellman}
	$\oprt$ is Bellman compliant.
\end{lemma}%
\noindent%
In fact, the same holds for the (extended) expected runtime calculus:
\begin{lemma}\label{thm:op:eert-bellman}
  $\eert$ is Bellman compliant.
\end{lemma}%
\begin{proof}
  By structural induction on the rules of our execution relation (\Cref{fig:op:rules}).
\end{proof}%
\noindent%
Equipped with two Bellman compliant runtime transformers, we now prove that $\eert \leq \oprt$ and $\oprt \leq \eert$. The first inequality can be proven directly by structural induction: 
\begin{lemma}\label{thm:op:eert-leq-oprt}
  $\eert \lleq \oprt$.
\end{lemma}

\begin{proof}
By construction, we have 
\begin{align*} \eert\llbracket\term\rrbracket(\rta) \eeq \rta \ppreceq \oprt\llbracket\term\rrbracket(\rta) \qand
  \eert\llbracket\fail\rrbracket(\rta) \eeq \infty \ppreceq \oprt\llbracket\fail\rrbracket(\rta)~.
\end{align*}
It then suffices to show that by induction on the structure of $\hpgcl$ programs that, for all $C \in \hpgcl$ and all $\rta \in \T$, we have $\eertt{C}{\rta} \preceq \oprtt{C}{\rta}$.
\end{proof}%
\noindent%
We do not directly show the converse direction, \ie $\oprt \leq \eert$, since we can invoke a more general result for MDPs that goes back to \citet{blackwell1967positive}:
\begin{lemma}\label{thm:op:oprt-least}
  For every Bellman compliant runtime transformer $\rt$, we have $\rt \geq \oprt$.
\end{lemma}

\begin{proof}
    Since all runtimes in $\T$ are non-negative and our execution relation
    never gets stuck, the MDPs induced by our operational semantics are
    positive models (cf. beginning of \cite[Chapter 7.2]{puterman2005markov}).
    The claim is then a special case of \cite[Theorem 2]{blackwell1967positive}.   
\end{proof}%
\noindent%
Notice that a variant of Blackwell's theorem is also found in the textbook of \citet[Theorem 7.2.2]{puterman2005markov}. However, Puterman considers only positive \emph{bounded} models, even though the proof of his theorem does not seem to rely on having a bounded model. Finally, we conclude:%
 %
%
\begin{proof}[Proof of \Cref{thm:ert-sound-op}]
  By \Cref{thm:op:eert-bellman}, $\eert$ is a Bellman compliant runtime transformer. By \Cref{thm:op:oprt-least}, this implies $\oprt \leq \eert$.
  Moreover, by \Cref{thm:op:eert-leq-oprt}, we have $\eert \leq \oprt$.
  Hence, $\eert = \oprt$.
  Now, for any $\hpgcl$ program $C$, runtime $\rta \in \T$, and state $(\sk,\hh)$,  we have
   \begin{align*}
  	\ert{C}{\rta}(\sk,\hh) 
  	\eeq & \eert\llbracket{C}\rrbracket(\rta)(\sk,\hh) \tag{$C \in \hpgcl$, by construction of $\eert$} \\
  	\eeq & \oprt\llbracket{C}\rrbracket(\rta)(\sk,\hh) \tag{$\eert = \oprt$ as shown above} \\
  	\eeq & \ExpRew{\MDPOP{C,\rta,\sk,\hh}}. \tag*{(definition of $\oprt$)\quad\qedhere}
  \end{align*}
\end{proof}

\section{The Amortized Expected Runtime Calculus}
\label{s:aert}

In \emph{amortized analysis}, instead of analyzing the worst-case runtime of $C$, we average the runtime of $C$ over a whole sequence $C^n = \COMPOSE{\COMPOSE{C_1}{\ldots}}{C_n}$ of $n$ consecutive executions of $C$.
One technique for amortized analysis is the \emph{potential method}, whose core idea is to \emph{make frequently occurring low \enquote{\underline{normal-case}} runtimes of $C$ \underline{mildl}y\underline{ lar}g\underline{er}} and in return be able to \emph{make the seldomly occurring \underline{worst-case} runtimes \underline{a lot smaller}}, thus smoothing out seldomly occurring runtime-peaks in the sequence $C^n$.
Key ingredient to achieve such smoothing is a potential function:%
\begin{definition}[Potential Functions~\textnormal{\cite{tarjan1985amortized,sleator1985self}}]
	A \emph{potential function} is a function $\potent$ of type $\States \to \PosReals$.
	Note that $\potent \in \T$ with $\potent \prec \infty$.\qedtriangle
\end{definition}%
\noindent{}%
The potential function needs to be chosen so that each time $C$ has small runtime, the potential is mildly increased.
Each time $C$ has large runtime, on the other hand, the potential should be drastically decreased.
The \emph{amortized runtime} of $C$ is then $C$'s actual runtime plus \emph{the change in potential}.
Indeed, then the amortized runtime of $C$ in the cheap case is $C$'s small runtime plus a mildly positive change in potential --- overall still a small number.
The amortized runtime of the expensive case, on the other hand, is $C$'s large runtime plus a large \emph{negative} change in potential --- overall, again (hopefully), a small number.
Why did the potential do the trick?

Let us denote the runtime of executing $C_i$ by $\mathit{rt}_i$ and the potential attained afterwards by $\potent_i$. 
The amortized runtime of executing sequence element $C_{i+1}$ is then $\mathit{rt}_{i+1} + \potent_{i+1} - \potent_{i}$.
Summing the amortized runtimes over the whole sequence gives%
\begin{align*}
	\left . \sum_{i=0}^n \middle( \mathit{rt}_{i+1} + \potent_{i+1} - \potent_{i} \right)
	\qqeq
	\underbrace{\left. \sum_{i=0}^n \middle( \mathit{rt}_{i+1}\right)}_{\mathclap{{\textnormal{actual runtime of the sequence}}}}{}  \pplus \overbrace{\potent_{n}}^{\mathclap{{\textnormal{non-negative}}}} \mminus \underbrace{\potent_{0}}_{\mathclap{{\textnormal{assumed to be 0}}}}~. \tag{$\dagger$}
\end{align*}%
It is now easy to see that if we start with initial potential $\potent_0 = 0$, then the amortized runtime of the whole sequence \emph{overapproximates} the actual runtime of the sequence.
Indeed, we could recover the actual runtime by subtracting from the amortized runtime the (non-negative) final potential $\potent_n$.

As we are concerned with \emph{expected} (amortized) runtimes of randomized algorithms, we would need to take expected changes in potential into account and telescoping is not as obvious anymore.
Moreover, changes in potential may become negative.
In fact: they should!
Otherwise, we have no chance of compensating for expensive operations.
In the following, we present an $\ertsymbol$-style calculus that can capture \emph{amortized expected runtimes} and we prove that it essentially satisfies the above telescoping property, so that it also over-approximates true \emph{expected} runtimes.

\needspace{2\baselineskip}
\subsection{The Calculus}
\begin{table}[t]
	\renewcommand{\arraystretch}{1.5}
	\begin{adjustbox}{max width=\textwidth}
		\begin{tabular}{@{\hspace{1em}}l@{\hspace{2em}}l@{\hspace{1em}}}
			\hline\hline
			$\boldsymbol{C}$			& $\boldsymbol{\textbf{\textsf{aert}}_\potent\,\left \llbracket C\right\rrbracket  \left(\arta \right)}$\\
			\hline\hline
			$\TICK{e}$ 						& $e + \arta $ \\
			$C$ is atomic and not $\TICKsymbol$ 	& $\ert{C}{\arta+\potent} - \potent$ \\
			$\COMPOSE{C_1}{C_2}$				& $\aert{\potent}{C_1}{\aert{\potent}{C_2}{\arta}}$ \\
			$\ITE{\varphi}{C_1}{C_2}$				& $\iverson{\varphi} \cdot \aert{\potent}{C_1}{\arta} \pplus \iverson{\neg\varphi} \cdot \aert{\potent}{C_2}{\arta}$ \\
			$\PCHOICE{C_1}{p}{C_2}$			& $p \cdot \aert{\potent}{C_1}{\arta} \pplus (1- p) \cdot \aert{\potent}{C_2}{\arta}$ \\
			$\WHILEDO{\varphi}{C'}$				& $\lfp \artb\mydot ~ \iverson{\neg \varphi} \cdot \arta \pplus \iverson{\varphi} \cdot \aert{\potent}{C'}{\artb}$\\[1ex]
			\hline\hline
		\end{tabular}
	\end{adjustbox}
	\renewcommand{\arraystretch}{1}
	\vspace{1ex}
	\caption{Rules for the $\aertsymbol$--transformer w.r.t.\ potential $\potent$.}	
	\label{table:aert}
\end{table}

Let us fix a potential function $\pi$ and define a set of \emph{amortized runtimes} relative to $\pi$.%
\begin{definition}[Amortized Runtimes]
	The set $\Api$ of amortized runtimes with respect to potential function $\potent$, $\potent$-runtimes for short, is defined as 
	\begin{align*}
	   \Api \eeq \setcomp{\arta \colon \States \to \Reals \cup \{\infty\} }{\forall (\sk,\hh) \in \States \colon {-}\potent(\sk,\hh) \leq \arta(\sk,\hh)}~.
	\end{align*}%
	We denote amortized runtimes by $\arta, \artb,\artc$ and variations. 
	We extend $\preceq$ from $\E$ to $\Api$ naturally by
	\begin{align*}
		\arta \ppreceq \artb \qquad\text{iff}\qquad 
		\text{for all $(\sk,\hh)\in\States$}\colon\quad \arta(\sk,\hh) \lleq \artb(\sk,\hh)~.
	\end{align*}%
	$(\Api,\, {\preceq})$ forms a complete lattice with least element $-\potent$.\qedtriangle
\end{definition}%
\noindent{}%
The backward-moving \emph{amortized expected runtime transformer}%
%
	%
	\begin{align*}
	\aertsymbol_\potent \colon \hpgcl \to (\Api \to \Api)
	\end{align*}%
	%
%
is defined by induction on $\hpgcl$ in \Cref{table:aert} and manipulates $\pi$-runtimes, which can in principle become negative (as negative as ${-}\potent$), instead of ordinary runtimes.
Like $\ertsymbol$, the $\aertsymbol$ transformer is defined in such a way that%
\begin{align*}
	\aert{\potent}{C}{0}(\sk,\hh) \eeq \infty \quad\textnormal{if}\quad C \textnormal{ is not memory-safe on } (\sk,\hh)~.
\end{align*}%
Let us briefly go over the rules defining $\aertsymbol$.

\paragraph{Time consumption} 
Executing $\TICK{e}$ and then letting (amortized) time $\arta$ pass takes $e + \arta$ units of time and cannot change the potential.
We cannot go for $\emprun{e} \sepadd \arta$ because $\sepadd$ is undefined on $\Api$.

\paragraph{All other atoms.}

We have $\aert{\potent}{\mathit{atom}}{\arta} = \ert{\mathit{atom}}{\arta + \potent} - \potent$.
\Eg, for $\ASSIGN{x}{e}$, this gives $\arta\subst{x}{e} + \potent\subst{x}{e} - \potent$.
Here, we can see how the change in potential is propagated through the program at the level of atoms.

\paragraph{Composite constructs} 

Defined like $\ertsymbol$, but we need to compose the terms of $\aertsymbol$ components.

\begin{theorem}
	\label{thm:aert_cont}
	$\aertsymbol_\potent$ is $\omega$-continuous, \ie for all $C \in \hpgcl$ and $\omega$-chains  $F = \{\arta_1 \preceq \arta_2 \preceq {\ldots}\}$, 
		\begin{align*}
		\aert{\potent}{C}{\sup F} \eeq \sup \aert{\potent}{C}{F}~.
		\end{align*}%
\end{theorem}%
\noindent{}%
Continuity of the $\aertsymbol$ transformer follows from continuity of the $\ertsymbol$ transformer and the following central theorem, which formalizes the telescoping principle of Equation $(\dagger)$ for \emph{expected} runtimes:%
\begin{theorem}[Telescoping for $\aertsymbol$]
	\label{thm:aert_decomp}
	For all $C\in\hpgcl$ and $\arta\in\Api$, we have%
	\begin{align*}%
		\aert{\potent}{C}{\arta} \eeq \ert{C}{\arta + \potent} - \potent~.
	\end{align*}
\end{theorem}%
\noindent{}%
In some sense, the above theorem tells us that the following \emph{almost} (modulo some technical particularities of the $\texttt{alloc}$ statement which we omit here) holds:%
\begin{align*}
	\aert{\potent}{C}{0} \morespace{\approx} \ert{C}{0} \pplus \textnormal{\enquote{expected potential after executing $C$}} - \potent \tag{$\ddagger$}
\end{align*}%
The reason is that $\ert{C}{\potent}$ can (again: almost) be decomposed into $\ert{C}{0}$ plus the expected value of $\potent$ after executing $C$.

If $C$ is an entire sequence of operations, we can now relate the right-hand-sides of $(\dagger)$ and $(\ddagger)$: 
the \enquote{actual runtime} corresponds to $\ert{C}{0}$, the expected potential after executing $C$ corresponds to $\potent_n$, and the initial potential $\potent_0$ corresponds to $\potent$.
Again, this explanation breaks slightly for programs featuring dynamic memory allocation, but \Cref{thm:aert_decomp} holds also for those programs and allows us to prove soundness of the $\aertsymbol$ transformer in the following \emph{overapproximating} sense:%
\begin{theorem}[Soundness of $\aertsymbol$]
	\label{thm:aert_sound}
	For every $C\in\hpgcl$, we have
	\begin{align*}
	   \ert{C}{0} \ppreceq \aert{\potent}{C}{0} + \potent~.
	\end{align*}%
\end{theorem}%
\begin{proof}
	%
	\abovedisplayskip=-1\baselineskip%
	\begin{align*}
	      & \ert{C}{\pi}  - \potent 
	       \eeq 
	       \aert{\potent}{C}{0}
	     \tag{by \Cref{thm:aert_decomp} where $\arta = 0$} \\
	     \text{implies}\quad &\ert{C}{0} - \potent 
	     \ppreceq
	     \aert{\potent}{C}{0}
	     \tag{$\ert{C}{0} \preceq \ert{C}{\potent}$ by monotonicity of $\ertsymbol$}\\
	     \text{implies}\quad &\ert{C}{0}
	     \ppreceq
	     \aert{\potent}{C}{0} + \potent \tag*{\qedhere}
	\end{align*}%
	\normalsize%
\end{proof}%
\noindent%
A further handy decomposition of $\aertsymbol$ is so-called \emph{constant propagation} also known from the non-\RSL $\ertsymbol$ calculus of \citet{kaminski2018weakest}:%
\begin{theorem}[Constant Propagation for $\aertsymbol$]
	\label{thm:aert_const-prop}
	For all $C\in\hpgcl$ and all constant $\mathit{const}\in\T$,
	\begin{align*}
		\aert{\potent}{C}{\arta + \mathit{const}} \ppreceq \aert{\potent}{C}{\arta} + \mathit{const}~.
	\end{align*}%
\end{theorem}%
\begin{figure}[t]
				\renewcommand{\arraystretch}{1}
				\begin{align*}
					\begin{array}{l@{\hspace{4em}}l}
						\boldaertsymbol\boldsymbol{\colon}			& \boldertsymbol\boldsymbol{\colon}		\\			
						\aannotate{1} 							& \annotate{\tfrac{x}{2}}				\\
						\{									& \{ 								\\
						\qquad \aannotate{x + 1} 					& \qquad \annotate{x}				\\
						\qquad \TICK{x} 						& \qquad \TICK{x}					\\
						\qquad \aannotate{1}						& \qquad \annotate{0}				\\
						\} \mathrel{[\sfrac{1}{2}]} \{ 				& \} \mathrel{[\sfrac{1}{2}]} \{ 			\\
						\qquad \aannotate{1 - x}					& \qquad \annotate{0}				\\
						\qquad \ASSIGN{x}{0} 					& \qquad \ASSIGN{x}{0}				\\
						\qquad \aannotate{1} 					& \qquad \annotate{0}				\\
						\} 		 							& \} 								\\
						\aannotate{1} 							& \annotate{0}						\\
						\ASSIGN{x}{x + 1} 						& \ASSIGN{x}{x + 1} 					\\
						\aannotate{0} 							& \annotate{0}
					\end{array}
				\end{align*}%
				\renewcommand{\arraystretch}{1}%
		
		\caption{$\aertsymbol$ reasoning (left) and $\ertsymbol$ reasoning (right).}
		\label{fig:aert_example_loopfree}
\end{figure}
\begin{example}[$\aertsymbol$ Reasoning]%
\label{ex:aert-loopfree}
Consider $\mathit{Op} = \COMPOSE{\PCHOICE{\TICK{x}}{\sfrac{1}{2}}{\ASSIGN{x}{0}}}{\ASSIGN{x}{x+1}}$.
This operation either consumes $x$ units of time or resets $x$ to $0$, each with probability $\sfrac{1}{2}$.
Furthermore, each invocation of $\mathit{Op}$ increases $x$. 
Let us now perform both $\aertsymbol$ as well as $\ertsymbol$ analyses.

For $\aertsymbol$, we choose as potential function $\potent = x$. 
We will then make program annotations as shown in \Cref{fig:aert_example_loopfree}.
The left one of these annotations are the $\aertsymbol$ annotations and they can be read (best from bottom to top) as follows:
$0$ is the postruntime.
$1$ is the result (after simplifications) of $\aert{\potent}{\ASSIGN{x}{x + 1}}{0} = 0 + x\subst{x}{x+1} - x = 1$.
The resulting $1$ is also the postruntime to consider for both branches of the probabilistic choice.
$1-x$ is the result (again, after simplification) of $\aert{\potent}{\ASSIGN{x}{0}}{1}$.
Likewise for the other branch of the probabilistic choice.
Finally, at the very top, $1$ is the result of combining (and simplifying) the two outcomes of the branches of the probabilistic choice according to the rule for $\aertsymbol$ --- in this case: \mbox{$1 = \tfrac{1}{2} \cdot (x+1) + \tfrac{1}{2} \cdot (1 - x)$}.
The same annotation style applies to the $\ertsymbol$ annotations on the far right.

We can read off that the amortized expected time to perform a single $\mathit{Op}$ is 1.
Moreover, \Cref{thm:aert_const-prop} immediately yields \mbox{$\aert{\potent}{\mathit{Op}}{\mathit{const}} = \mathit{const} + 1$} and from there it is easy to prove by induction that $\aert{\potent}{\mathit{Op}^n}{0} = n$. From there, in turn, we obtain by \Cref{thm:aert_sound} that%
\begin{align*}
	\ert{\mathit{Op}^n}{0} \ppreceq \aert{\potent}{\mathit{Op}^n}{0} + x \eeq n + x~.
\end{align*}%
This was relatively easy and gives us a very clear upper bound on the expected time it takes to execute a sequence of $n$ $\mathit{Op}$'s, namely $n + x$ which is $n$ if potential $x$ was 0 initially.

Obtaining the same insight solely via $\ertsymbol$ reasoning would have been harder.
First, we can read off that $\ert{\mathit{Op}}{0} = \tfrac{x}{2}$.
But what about $\ert{\mathit{Op}^n}{0}$?
For $n = 2$, we get $\ert{\COMPOSE{\mathit{Op}}{\mathit{Op}}}{0} = \tfrac{3x}{4} + \tfrac{1}{4}$.
For $n = 3$, we get $\tfrac{7x}{8} + \tfrac{5}{8}$.
For $n = 4$, we get $\tfrac{15x}{16} + \tfrac{17}{16}$.
Seeing a pattern here is less easy than it was for $\aertsymbol$, especially for the constant part of the term. 
\end{example}

\needspace{2\baselineskip}
\subsection{Local Reasoning for Amortized Expected Runtimes}

To enable local reasoning, the $\aertsymbol$ calculus also features a \emph{frame rule for upper bounds}:%
\begin{theorem}[Frame Rule for $\aertsymbol$]
	\label{thm:aert_frame}
	Let $C\in\hpgcl$ and $\rta, \rtb \in \T$. Then
	\begin{align*}
		&\modC{C} \cap \Vars(\rtb) = \emptyset
		\qqimplies \\
		& \qquad \aert{\potent}{C}{(\rta \sepadd \rtb) - \potent} \ppreceq \bigl( (\aert{\potent}{C}{\rta - \potent} + \potent ) \sepadd \rtb \bigr) - \potent~.
	\end{align*}%
\end{theorem}%
\begin{proof}
	\abovedisplayskip=-1\baselineskip%
	\begin{align*}
		\aert{\potent}{C}{(\rta \sepadd \rtb) - \potent} 
		& \eeq \ert{C}{\rta \sepadd \rtb} - \potent \tag{by \Cref{thm:aert_decomp}} \\
		& \ppreceq (\ert{C}{\rta} \sepadd \rtb) - \potent \tag{by \Cref{thm:ert_frame}} \\
		& \eeq \bigl( ( \ert{C}{\rta} - \potent + \potent )\sepadd \rtb\bigr) - \potent  \\
		& \eeq \bigl( ( \aert{\pi}{C}{\rta - \potent} + \potent )\sepadd \rtb\bigr) - \potent\hspace{-4em} \tag*{(by \Cref{thm:aert_decomp})\quad\qedhere}
	\end{align*}%
	\normalsize%
\end{proof}%
%
%
\noindent%
While this rule may look quite involved, it is still helpful:
(1) postruntimes during reasoning are often of the form $\cloze{f} - \potent$ and (2) the heavy lifting, \ie the $\aertsymbol$-reasoning, is still done more locally, namely on $\rta - \potent$ instead of $(\rta \sepadd \rtb) - \potent$.

\subsection{Compositional Reasoning about Nested Data Structures}\label{sec:compositional-nested}
The $\aertsymbol$ calculus is parameterized in a potential function $\potent$, which must be chosen carefully with respect to the data structure that is analyzed. 
On a first glance, having to fix a potential function for $\aertsymbol$-based reasoning might hamper compositionality.
For example, assume we have already analyzed the amortized expected runtime of a data structure, say $D_1$, using some potential function~$\potent_1$. 
Furthermore, suppose a second data structure, say $D_2$, internally uses $D_1$ as a sub-component, and its analysis requires a slightly different potential function, say $\potent_1 \sepadd \potent_2$, to account for additional potential of other elements of $D_2$.
Naively, we then have to analyze the $\aertsymbol$ of $D_1$ again, this time with respect to the extended potential $\potent_1 \sepadd \potent_2$.
However, this \emph{need not be necessary}.
Instead, we can re-use our existing analysis of $D_1$ with respect to potential $\potent_1$ to analyze the amortized expected runtime of $D_2$ with respect to the extended potential $\potent_1 \sepadd \potent_2$.

More precisely, we can -- under mild conditions for $\potent_1$ and $\potent_2$ -- re-use an existing upper bound on $\aert{\potent_1}{C}{\arta}$ to obtain an upper-bound on $\aert{\potent_1\oplus \potent_2}{C}{\arta}$.
The following theorem leverages the frame rule for $\aertsymbol$ and monotonicity to enable such compositional reasoning:
\begin{theorem}
	\label{thm:aert_compositional}
    Let 	$\potent_1,\potent_2$ be potentials with $\potent_1 \preceq \potent_1 \sepadd \potent_2$.
    Then, for all $C \in \hpgcl$ and $\arta \in \A_{\potent_1}$,
	%
	\begin{align*}
	   &\modC{C} \cap \Vars(\potent_2) = \emptyset 
	      \quad\text{and}\quad \arta + (\potent_1 \sepadd \potent_2) \preceq (\arta + \potent_1) \sepadd \potent_2
	   \\
	   \qqimplies & \aert{\potent_1 \oplus \potent_2}{C}{\arta} \preceq 
	      (\aert{\potent_1}{C}{\arta} + \potent_1) \sepadd \potent_2 ~~~-~~~ \potent_1 \sepadd \potent_2~.
	\end{align*}
\end{theorem}
\begin{proof}
	Notice that all of the above expressions are well-defined. Then, consider the following:
	\begin{align*}
	   & \aert{\potent_1 \oplus \potent_2}{C}{\arta} \\
	   \eeq &\ert{C}{\arta + (\potent_1 \oplus \potent_2)} ~~~-~~~ \potent_1 \oplus \potent_2
	   \tag{by \Cref{thm:aert_decomp}} \\
	   \ppreceq & \ert{C}{(\arta + \potent_1) \oplus \potent_2} ~~~-~~~ \potent_1 \oplus \potent_2
	   \tag{by monotonicity of $\ertsymbol$ (\Cref{thm:ert-health}) and assumption} \\
	   \ppreceq & \ert{C}{\arta + \potent_1} \oplus \potent_2  ~~~-~~~ \potent_1 \oplus \potent_2
	   \tag{by frame rule for $\ertsymbol$ (\Cref{thm:ert_frame})} \\
	   \eeq & (\aert{\potent_1}{C}{\arta} + \potent_1)
	       \oplus \potent_2  ~~~-~~~ \potent_1 \oplus \potent_2
	   \tag*{(by \Cref{thm:aert_decomp})\qquad\qedhere}
	\end{align*}
\end{proof}
\noindent
In the above proof, the assumptions $\modC{C} \cap \Vars(\potent_2) = \emptyset $ and $\arta + (\potent_1 \sepadd \potent_2) \preceq (\arta + \potent_1) \sepadd \potent_2$ enable framing the potential $\potent_2$; we remark that the latter assumption immediately holds if both potentials do not depend on the heap.
Moreover, the assumption $\potent_1 \preceq \potent_1 \sepadd \potent_2$ ensures that $\aert{\potent_1 \oplus \potent_2}{C}{\arta}$ is well-defined. 
We use the above compositionality theorem to analyze a load-balancing approach on top of an already analyzed amortized data structure in \Cref{sec:load_balanced_table}.

\subsection{Reasoning about Loops}

As $\aertsymbol$ for loops is also defined via a least fixed point of a function $\aertcharfun{\arta}$, similarly to $\ertsymbol$, we obtain an invariant-based proof rule for upper-bounding  amortized expected runtimes:%
\begin{theorem}[Park Induction for $\aertsymbol$]
\label{aert:induction}
	Let $\mathit{loop} = \WHILEDO{\guard}{\mathit{body}}$ and $\arta,\ainv\in \Api$. Then
	\begin{align*}
	\aertcharfun{\arta}(\ainv) \ppreceq \ainv \quad \text{implies} \quad
	\aert{\potent}{\mathit{loop}}{\arta} \ppreceq \ainv~.
	\end{align*}%
\end{theorem}%
\begin{example}
	Consider the loop in \Cref{fig:loopy_aert} with clearly non-constant expected runtime. For every $m \in \Nats$, let $\nextpoww{m}$ be the smallest power of $2$ greater than- or equal to $m$, and let $\poww{m}$ be the predicate that evaluates to $\true$ iff $m$ is a power of $2$.
	Using the potential function $\potent \defeq 2\cdot x - \nextpoww{x}$ and the $\aertsymbol$ loop invariant $\ainv = \iverson{\varcur = 0}\cdot 2$, the amortized expected runtime of the loop is shown to be constant. 
	Annotations are best read from bottom to top:
	$\arta \coloneqq 0$ is the postruntime.
	Somewhat differently from \Cref{ex:aert-loopfree}, $0$ is not copied to the loop body, but instead, invariant $\ainv$ is employed and pushed through the loop body (possibly with simplifications and \emph{overapproximations}), obtaining $\iverson{\varcur = 0} +1$.
	We have now overapproximated $\aert{\potent}{\mathit{body}}{\ainv}$ by $\iverson{\varcur = 0} +1$.
	To resemble an overapproximation of the characteristic function $\aertcharfun{0}$ applied to $\ainv$, we construct $\artc = \iverson{\neg \guard} \cdot 0 \pplus \iverson{\guard} \cdot (\iverson{\varcur = 0} +1)$.
	The final (topmost) annotation indicated that, indeed $\ainv \succeq Z$ which in total confirms $\ainv \succeq \aertcharfun{0}(\ainv)$, thus confirming --- by \Cref{aert:induction} --- that $\ainv$ is an upper bound for the total amortized expected runtime of the loop with respect to postruntime $0$. 
	
	\begin{figure}[t]
		
	\begin{align*}
	  &\asucceqannotate{\iverson{\varcur=0} \cdot 2} 
	  \tag{$\aertcharfun{0}(\ainv) \preceq \ainv$
	                hence $\aert{\potent}{\mathit{loop}}{0} \preceq \ainv$}\\
	  & \psiannotate{\iverson{\varcur \neq 0} \cdot 0 \pplus \iverson{\varcur = 0} \cdot (\iverson{\varcur=0} + 1)} \\
	  & \WHILE{\varcur = 0} \\
	  &\qquad \aannotate{\iverson{\varcur=0} + 1}
	  \tag{obtain $\aert{\potent}{\mathit{body}}{\ainv} \preceq  \iverson{\varcur = 0} +1$} \\
	  &\qquad \aannotate{\sfrac{1}{2} \cdot (0 + \iverson{\varcur=0} \cdot 2 + 2) } \\
	  &\qquad \{ \\
	   & \qquad \qquad  \aannotate{0} \\
	  &\qquad \qquad \ASSIGN{\varcur}{1}\\
	  & \qquad \qquad  \aannotate{\iverson{\varcur = 0}\cdot 2} \\
	  &\qquad \} [\sfrac{1}{2}] \{ \\
	  & \qquad \qquad \aannotate{\iverson{\varcur=0} \cdot 2 + 2} \\
	  %
	  %
	  &\qquad \qquad \ITE{\poww{\varremove}}{\TICK{\varremove}}{\SKIP} \fatsemi \\
	  & \qquad \qquad  \aannotate{\iverson{\varcur = 0}\cdot 2 + 2 - \iverson{\poww{\varremove}} \cdot \varremove} \\
	  %
	  %
	  &\qquad \qquad \ASSIGN{\varremove}{\varremove + 1} \\
	  & \qquad \qquad\aannotate{\iverson{\varcur = 0}\cdot 2} \\
	  &\qquad \} 
	  ~~ \astarannotate{\iverson{\varcur = 0} \cdot 2} 
	  \tag{we employ invariant $\ainv \coloneqq \iverson{\varcur = 0} \cdot 2$}\\
	  &\} 
	  ~\aannotate{0}
	  \tag{$0$ is the post $\potent$-runtime}
	\end{align*}%

		\caption{A probabilistic loop with constant amortized expected runtime. Here $\potent = 2\cdot \varremove - \nextpoww{\varremove}$.}%
		\label{fig:loopy_aert}%
	\end{figure}
\end{example}

\subsection{Case Studies}

\subsubsection{The Randomized Dynamic Table}
\label{sec:dynlist}
A dynamic table is a dictionary data structure for maintaining a table of elements in the heap. The data structure provides, amongst others, an operation $\listinsert{\varadd}$ for inserting a new element with content $\varadd$ at the end of the table. We first describe a well-known deterministic implementation of dynamic tables using fixed-size arrays that runs in  \emph{constant} amortized time. We then employ our $\aertsymbol$ calculus to prove that a \emph{randomized} variant of this implementation runs in \emph{constant expected} amortized time.

We can implement dynamic tables by means of fixed-size arrays: Maintain an array $\ahead$ of size $\asize \geq 1$ in the heap and keep track of the number of cells $\aoff$ currently occupied by some element. A call to $\listinsert{\varadd}$ then behaves as follows. If $\aoff < \asize$, then store $\varadd$ at $\ahead[\aoff]$---the first (i.e., with smallest offset) cell that is not occupied by some element yet---and increase $\aoff$ by one. In this case, we assume a runtime of $1$ for storing the value $\varadd$. Otherwise, i.e, if  $\aoff = \asize$, we need to allocate a new array $\ahead'$ of size $\asize' > \asize$, copy all elements from $\ahead$ to $\ahead'$, and store the new element $v$ at $\ahead'[\aoff]$. We then increase $\aoff$ by one, deallocate the old array $\ahead$, and set $\ahead$ to $\ahead'$. In this case, we assume a runtime of $\asize+1$ for copying the elements from $\ahead$ to $\ahead'$ and for storing the new element $\varadd$. A clear downsize of this implementation is that $\listinsert{\varadd}$ has a \emph{non-constant} worst-case runtime of $\asize+1$.

However, by choosing the size $\asize'$ of the new array $\ahead'$ carefully, we can do better in an \emph{amortized} sense---a prime example of amortized analysis \cite[Chapter 17]{leiserson2001introduction}. For $\asize' = 2\cdot \asize$, i.e., by doubling the size of the array each time it is full, we achieve a \emph{constant} amortized time for  $\listinsert{\varadd}$. We remark that, when increasing the list size by some constant by, e.g., choosing $\asize' = \asize+1$, the amortized runtime of $\listinsert{\varadd}$ is \emph{not} constant. 

We now consider our randomized variant $\procrandlistinsert{\varadd}$ shown in \Cref{fig:rand_dynlist} on p.\ \pageref{fig:rand_dynlist}. Instead of deterministically choosing $\asize' = 2 \cdot \asize$, our randomized implementation chooses $\asize' = \asize +1$ with probability $\nicefrac{1}{\asize +1}$ and $\asize' = 2 \cdot \asize$ with the remaining probability $1-\nicefrac{1}{\asize +1}$ in case $\aoff=\asize$. Thus, for small array sizes $\asize$, our randomized variant behaves with high probability like the deterministic variant with non-constant amortized runtime, while \emph{in the limit} behaving like the classical variant with constant amortized runtime. Program $\arraycopy{\ahead}{\asize}{\ahead'}$ copies the array $\ahead$ of size $\asize$ to $\ahead'$ and consumes $\asize$ units of time. $\arraydelete{\ahead}{\asize}$ deallocates the array $\ahead$  and consumes no time.

Using our $\aertsymbol$ calculus and the potential $\potent = 2\cdot \aoff \monus \asize$, we prove in a fully calculational way that our randomized variant is \emph{memory-safe} and runs in \emph{constant amortized expected time}. We have 
%
%
%
\begin{align}
	\label{eqn:dynlist}
\aert{\potent}{\procrandlistinsert{\varadd}}{0} \lleq \emprun{4}  \sepadd \pureemp{\aoff\leq \asize \wedge \asize \geq 1} \sepadd \bigoplus_{i=1}^{\asize} \ivalidpointer{\ahead + i-1} ~,
\end{align}
i.e., when invoked on an array with head $\ahead$ of size at least $1$ and where the offset $\aoff$ of the last occupied cell is at most $s$, $\procrandlistinsert{\varadd}$ is \emph{memory-safe} and runs in an \emph{amortized expected time} of at most $4$. We emphasize that local reasoning simplifies our amortized analysis significantly: The frame rule enables to specify memory-safety and expected runtimes of the sub-programs $\arraycopy{\ahead}{\asize}{\ahead'}$ and $\arraydelete{\ahead}{\asize}$ separately, and to employ these specifications in the broader context of $\procrandlistinsert{\varadd}$.

\begin{figure}[t]
	\begin{align*}
		%
		& \IF{\aoff = \asize} \\
		%
		%
		&\qquad  \PCHOICE{\ASSIGN{\asize'}{s+1}}{\nicefrac{1}{\asize +1}}{\ASSIGN{\asize'}{2\cdot \asize}} \fatsemi \\
		%
		%
		& \qquad \ALLOC{\ahead'}{\asize'}\fatsemi \\
		&\qquad
		\arraycopy{\ahead}{\asize}{\ahead'}\fatsemi 
		\\
		&\qquad 
		\arraydelete{\ahead}{\asize} ~ \fatsemi \\
		&\qquad \ASSIGN{\ahead}{\ahead'} \fatsemi \\
		&\qquad 
		\ASSIGN{\asize}{\asize'} \\ %
		%
		%
		&\} \fatsemi\\
		%
		%
		&\HASSIGN{\ahead+\aoff}{y} \fatsemi \\%
		& 
		\ASSIGN{\aoff}{\aoff+1} \fatsemi \\
		%
		&\TICK{1}
		%
		%
	\end{align*}%
	\caption{Randomized Dynamic Table Insert $\procrandlistinsert{y}$.}%
	\label{fig:rand_dynlist}%
\end{figure}

%
%

\subsubsection{The Load-balanced Randomized Dynamic Table}
\label{sec:load_balanced_table}
To demonstrate the $\aertsymbol$ calculus' capabilities for compositional reasoning about nested data structures
(cf.\ \Cref{sec:compositional-nested}), we consider a variant of the randomized dynamic table in \Cref{sec:dynlist}
that internally uses \emph{two} dynamic tables instead of one for load balancing reasons, \eg, to enable parallel operations on the smaller tables.
In particular, our variant supports an operation $\procbalancedlistinsert{\varadd}$ for inserting a value $\varadd$.
To ensure that the load of the two internal dynamic tables is balanced in expectation, we flip a fair coin to decide in which of the dynamic tables the value $\varadd$ is to be inserted.

To model this operation in $\hpgcl$, we take two copies $\procrandlistinsertcopy{1}{\varadd}$ and $\procrandlistinsertcopy{2}{\varadd}$ of the program in \Cref{fig:rand_dynlist}, where every variable $x$ is replaced by a copy variable $x_1$ and $x_2$, respectively. The operation $\procbalancedlistinsert{\varadd}$ is then given by the $\hpgcl$ program%
%
%
\begin{align*}
   \PCHOICE{\procrandlistinsertcopy{1}{\varadd}}{0.5}{\procrandlistinsertcopy{2}{\varadd}}~.
\end{align*}%
%
%
We analyze the amortized expected runtime $\aert{\potent_1 \sepadd \potent_2}{\procbalancedlistinsert{\varadd}}{0}$ using the \emph{extended potential function} $\potent_1 \sepadd \potent_2$, where, for $i \in \{1,2\}$, the potential $\potent_i = 2\cdot \aoff_i \monus \asize_i$ is a copy of the potential used for analyzing $\procrandlistinsert{y}$. 
We will apply \Cref{thm:aert_compositional} to reuse our existing analysis for $\procrandlistinsert{y}$ (c.f.\ \Cref{eqn:dynlist}) and the frame rule (\Cref{thm:aert_frame}) to account for the second dynamic table.
To this end, we first calculate for $j=1$ and $j'=2$ (and analogously for $j=2$ and $j'=1$):
\begin{align*}
&\aert{\potent_j}{\procrandlistinsertcopy{j}{\varadd}}{0} \\
\ppreceq &  \aert{\potent_j}{\procrandlistinsertcopy{j}{\varadd}}{\potent_j\sepadd\bigoplus_{i=1}^{\asize_{j'}} \ivalidpointer{\ahead_{j'} + i-1} ~~~ - \potent_j}
\tag{by monotonicity of $\aertsymbol$} \\
\ppreceq &  (\aert{\potent_j}{\procrandlistinsertcopy{j}{\varadd}}{0}+\potent_j) \oplus 
\bigoplus_{i=1}^{\asize_{j'}} \ivalidpointer{\ahead_{j'} + i-1} ~~~ - \potent_j
\tag{by \Cref{thm:aert_frame}} \\
\ppreceq &  \big(\emprun{4}  \sepadd \pureemp{\aoff_j\leq \asize_j \wedge \asize_j \geq 1} \sepadd \bigoplus_{i=1}^{\asize_j} \ivalidpointer{\ahead_j + i-1}+\potent_j\big) \oplus 
\bigoplus_{i=1}^{\asize_{j'}} \ivalidpointer{\ahead_{j'} + i-1} ~~~ - \potent_j
\tag{\Cref{eqn:dynlist}} \\
\ppreceq &  \underbrace{\emprun{4}  \sepadd \pureemp{\aoff_j\leq \asize_j \wedge \asize_j \geq 1} \sepadd \big(\bigoplus_{i=1}^{\asize_j} \ivalidpointer{\ahead_j + i-1}\big) \oplus 
	\big(\bigoplus_{i=1}^{\asize_{j'}} \ivalidpointer{\ahead_{j'} + i-1}}_{{}\eqqcolon \arta_j}\big)
\tag{$\potent_j$ does not depend upon the heap}
\end{align*}%
In other words, we can extend the bound for $\procrandlistinsertcopy{j}{\varadd}$ by a specification involving the array $\ahead_{j'}$ not occurring in $\procrandlistinsertcopy{j}{\varadd}$. This gives us
%
%
%
%
%
\begin{align*}
	&\aert{\potent_1 \sepadd \potent_2}{\PCHOICE{\procrandlistinsertcopy{1}{\varadd}}{0.5}{\procrandlistinsertcopy{2}{\varadd}}}{0} \\
	\eeq & 0.5 \cdot  \aert{\potent_1 \sepadd \potent_2}{\procrandlistinsertcopy{1}{\varadd}}{0}+ 0.5\cdot \aert{\potent_1 \sepadd \potent_2}{\procrandlistinsertcopy{2}{\varadd}}{0}
	\tag{by \Cref{table:aert}} \\
	\ppreceq & 0.5 \cdot \big ( (\aert{\potent_1}{\procrandlistinsertcopy{1}{\varadd}}{0} + \potent_1) \sepadd \potent_2 ~~~-~~~ \potent_1 \sepadd \potent_2 \big) \\
	                 & \qquad +  0.5 \cdot \big ( (\aert{\potent_2}{\procrandlistinsertcopy{2}{\varadd}}{0} + \potent_2) \sepadd \potent_1 ~~~-~~~ \potent_1 \sepadd \potent_2 \big)
	\tag{apply \Cref{thm:aert_compositional} twice} \\
	\ppreceq & 0.5 \cdot \big ( (\arta_1+ \potent_1) \sepadd \potent_2 ~~~-~~~ \potent_1 \sepadd \potent_2 \big)  +  0.5 \cdot \big ( (\arta_2 + \potent_2) \sepadd \potent_1 ~~~-~~~ \potent_1 \sepadd \potent_2 \big) 
	\tag{by above reasoning}
	\\
	\ppreceq & 0.5 \cdot \arta_1+  0.5 \cdot \arta_2
	\tag{$\potent_1$ and $\potent_2$ do not depend upon the heap} \\
	\eeq & \emprun{4}  \sepadd \pureemp{\aoff_1\leq \asize \wedge \asize_1 \geq 1 \wedge \aoff_2\leq \asize \wedge \asize_2 \geq 1} \\
	 &\quad \sepadd \big(\bigoplus_{i=1}^{\asize_1} \ivalidpointer{\ahead_1 + i-1}\big) \sepadd \big(\bigoplus_{i=1}^{\asize_2} \ivalidpointer{\ahead_2 + i-1}\big) ~.
\end{align*}%
That is, if both of the arrays $\ahead_1$ and $\ahead_2$ satisfy their respective specifications, then $\procbalancedlistinsert{\varadd}$ is \emph{memory-safe} and runs in \emph{constant} amortized expected time.

\subsubsection{The Insert-Delete-FindAny Problem \cite{DBLP:journals/njc/BrodalCR96}}

\begin{figure}[t]
	\begin{subfigure}[t]{.45\linewidth}
		\begin{align*}
		& \procremove{x}\fatsemi \\
		& \IF{\ls = 0} \\
		&\qquad \COMPOSE{\ASSIGN{\varany}{0}}{\ASSIGN{\varrank}{0}}  \\
		&\ELSE  \\
		& \qquad \IF{\varremove = \varany} \\
		&\qquad \qquad \COMPOSE{\procsample}{\procrank} \\
		&\qquad \ELSE \\
		& \qquad \qquad \COMPOSE{\ASSIGNH{\varvalrem}{\varremove+2}}{\ASSIGNH{\varvalany}{\varany+2}}\fatsemi\\
		&\qquad\qquad \TICK{1}\fatsemi \\
		&\qquad \qquad \IF{\varvalrem < \varvalany} 
		%
		\ASSIGN{\varrank}{\varrank \monus 1}~\} \\
		&\qquad	 \} \\
		&\} \fatsemi
		%
		\FREE{\varremove,\varremove+1, \varremove+2}
		\end{align*}
		\caption{Program $\procdelete{\varremove}$.}
	\end{subfigure}
	\quad\hfill
	\begin{subfigure}[t]{.45\linewidth}
		\vspace{-0.01cm}
		\begin{align*}
		&\procadd{\varadd} \fatsemi \\
		&\{ \\
		&\qquad \ASSIGN{\varany}{\lh} \fatsemi  \procrank \\
		&\}~[\nicefrac{1}{\ls+1}]~\{ \\
		&\qquad \IF{\ls \geq 2} \\
		&\qquad\qquad \ASSIGNH{\varvalany }{\varany+2}\fatsemi   \\
		&\qquad \qquad \TICK{1}\fatsemi \\
		&\qquad\qquad \IF{\varadd < \varvalany}
		%
		\ASSIGN{\varrank}{\varrank +1}~\} \\
		%
		%
		&\qquad \ELSE \\
		&\qquad \qquad \COMPOSE{\ASSIGN{\varany}{\lh}}{\ASSIGN{\varrank}{1}} \\
		&\qquad \} \\
		%
		%
		&\} 
		\end{align*}
		\caption{Program $\procinsert{\varadd}$.}
	\end{subfigure}
	
	\caption{Programs $\procdelete{\varremove}$ and $\procinsert{\varadd}$.}
	\label{fig:findany_delete_insert}
\end{figure}

The Insert-Delete-FindAny problem is to maintain a dictionary data structure storing numbers, which provides three operations:
\begin{itemize}
	\item $\procinsert{\varadd}$ inserts a new element with content $\varadd$ into the dictionary.
	\item $\procdelete{\varremove}$ gets a \emph{pointer} $\varremove$ to some element in the dictionary and removes this element.
	\item $\procfindany$ returns an \emph{arbitrary} element $\varany$ from the dictionary together with its \emph{rank}, which is defined as one plus the number of elements in the dictionary whose content is strictly smaller than the content of $\varany$, or returns $0$ if the dictionary is empty.
\end{itemize}
We assume a runtime model that counts the number of comparisons of elements in the dictionary. 
\citet{DBLP:journals/njc/BrodalCR96} provide randomized algorithms of the above operations using doubly-linked lists  that each run in \emph{constant} amortized expected time. Remarkably, they prove that every \emph{deterministic} implementation is less efficient in the sense that  there is \emph{no} deterministic implementation of the above operations that achieves a constant amortized runtime.

We encode the algorithms provided by \citet{DBLP:journals/njc/BrodalCR96} in $\hpgcl$ and use our $\aertsymbol$ calculus to prove on source-code level that these operations indeed run in constant amortized expected time. The operations $\procinsert{\varadd}$ and $\procdelete{\varremove}$ are depicted in \Cref{fig:findany_delete_insert}. $\procfindany$ is realized by maintaining the variables $\varany$ and $\varrank$ with the desired properties. We specify doubly-linked lists co-inductively:
\allowdisplaybreaks%
\begin{align*}
&\dll{\lh}{\lend}{\lpre}{\ls}{\lsucc} \\
\ddefeq&  \pureemp{\lh=\lsucc \wedge \lpre=\lend\wedge \ls = 0 } \\
 &\qquad {}\sqcap \big(\pureemp{\ls \geq 1} \sepadd
\Inf v \colon \singleton{\lh}{v, \lpre,{-}} \sepadd \dll{v}{\lend}{\lh}{\ls-1}{\lsucc}\big)
\end{align*}
In particular, $\Inf \lend \colon \dll{\lh}{\lend}{0}{\ls}{0}$ specifies that the heap consists of a null-terminated doubly-linked list (with head) $\lh$. Every element $\varremove$ of a doubly-linked list consists of three locations in the heap: location $\varremove$ stores the \emph{successor} element (or  $0$ if $\varremove$ is the last element), location $\varremove +1$ stores the \emph{predecessor} element (or $0$ if $\varremove$ is the first element), and $\varremove +2$ stores the \emph{content}.

Let us now consider $\procdelete{\varremove}$. Assume that the heap consists of a doubly-linked list $\lh$ of length $\ls$ containing the element $\varremove$. We first execute $\procremove{\varremove}$, which removes the element $\varremove$ from the list without deallocating the locations associated to $\varremove$, and decreases $\ls$ by one. Then, if the list becomes empty, we set $\varany$ and $\varrank$ to $0$. Otherwise, i.e, if the list is not empty, there are two possible cases: Either $\varremove = \varany$, which means that we removed our current $\varany$ element. We thus need to find a new $\varany$ together with its rank. This is realized by the programs $\procsample$ and $\procrank$. $\procsample$ first samples some element uniformly at random from the list with head $\lh$ and stores the result in variable $\varany$ (requires no comparisons). $\procrank$ then computes the rank of the new $\varany$ and stores the result in variable $\varrank$ (the number of comparisons required is $\ls$). In the other case, we have $\varremove \neq \varany$. We check whether the rank of $\varany$ needs to be updated by comparing the content of $\varremove$ to the content of $\varany$. Finally, we deallocate the pointers associated to the element $\varremove$.

Let us now consider $\procinsert{\varadd}$. Assume that the heap consists of a doubly-linked list $\lh$ of length $\ls$. We start by executing $\procadd{\varadd}$, which allocates a new element with content $\varadd$, inserts this new element at the front of the doubly-linked list with head $\lh$---thus becoming the new head---, and increases $\ls$ by one. We then proceed randomly: With probability $\nicefrac{1}{\ls+1}$, we set $\varany$ to the new element $\lh$ and compute its rank (which again takes $\ls$ comparisons). With the remaining probability $1-\nicefrac{1}{\ls+1}$, we either keep the current $\varany$ and check whether its rank needs to updated if the list was not empty before (i.e., if $\ls \geq 2$), or set $\varany$ to $\lh$ and $\varrank$ to $1$ if the list was empty before (i.e., if $\ls \leq 1$).

Now define the potential $\potent \coloneqq \ls \cdot (1 + \iverson{\varany = \varremove})$. We prove 
\[
\aert{\potent}{\procdelete{\varremove}}{0} \ppreceq
1  + \isingleton{\varremove}{-,-,-}\sepadd 0
+ \isingleton{\varany}{-,-,-}\sepadd 0 
+ \Inf \lend \colon \dll{\lh}{\lend}{0}{\ls}{0}
\]
i.e., if the heap consists of a doubly-linked list $\lh$ of size $\ls$ containing the (not necessarily distinct) elements $\varany$ and $\varremove$, then $\procdelete{\varremove}$ is memory-safe and runs in an amortized expected time of at most $1$. Moreover, we show 
\[
\aert{\potent}{\procinsert{\varadd}}{0} 
\ppreceq 4 + \iverson{\ls \geq 1} \cdot(\isingleton{\varany}{-,-,-}\sepadd 0)
+ \Inf \lend \colon \dll{\lh}{\lend}{0}{\ls}{0}~ 
\]
i.e., if the heap consists of a doubly-linked list $\lh$ of size $\ls$ containing the element $\varany$ in case $\lh$ is non-empty, then $\procinsert{\varadd}$ is memory-safe and runs in an amortized expected time of at most $4$. 


\section{Related Work}
\label{s:related}

There is a plethora of research on the verification of runtime bounds.
We focus on literature most closely related to our approach, specifically techniques for formal reasoning about 
(1) expected runtimes of probabilistic programs and
(2) amortized runtimes of non-probabilistic programs.

\paragraph{Reasoning about expected runtimes}
Our $\ertsymbol$ calculus combines two existing approaches to enable proving upper bounds on expected runtimes of randomized algorithms manipulating dynamic data structures: the original $\ertsymbol$ calculus of~\citet{kaminski2018weakest} and quantitative separation logic (QSL) of ~\citet{QSLpopl}.
Developing a calculus based on a separating addition (our~$\sepadd$) was initially proposed by \citet[Chapter 9.1]{christophPhd}.
\citet[Chapter 4]{haslbeckPhd} formalized this idea and proved that one obtains a variant of QSL for reasoning about upper bounds. 
His calculus and its properties essentially coincide with our $\ertsymbol$ with two exceptions:
(1) $\ertsymbol$ allows the allocation of arbitrarily large chunks of memory instead of fixed-sized ones; and
(2) we prove soundness of $\ertsymbol$ with respect to an operational semantics based on MDPs; earlier attempts to prove soundness by \citet[p. 47]{haslbeckPhd} lead to technical issues which were not further pursued.

\citet{DBLP:conf/pldi/NgoC018,DBLP:journals/pacmpl/WangKH20} apply the potential method for automatic reasoning about expected runtimes. 
The soundness theorem in \citet{DBLP:conf/pldi/NgoC018} relies on the soundness of the original $\ertsymbol$ calculus of~\citet{kaminski2018weakest}. 
Our more general calculus for \RSL provides foundations for proving their techniques sound when applied to probabilistic pointer programs.
\citet{DBLP:journals/pacmpl/WangKH20} presents a type-based analysis for deriving over-approximations of expected runtime.
Their upper bounds are proven sound w.r.t.\ a distribution-based operational cost semantics.
Other approaches for analyzing expected runtimes of probabilistic programs, such as \cite{DBLP:conf/sas/Monniaux01,McIver:FM:2005,DBLP:journals/jcss/BrazdilKKV15,DBLP:conf/tacas/MeyerHG21,DBLP:conf/esop/MoosbruggerBKK21}, neither support dynamic data structures nor consider amortization.
Recently, \citet{DBLP:journals/corr/abs-2206-03537} defined a type-and-effect system for a functional programming language that can automatically infer logarithmic amortized bounds on randomized tree and heap structures.
Our amortized calculus is a weakest-precondition-style framework whose soundness w.r.t.\ an operational semantics is shown using a novel technique, which recovers a well-known interpretation of amortized expected runtime analysis at the level of program semantics for arbitrary sequences of data structure operations.

\paragraph{Verifying amortized runtimes for non-probabilistic programs}
Amortized runtime analysis builds upon either the \emph{potential method} (a.k.a. physicists view) or the \emph{banker's view} as already proposed in \citet{tarjan1985amortized}, who already noted that these two views are equivalent. 
\citet{DBLP:conf/tacas/HaslbeckN18} survey existing verification techniques for amortized runtimes.
In particular, the potential method has been formalized and applied in an interactive theorem prover by \citet{DBLP:conf/itp/Nipkow15,DBLP:journals/jar/NipkowB19}.
\citet{DBLP:conf/pldi/Carbonneaux0RS14} developed a quantitative logic similar to our $\ertsymbol$ transformer (for deterministic programs) based on potentials. 
Potential functions are also at the foundation of (automatic but not necessarily amortized) type-based runtime analyses, \eg, \cite{DBLP:journals/pacmpl/RajaniG0021,DBLP:conf/fossacs/Kahn020}, pioneered by \citet{DBLP:phd/dnb/Hoffmann11b}.
A recent survey of type-based analysis is given in \citet{hoffmann-jost2022}.

The banker's view of amortized analysis has been integrated into separation logic by \citet{DBLP:journals/corr/abs-1104-1998}. 
He introduced time credits, a dedicated resource modeling the remaining amount of time a program may consume. 
With this view, one can naturally reason about time credits in the same way as for heap allocated memory, \eg by storing time credits in individual elements of dynamic data structures.
In contrast to many other runtime verification techniques, Atkey proves his approach sound w.r.t.\ a program semantics.
The intricacies encountered when using time credits for reasoning about \emph{asymptotic} (amortized) complexities are discussed by \citet{DBLP:conf/esop/GueneauCP18}.
\citet{DBLP:journals/jar/ChargueraudP19} implemented time credits in a verification tool and verified the amortized complexity of the Union-Find data structure.
A variant of time credits, called time receipts~\cite{DBLP:conf/esop/MevelJP19}, enables reasoning about lower runtime bounds.

None of these works reason about amortized \emph{expected} runtimes of randomized algorithms. To enable this, we chose to use the potential method to formalize $\aertsymbol$, since potentials are closely related to expectations, which also map states to a quantity. Reasoning about potentials thus seems natural if one is used to working with expectations and quantitative invariants. Some of our proof rules exploit the above similarity to mix potentials and expectations, e.g.\ \Cref{thm:aert_frame,thm:aert_compositional}.

\section{Conclusion}
\label{s:conclusion}

We have presented calculi featuring compositionality and local reasoning for the verification of (amortized) expected runtimes of probabilistic pointer programs. We have established soundness results w.r.t.\ an operational semantics and demonstrated the applicability of our techniques.

Future work includes the runtime verification of (randomized) splay-trees \cite{sleator1985self,DBLP:conf/soda/Furer99,DBLP:journals/ipl/AlbersK02} and skip lists \cite{DBLP:conf/wads/Pugh89}, and mechanizing the $\aertsymbol$-calculus building upon the work by \citet{haslbeckPhd}. Further promising directions for automated $\aertsymbol$ reasoning include leveraging entailment checking techniques for quantitative separation logic  \cite{DBLP:conf/esop/BatzFJKKMN22} and generalizations of $k$-induction for probabilistic programs \cite{DBLP:conf/cav/BatzCKKMS20}.

\begin{acks}                            
  We thank Gerhard Woeginger on the fruitful discussions about amortized analysis.
  Furthermore, we are grateful for the reviewers for their highly constructive feedback that, in particular, contributed to the development of \Cref{thm:aert_compositional}. We also thank Eleanore Meyer for pointing out that standard addition and separating addition are not sub-distributive.
\end{acks}

\bibliographystyle{ACM-Reference-Format}
\bibliography{bibfile}

\allowdisplaybreaks
\appendix
\newpage

\section{Appendix}
\subsection{Appendix to Section~\ref{s:ert}}
\subsection{Auxiliary results}

\begin{lemma}\label{lem:inf-prenex}
	For $\rta,\rtb \in \T$ and $x \notin \Vars(\rta)$, we have
	\[
	   \rta \sepadd \Inf x\colon \rtb \qeq \Inf x\colon \rta \sepadd \rtb.
	\]
\end{lemma}

\begin{proof}
For $(\sk,\hh) \in \States$ consider the following:

\begin{align*}
&
(\rta \sepadd \Inf x\colon \rtb )(\sk,\hh)
\\ 
\eeq &
\min \{ \rta(\sk,\hh_1) \pplus \inf_v \rtb(\sk\statesubst{x}{v},\hh_2) ~|~ \hh = \hh_1 \sepcon \hh_2 \}
\tag{Def. of $\sepadd$} \\
\eeq &
\inf \{ \rta(\sk,\hh_1) \pplus \inf_v \rtb(\sk\statesubst{x}{v},\hh_2) ~|~ \hh = \hh_1 \sepcon \hh_2 \}
\tag{there are only finitely many partitions $\hh = \hh_1 \sepcon \hh_2$} 
\\
\eeq &
\inf \{ \inf_v \rta(\sk\statesubst{x}{v},\hh_1) \pplus  \rtb(\sk\statesubst{x}{v},\hh_2) ~|~ \hh = \hh_1 \sepcon \hh_2 \}
\tag{$x \notin \Vars(\rta)$} \\
\eeq &
\inf_v \inf \{ \rta(\sk\statesubst{x}{v},\hh_1) \pplus  \rtb(\sk\statesubst{x}{v},\hh_2) ~|~ \hh = \hh_1 \sepcon \hh_2 \}
\\
\eeq &
\inf_v \min \{ \rta(\sk\statesubst{x}{v},\hh_1) \pplus  \rtb(\sk\statesubst{x}{v},\hh_2) ~|~ \hh = \hh_1 \sepcon \hh_2 \}
\tag{there are only finitely many partitions $\hh = \hh_1 \sepcon \hh_2$} \\
\eeq &
\Inf x\colon \rta \sepadd \rtb.
\tag{Def. of $\sepadd$} 
\end{align*}

\end{proof}

\begin{lemma}
	\label{lem:lookup}
	For all $(\sk,\hh) \in \States$ and $\rta \in \T$, 
	\begin{align*}
		\ert{\ASSIGNH{x}{e}}{\rta}(\sk,\hh) \eeq
		\begin{cases}
			\rta(\sk\statesubst{x}{\hh(\sk(e))},\hh) & \qif \sk(e)\in\dom{\hh} \\
			\infty & \qif \sk(e) \not\in \dom{\hh}.
		\end{cases}
	\end{align*}
	
\end{lemma}

\begin{proof}
	
First, assume $\sk(e) \in \dom{\hh}$ and let $\hh = \hh_1 \sepcon \hh_2$ such that $\dom{\hh_1} = \{ \sk(e) \}$. Then:

\begin{align*}
& \ert{\ASSIGNH{x}{e}}{\rta}(\sk,\hh)
\\
\eeq & \left(\Inf\, v \colon \isingleton{e}{v} \sepadd \bigl(\isingleton{e}{v} \sepmon \rta\subst{x}{v}   \bigr)\right)(\sk,\hh)
\tag{\Cref{table:ert}} 
\\
\eeq & \left(\isingleton{e}{\hh(\sk(e))} \sepadd \bigl(\isingleton{e}{\hh(\sk(e))} \sepmon \rta\subst{x}{\hh(\sk(e))}   \bigr)\right)(\sk,\hh)
\tag{$\isingleton{e}{v}(\sk,\hh') = 0$ iff $v = \hh(\sk(e))$ and $\hh' = \hh_1$}
\\
\eeq &  \underbrace{\isingleton{e}{\hh(\sk(e))}(\sk,\hh_1)}_{\eeq 0} + \bigl(\isingleton{e}{\hh(\sk(e))} \sepmon \rta\subst{x}{\hh(\sk(e))}   \bigr)(\sk,\hh_2) 
\tag{Def. of $\sepadd$; minimum obtained only for $\hh = \hh_1 \sepcon \hh_2$}
\\
\eeq & \bigl(\isingleton{e}{\hh(\sk(e))} \sepmon \rta\subst{x}{\hh(\sk(e))}   \bigr)(\sk,\hh_2)
\\
\eeq & \sup \{ \rta\subst{x}{\hh(\sk(e))}(\sk,\hh_2 \sepcon \hh') \monus \isingleton{e}{\hh(\sk(e))}(\sk,\hh') ~|~ \hh' \disjoint \hh  \}
\tag{Def. of $\sepmon$}
\\
\eeq & \rta\subst{x}{\hh(\sk(e))}(\sk,\hh_2 \sepcon \hh') \monus \underbrace{\isingleton{e}{\hh(\sk(e))}(\sk,\hh')}_{\eeq 0}
\tag{$\isingleton{e}{\hh(\sk(e))}(\sk,\hh') = 0$ iff $\hh = \{\sk(e) \mapsto \hh(\sk(e))\}$}
\\
\eeq & \rta\subst{x}{\hh(\sk(e))}(\sk,\hh_2 \sepcon \hh')
\tag{$\hh' = \{\sk(e) \mapsto \hh(\sk(e))\}$}
\\
\eeq & \rta(\sk\statesubst{x}{\hh(\sk(e))}, \hh).
\tag{$\hh_2 \sepcon \hh(\sk(e)) = \hh$}
\end{align*}

Second, assume $\sk(e) \not\in \dom{\hh}$:

\begin{align*}
& \ert{\ASSIGNH{x}{e}}{\rta}(\sk,\hh)
\\
\eeq & \left(\Inf\, v \colon \isingleton{e}{v} \sepadd \bigl(\isingleton{e}{v} \sepmon \rta\subst{x}{v}   \bigr)\right)(\sk,\hh)
\tag{\Cref{table:ert}}
\\
\eeq & (\Inf\, v \colon\infty \sepadd \ldots)(\sk,\hh) \tag{for all $v$, and all $\hh_1 \sepcon \hh_2$, $\isingleton{e}{v}(\sk,\hh_1) = \infty$ since $\sk(e) \notin \dom{\hh}$}
\\
\eeq & \infty.
\end{align*}

\end{proof}

\begin{lemma}
	\label{lem:mut}
	For all $(\sk,\hh) \in \States$ and $\rta \in \T$,

    \begin{align*}
    	\ert{\HASSIGN{e}{e'}}{\rta}(\sk,\hh) \eeq
    	\begin{cases}
    		\rta(\sk,\hh\statesubst{\sk(e)}{\sk(e')}) & \qif \sk(e) \in \dom{\hh} \\
    		\infty & \qif \sk(e) \notin \dom{\hh}.
    	\end{cases}
    \end{align*}
    
\end{lemma}

\begin{proof}
First, assume $\sk(e) \in \dom{\hh}$ and let $\hh = \hh_1' \sepcon \hh_2'$ such that $\dom{\hh_1'} = \{ \sk(e) \}$. Then:

\begin{align*}
	 & \ert{\HASSIGN{e}{e'}}{\rta}(\sk,\hh) 
\\
\eeq & \left(\ivalidpointer{e} \sepadd \bigl(\isingleton{e}{e'} \sepmon \rta \bigr)\right)(\sk,\hh) 
\tag{\Cref{table:ert}} \\
\eeq & \min \{ \ivalidpointer{e}(\sk,\hh_1) \pplus \bigl(\isingleton{e}{e'} \sepmon \rta \bigr)(\sk,\hh_2) ~|~ \hh = \hh_1 \sepcon \hh_2 \} 
\tag{Def. of $\sepadd$} \\
\eeq & \ivalidpointer{e}(\sk,\hh_1') \pplus \bigl(\isingleton{e}{e'} \sepmon \rta \bigr)(\sk,\hh_2')
\tag{$\ivalidpointer{e}(\sk,\hh_1) = \infty$ unless $\hh_1 = \hh_1'$} 
\\ 
\eeq & \bigl(\isingleton{e}{e'} \sepmon \rta \bigr)(\sk,\hh_2') 
\\
\eeq & \sup \{ \rta(\sk,\hh_2' \sepcon \hh') \monus \isingleton{e}{e'}(\sk,\hh') ~|~ \hh' \disjoint \hh_2 \}
\tag{Def. of $\sepmon$}
\\
\eeq & \rta(\sk,\hh_2' \sepcon \{\sk(e) \mapsto \sk(e') \})
\tag{$\isingleton{e}{e'}(\sk,\hh') = 0$ iff $\hh' = \{\sk(e) \mapsto \sk(e') \}$}\\
\eeq & \rta(\sk, \hh\statesubst{\sk(e)}{\sk(e')}).
\end{align*}

Second, assume $\sk(e) \notin \dom{\hh}$. Then:

\begin{align*}
& \ert{\HASSIGN{e}{e'}}{\rta}(\sk,\hh) 
\\
\eeq & \left(\ivalidpointer{e} \sepadd \bigl(\isingleton{e}{e'} \sepmon \rta \bigr)\right)(\sk,\hh) 
\tag{\Cref{table:ert}} \\
\eeq & \min \{ \ivalidpointer{e}(\sk,\hh_1) \pplus \bigl(\isingleton{e}{e'} \sepmon \rta \bigr)(\sk,\hh_2) ~|~ \hh = \hh_1 \sepcon \hh_2 \} 
\tag{Def. of $\sepadd$} \\
\eeq & \min \{ \infty \pplus \bigl(\isingleton{e}{e'} \sepmon \rta \bigr)(\sk,\hh_2) ~|~ \hh = \hh_1 \sepcon \hh_2 \} 
\tag{$\ivalidpointer{e}(\sk,\hh_1) = \infty$ since $\sk(e) \notin \dom{\hh}$} \\
\eeq & \infty.
\end{align*}

\end{proof}

\begin{lemma}
	\label{lem:free}
	For all $(\sk,\hh) \in \States$ and $\rta \in \T$, 
	\begin{align*}
		\ert{\FREE{e}}{\rta}(\sk,\hh) \eeq
		\begin{cases}
			\rta(\sk,\hh_2) & \qif \hh = \hh_1 \sepcon \hh_2, \dom{\hh_1} = \{ \sk(e) \} \\
			\infty & \qif \sk(e) \notin \dom{\hh}.
		\end{cases}
	\end{align*}
	
\end{lemma}

\begin{proof}
First, assume $\sk(e) \in \dom{\hh}$ and let $\hh = \hh_1 \sepcon \hh_2$ such that $\dom{\hh_1} = \{ \sk(e) \}$. Then:

\begin{align*}
     & \ert{\FREE{e}}{\rta}(\sk,\hh) \\
\eeq &  (\ivalidpointer{e} \sepadd \rta)(\sk,\hh) 
\tag{\Cref{table:ert}} \\
\eeq & \min \{ \ivalidpointer{e}(\sk,\hh_1') \pplus \rta(\sk,\hh_2') ~|~ \hh = \hh_1' \sepcon \hh_2' \}
\tag{Def. of $\sepadd$} \\
\eeq & \ivalidpointer{e}(\sk,\hh_1) \pplus \rta(\sk,\hh_2)
\tag{$\ivalidpointer{e}(\sk,\hh_1') = 0$ iff $\hh_1' = \hh_1$}  \\
\eeq & \rta(\sk,\hh_2).
\end{align*}

Second, assume $\sk(e) \not\in \dom{\hh}$. Then:

\begin{align*}
     & \ert{\FREE{e}}{\rta}(\sk,\hh) \\
\eeq &  (\ivalidpointer{e} \sepadd \rta)(\sk,\hh) 
\tag{\Cref{table:ert}} \\
\eeq & \min \{ \ivalidpointer{e}(\sk,\hh_1') \pplus \rta(\sk,\hh_2') ~|~ \hh = \hh_1' \sepcon \hh_2' \}
\tag{Def. of $\sepadd$} \\
\eeq & \min \{ \infty \pplus \rta(\sk,\hh_2') ~|~ \hh = \hh_1' \sepcon \hh_2' \}
\tag{$\ivalidpointer{e}(\sk,\hh_1') = \infty$ since $\sk(e) \notin \dom{\hh}$} \\
\eeq & \infty \pplus \rta(\sk,\hh_2) \\
\eeq & \infty.
\end{align*}

\end{proof}

\subsection{Proof of \Cref{thm:ert-health}}
\label{proof:thm_cont}

Let $C$ be a program; $F = \{\rta_1 \preceq \rta_2 \preceq \ldots\}$ be an $\omega$-chain of runtimes; $\rta, \rtb$ be runtimes; and let $\rtc\in\T$ be a constant runtime.
	Then the following hold:

	\begin{enumerate}
		\item \emph{$\omega$-Continuity}: \qquad $\ert{C}{\sup F} \eeq \sup \ert{C}{F}$\vspace{.5em}
		
		\item \emph{Monotonicity:} \qquad $\rta \ppreceq \rtb$ \qimplies $\ert{C}{\rta} \ppreceq \ert{C}{\rtb}$ \vspace{.5em}
		
		\item \emph{Sub-additivity:} \qquad $\ert{C}{ \rta + \rtb}
		\ppreceq
		\ert{C}{\rta} + \ert{C}{\rtb}$\vspace{.5em}
		
		\item \emph{Constant propagation:} \qquad
		%
		$ \ert{C}{\rtc + \rta} \ppreceq \rtc + \ert{C}{\rta}$
	\end{enumerate}

\begin{proof}
We prove $\omega$-continuity by induction on $C$. Monotonicity then follows from $\omega$-continuity. After that, we prove sub-additivity and constant propagation by induction on $C$.

Now let $F = \{\rta_1,\rta_2,\ldots\}$ be an $\omega$-chain.
We proceed by structural induction on $C$.	

\paragraph{The case $C = \TICK{e}$}

\begin{align*}
&
\ert{\TICK{e}}{\sup F}
\\
\eeq & 
\emprun{e} \sepadd \sup F 
\tag{\Cref{table:ert}} \\
\eeq &
\lambda(\sk,\hh)\mydot \min \{ \emprun{e}(\sk,\hh_1) \pplus \sup \{ \rta(\sk,\hh_2) ~|~ \rta \in F \} ~|~ \hh = \hh_1 \sepcon \hh_2 \}
\tag{Def. of $\sepadd$} \\
\eeq &
\lambda(\sk,\hh)\mydot \min \{ \emprun{e}(\sk,\emptyheap) \pplus \sup \{ \rta(\sk,\hh_2) ~|~ \rta \in F \} ~|~ \hh = \hh_2 \}
\tag{$\emprun{e}(\sk,\hh_1) = \infty$ unless $\hh_1 = \emptyheap$} \\
\eeq &
\lambda(\sk,\hh)\mydot \emprun{e}(\sk,\emptyheap) \pplus \sup \{ \rta(\sk,\hh_2) ~|~ \rta \in F \} 
\tag{ $\min$ over singleton } \\
\eeq &
\lambda(\sk,\hh)\mydot \sup \{ \emprun{e}(\sk,\emptyheap) \pplus  \rta(\sk,\hh_2) ~|~ \rta \in F \} 
\tag{ $+$ is continuous } \\
\eeq & 
\lambda(\sk,\hh)\mydot \sup \{ \min \{ \emprun{e}(\sk,\hh_1) \pplus  \rta(\sk,\hh_2) ~|~ \hh = \hh_1 \sepcon \hh_2 \} ~|~ \rta \in F \} 
\tag{ $\min$ over singleton } \\
\eeq &
\lambda(\sk,\hh)\mydot \sup \{ \emprun{e} \sepadd \rta ~|~ \rta \in F \} 
\tag{Def. of $\sepadd$ } \\
\eeq &
\sup \ert{\TICK{e}}{F}~.
\tag{\Cref{table:ert}} 
\end{align*}

\paragraph{The case $C = \ASSIGN{x}{e}$}

\begin{align*}
& \ert{\ASSIGN{x}{e}}{\sup F} \\	
\eeq & (\sup F)\subst{x}{e} \\
\eeq & \sup F\subst{x}{e} \\
\eeq & \sup \ert{\ASSIGN{x}{e}}{F}~.
\end{align*}

\paragraph{The case $C = \ASSIGNH{x}{e}$}

Let $(\sk,\hh) \in \States$; we distinguish two cases.
First, assume $\sk(e) \in \dom{\hh}$:

\begin{align*}
& \ert{\HASSIGN{x}{e}}{\sup F}(\sk,\hh) \\		
\eeq & (\sup F)(\sk\statesubst{x}{\hh(\sk(e))}, \hh) 
\tag{\Cref{lem:lookup}} \\
\eeq & \sup \{ \rta(\sk\statesubst{x}{\hh(\sk(e))}, \hh) ~|~ \rta \in F \} \\
\eeq & \sup \{ \ert{\ASSIGNH{x}{e}}{\rta}(\sk,\hh) ~|~ \rta \in F \}
\tag{\Cref{lem:lookup}} \\ 
\eeq & \sup \ert{\ASSIGNH{x}{e}}{F}(\sk,\hh)~.
\end{align*}

Second, assume $\sk(e) \notin \dom{\hh}$:

\begin{align*}
& \ert{\ASSIGNH{x}{e}}{\sup F}(\sk,\hh) \\		
\eeq & \infty
\tag{\Cref{lem:lookup}} \\
\eeq & \sup \{ \infty ~|~ \rta \in F \} \\
\eeq & \sup \{ \ert{\ASSIGNH{x}{e}}{\rta}(\sk,\hh) ~|~ \rta \in F \}
\tag{\Cref{lem:lookup}} \\ 
\eeq & \sup \ert{\ASSIGNH{x}{e}}{F}(\sk,\hh)~.
\end{align*}

\paragraph{The case $C = \HASSIGN{e}{e'}$}

Let $(\sk,\hh) \in \States$; we distinguish two cases.
First, assume $\sk(e) \in \dom{\hh}$:

\begin{align*}
& \ert{\HASSIGN{e}{e'}}{\sup F}(\sk,\hh) \\		
\eeq & (\sup F)(\sk, \hh\statesubst{\sk(e)}{\sk(e')}) 
\tag{\Cref{lem:mut}} \\
\eeq & \sup \{ \rta(\sk, \hh\statesubst{\sk(e)}{\sk(e')}) ~|~ \rta \in F \} \\
\eeq & \sup \{ \ert{\HASSIGN{e}{e'}}{\rta}(\sk,\hh) ~|~ \rta \in F \}
\tag{\Cref{lem:mut}} \\ 
\eeq & \sup \ert{\HASSIGN{e}{e'}}{F}(\sk,\hh)~.
\end{align*}

Second, assume $\sk(e) \notin \dom{\hh}$:

\begin{align*}
& \ert{\HASSIGN{e}{e'}}{\sup F}(\sk,\hh) \\		
\eeq & \infty
\tag{\Cref{lem:mut}} \\
\eeq & \sup \{ \infty ~|~ \rta \in F \} \\
\eeq & \sup \{ \ert{\HASSIGN{e}{e'}}{\rta}(\sk,\hh) ~|~ \rta \in F \}
\tag{\Cref{lem:mut}} \\ 
\eeq & \sup \ert{\HASSIGN{e}{e'}}{F}(\sk,\hh)~.
\end{align*}

\paragraph{The case $C = \FREE{x}$}

\begin{align*}
&
\ert{\FREE{e}}{\sup F}
\\
\eeq & 
\ivalidpointer{e} \sepadd \sup F 
\tag{\Cref{table:ert}} \\
\eeq &
\lambda(\sk,\hh)\mydot \min \{ \ivalidpointer{e}(\sk,\hh_1) \pplus \sup \{ \rta(\sk,\hh_2) ~|~ \rta \in F \} ~|~ \hh = \hh_1 \sepcon \hh_2 \}
\tag{Def. of $\sepadd$} \\
\eeq &
\lambda(\sk,\hh)\mydot \min \{ \ivalidpointer{e}(\sk,\hh_1') \pplus \sup \{ \rta(\sk,\hh_2) ~|~ \rta \in F \} ~|~ \hh = \hh_1' \sepcon \hh_2 \}
\tag{$\ivalidpointer{e}(\sk,\hh_1') = \infty$ unless $\dom{\hh_1'} = \sk(x)$} \\
\eeq &
\lambda(\sk,\hh)\mydot \ivalidpointer{e}(\sk,\hh_1') \pplus \sup \{ \rta(\sk,\hh_2) ~|~ \rta \in F \} 
\tag{ $\min$ over singleton } \\
\eeq &
\lambda(\sk,\hh)\mydot \sup \{ \ivalidpointer{e}(\sk,\hh_1') \pplus  \rta(\sk,\hh_2) ~|~ \rta \in F \} 
\tag{ $+$ is continuous } \\
\eeq & 
\lambda(\sk,\hh)\mydot \sup \{ \min \{ \ivalidpointer{e}(\sk,\hh_1) \pplus  \rta(\sk,\hh_2) ~|~ \hh = \hh_1 \sepcon \hh_2 \} ~|~ \rta \in F \} 
\tag{ $\min$ over singleton } \\
\eeq &
\lambda(\sk,\hh)\mydot \sup \{ \ivalidpointer{e} \sepadd \rta ~|~ \rta \in F \} 
\tag{Def. of $\sepadd$ } \\
\eeq &
\sup \ert{\FREE{e}}{F}~.
\tag{\Cref{table:ert}} 
\end{align*}

\paragraph{The case $C = \ALLOC{x}{e}$}

\begin{align*}
& 
\ert{\ALLOC{x}{e}}{\sup F} 
\\
\eeq &
\Sup v\colon \left(\bigoplus_{i=1}^{e} \isingleton{v+i-1}{0}\right) \sepmon (\sup F)\subst{x}{v}
\tag{\Cref{table:ert}} \\
\eeq & 
\sup_v \lambda(\sk,\hh)\mydot \sup \left\{ (\sup F)\subst{x}{v}(\sk,\hh \sepcon \hh')\monus \left(\bigoplus_{i=1}^{e} \isingleton{v+i-1}{0}\right)(\sk,\hh') ~\middle|~ \hh \disjoint \hh' \right\}
\tag{Def. of $\sepmon$} \\
\eeq & 
\sup_v \lambda(\sk,\hh)\mydot \sup \left\{ \sup \{ \rta\subst{x}{v}(\sk,\hh \sepcon \hh') ~|~ \rta \in F \}\monus \left(\bigoplus_{i=1}^{e} \isingleton{v+i-1}{0}\right)(\sk,\hh') ~\middle|~ \hh \disjoint \hh' \right\} 
\tag{substitution is continuous} \\
\eeq &
\sup_v \lambda(\sk,\hh)\mydot \sup \left\{ \sup \{ \rta\subst{x}{v}(\sk,\hh \sepcon \hh') \monus \left(\bigoplus_{i=1}^{e} \isingleton{v+i-1}{0}\right)(\sk,\hh') ~|~ \rta \in F \} ~\middle|~ \hh \disjoint \hh' \right\} 
\tag{monus a constant is continuous} \\
\eeq &
\sup \{ \sup_v \lambda(\sk,\hh)\mydot \sup \left\{ \rta\subst{x}{v}(\sk,\hh \sepcon \hh') \monus \left(\bigoplus_{i=1}^{e} \isingleton{v+i-1}{0}\right)(\sk,\hh') ~\middle|~ \hh \disjoint \hh' \right\} ~|~ \rta \in F \} 
\tag{$\sup_u \sup_v \ldots \eeq \sup_v \sup_u \ldots$} \\
\eeq & \sup \{ \Sup v\colon \left(\bigoplus_{i=1}^{e} \isingleton{v+i-1}{0}\right) \sepmon f\subst{x}{v} ~|~ \rta \in F \}
\tag{Def. of $\sepmon$} \\
\eeq & 
\sup \ert{\ALLOC{x}{e}}{F}~.
\end{align*}

\paragraph{The case $C = \COMPOSE{C_1}{C_2}$}
We have
\begin{align*}
& \ert{\COMPOSE{C_1}{C_2}}{\sup F} 
\\
\eeq & 
\ert{C_1}{\ert{C_2}{\sup F}} 
\tag{\Cref{table:ert}} \\
\eeq &
\ert{C_1}{\sup \ert{C_2}{F}}
\tag{I.H.} \\ 
\eeq &
\sup \ert{C_1}{\ert{C_2}{F}}
\tag{I.H.} \\
\eeq &
\sup \ert{\COMPOSE{C_1}{C_2}}{F}~. 
\tag{\Cref{table:ert}} 
\end{align*}

\paragraph{The case $C = \ITE{\varphi}{C_1}{C_2}$}
We have
\begin{align*}
& \ert{\ITE{\varphi}{C_1}{C_2}}{\sup F} 
\\
\eeq & 
\iverson{\varphi} \cdot \ert{C_1}{\sup F} + \iverson{\neg \varphi} \cdot \ert{C_2}{\sup F} 
\tag{\Cref{table:ert}} \\
\eeq &
\iverson{\varphi} \cdot (\sup \ert{C_1}{F}) + \iverson{\neg \varphi} \cdot (\sup \ert{C_2}{F})
\tag{I.H.} \\
\eeq & 
(\sup \iverson{\varphi} \cdot \ert{C_1}{F}) + (\sup \iverson{\neg \varphi} \cdot \ert{C_2}{F})
\tag{multiplication with runtimes is continuous} \\
\eeq &
\sup (\iverson{\varphi} \cdot \ert{C_1}{F}) + \iverson{\neg \varphi} \cdot \ert{C_2}{F})
\\
\eeq & 
\sup \ert{\ITE{\varphi}{C_1}{C_2}}{F}~. 
\tag{\Cref{table:ert}} 
\end{align*}

\paragraph{The case $C = \PCHOICE{C_1}{p}{C_2}$}
We have
\begin{align*}
& \ert{\PCHOICE{C_1}{p}{C_2}}{\sup F} 
\\
\eeq & 
p \cdot \ert{C_1}{\sup F} + (1-p) \cdot \ert{C_2}{\sup F} 
\tag{\Cref{table:ert}} \\
\eeq &
p \cdot (\sup \ert{C_1}{F}) + (1-p) \cdot (\sup \ert{C_2}{F})
\tag{I.H.} \\
\eeq & 
(\sup p \cdot \ert{C_1}{F}) + (\sup (1-p) \cdot \ert{C_2}{F})
\tag{multiplication with runtimes is continuous} \\
\eeq &
\sup ~(p \cdot \ert{C_1}{F}) + (1-p) \cdot \ert{C_2}{F})
\tag{by monotone convergence, $\sup = \lim$}
\\
\eeq & 
\sup \ert{\PCHOICE{C_1}{p}{C_2}}{F}~. 
\tag{\Cref{table:ert}} 
\end{align*}

\paragraph{The case $C = \WHILEDO{\varphi}{C_1}$}
Recall from \Cref{table:ert} that 
\[
\ert{\WHILEDO{\varphi}{C_1}}{\rta}
\eeq 
\lfp I\mydot \underbrace{ \iverson{\neg\varphi} \cdot \rta + \iverson{\varphi} \cdot \ert{C_1}{I} }_{\eeq \Phi_{\rta}(I)}
\]
Our continuity proof for loops relies on three facts:

\begin{enumerate}
	\item $\Phi_{\sup F} = \sup \{ \Phi_{\rta} ~|~ \rta \in F \}$, which is straightforward;
	\item $\sup \{ \Phi_{\rta} ~|~ \rta \in F \}$ is continuous in ($\T \to \T$), since $\{\Phi_{\rta} ~|~ \rta \in F\}$ is an $\omega$-chain of continuous transformers (by I.H. $\ertC{C_1}$ is continuous) and continuous functions are closed under taking the supremum;
	\item Taking the least fixed point is itself continuous when restricted to the set of continuous transformers in $(\T \to \T)$, see \cite[Proposition 12]{DBLP:series/eatcs/Wechler92}.
\end{enumerate}

We then reason as follows:
\begin{align*}
& \ert{\WHILEDO{\varphi}{C_1}}{\sup F} \\
\eeq & \lfp \Phi_{\sup F} \tag{see above} \\
\eeq & \lfp (\sup \{ \Phi_{\rta} ~|~ \rta \in F \}) \tag{Fact 1} \\
\eeq & \sup \{ \lfp \Phi_{\rta} ~|~ \rta \in F \} \tag{Facts 1 and 2} \\
\eeq & \sup \{ \ert{\WHILEDO{\varphi}{C_1}}{\rta} ~|~ \rta \in F \} \tag{\Cref{table:ert}} \\
\eeq & \sup \ert{\WHILEDO{\varphi}{C_1}}{F}~.
\end{align*}
This concludes the proof of $\omega$-continuity. \\

\noindent
Let us now prove sub-additivity, i.e., 
\[
\ert{C}{\rta + \rtb}
\ppreceq
\ert{C}{\rta} + \ert{C}{\rtb}
\]
by induction on $C$. The cases $\ASSIGN{x}{e}$, $\ASSIGNH{x}{e}$, $\HASSIGN{e}{e'}$, and $\FREE{e}$ are immediate by \Cref{lem:lookup}, \Cref{lem:mut}, and \Cref{lem:free}. 

\paragraph{The case $C = \TICK{e}$}
We have
\begin{align*}
   &\ert{\TICK{e}}{\rta + \rtb} \\
   \eeq& \emprun{e} \sepadd (\rta + \rtb)  \\
   \eeq & e+ \rta + \rtb
   \tag{by \Cref{table:ert}} \\
   \ppreceq & (e+ \rta) + (e+ \rtb)
   \tag{since $e \geq 0$} \\
   \eeq & (\emprun{e} \sepadd  \rta) + (\emprun{e} \sepadd  \rtb)
   \\
   \eeq& \ert{\TICK{e}}{\rta} + \ert{\TICK{e}}{\rtb} 
   \tag{by \Cref{table:ert}}
\end{align*}

\paragraph{The case $C = \ALLOC{x}{e}$}
We have
\begin{align*}
&\ert{\ALLOC{x}{e}}{\rta + \rtb}(\sk,\hh) \\
\eeq & (\Sup\, v \colon \bigl(\bigoplus_{i=1}^{e } \isingleton{v+i-1}{0}\bigr) \sepmon (\rta + \rtb)\subst{x}{v})(\sk,\hh)
\tag{by \Cref{table:ert}} \\
\eeq & \sup \setcomp{(\bigl(\bigoplus_{i=1}^{e} \isingleton{v+i-1}{0}\bigr)\sepmon (\rta + \rtb)\subst{x}{v})(\sk, \hh)}{v \in \Vals}
\\
\eeq & \sup \setcomp{\setcomp{(\rta + \rtb)\subst{x}{v})(\sk, \hh \sepcon\hh')
		\monus
		\bigl(\bigoplus_{i=1}^{e} \isingleton{v+i-1}{0}\bigr)(\sk, \hh')}{\hh' \disjoint \hh}
}{v \in \Vals}
\tag{by \Cref{def:sepmon}}
\\
%
%
\eeq & \sup \setcomp{\setcomp{(\rta + \rtb)\subst{x}{v})(\sk, \hh \sepcon\hh')
	}{\hh \disjoint \hh' =\singleheap{v}{0}\sepcon\ldots\sepcon\singleheap{v + \sk(e) -1}{0}}
}{v \in \Vals}
\\
\eeq & \sup_v \sup_{\substack{\hh' \eeq \ldots\\\hh'\disjoint \hh}} (\rta + \rtb)\subst{x}{v}(\sk, \hh \sepcon\hh')
\tag{rewrite} \\
\eeq & \sup_v \sup_{\substack{\hh' \eeq \ldots\\\hh'\disjoint \hh}}  \rta\subst{x}{v}(\sk, \hh \sepcon\hh') + \rtb\subst{x}{v}(\sk, \hh \sepcon\hh') \\
\lleq & \sup_v \sup_{\substack{\hh' \eeq \ldots\\\hh'\disjoint \hh}} \rta\subst{x}{v}(\sk, \hh \sepcon\hh')  \quad+ \quad \sup_v \sup_{\substack{\hh' \eeq \ldots\\\hh'\disjoint \hh}} \rtb\subst{x}{v}(\sk, \hh \sepcon\hh')
\tag{standard property of suprema} \\
\eeq & \ert{\ALLOC{x}{e}}{\rta }(\sk, \hh)  +  \ert{\ALLOC{x}{e}}{\rtb}(\sk, \hh)~.
\tag{analogous to above reasoning}
\end{align*}
The remaining cases are analogous to \cite[Theorem 1]{DBLP:conf/esop/KaminskiKMO16}.

Let us now prove constant propagation by induction on $C$. The cases $\ASSIGN{x}{e}$, $\ASSIGNH{x}{e}$, $\HASSIGN{e}{e'}$, and $\FREE{e}$ are immediate by \Cref{lem:lookup}, \Cref{lem:mut}, and \Cref{lem:free}.

\paragraph{The case $C=\TICK{e}$}
We have
\begin{align*}
   &\ert{\TICK{e}}{\rtc + \rta} \\
   \eeq& \emprun{e} \sepadd (\rtc + \rta)
   \tag{by \Cref{table:ert}} \\
   \eeq& e + \rtc+ \rta \\
   \eeq & \rtc+ e+\rta \\
   \eeq & \rtc + (\emprun{e}\sepadd \rta) \\
   \eeq & \rtc+ \ert{\TICK{e}}{\rta} ~.
   \tag{by \Cref{table:ert}}
\end{align*}

\paragraph{The case $C=\ALLOC{x}{e}$}
We have
\begin{align*}
&\ert{\ALLOC{x}{e}}{\rtc + \rta}(\sk,\hh) \\
\eeq & (\Sup\, v \colon \bigl(\bigoplus_{i=1}^{e } \isingleton{v+i-1}{0}\bigr) \sepmon (\rtc + \rta)\subst{x}{v})(\sk,\hh)
\tag{by \Cref{table:ert}} \\
\eeq & \sup \setcomp{(\bigl(\bigoplus_{i=1}^{e} \isingleton{v+i-1}{0}\bigr)\sepmon (\rtc + \rta)\subst{x}{v})(\sk, \hh)}{v \in \Vals}
\\
\eeq & \sup \setcomp{\setcomp{(\rtc + \rta)\subst{x}{v})(\sk, \hh \sepcon\hh')
		\monus
		\bigl(\bigoplus_{i=1}^{e} \isingleton{v+i-1}{0}\bigr)(\sk, \hh')}{\hh' \disjoint \hh}
}{v \in \Vals}
\tag{by \Cref{def:sepmon}}
\\
\eeq & \sup \setcomp{\setcomp{(\rtc + \rta)\subst{x}{v})(\sk, \hh \sepcon\hh')
	}{\hh \disjoint \hh' =\singleheap{v}{0}\sepcon\ldots\sepcon\singleheap{v + \sk(e) -1}{0}}
}{v \in \Vals}
\\
\eeq & \sup_v \sup_{\substack{\hh' \eeq \ldots\\\hh'\disjoint \hh}} (\rtc + \rta)\subst{x}{v}(\sk, \hh \sepcon\hh')
\tag{rewrite} \\
\eeq & \sup_v \sup_{\substack{\hh' \eeq \ldots\\\hh'\disjoint \hh}}  \rtc + \rta\subst{x}{v}(\sk, \hh \sepcon\hh') 
\tag{since $\rtc$ is constant }\\
\eeq & \rtc \quad+ \quad \sup_v \sup_{\substack{\hh' \eeq \ldots\\\hh'\disjoint \hh}} \rta\subst{x}{v}(\sk, \hh \sepcon\hh')
\tag{standard property of suprema} \\
\eeq &\rtc  +  \ert{\ALLOC{x}{e}}{\rta}(\sk, \hh)~.
\tag{analogous to above reasoning}
\end{align*}

The cases $C = \COMPOSE{C_1}{C_2}$,  $C = \ITE{\varphi}{C_1}{C_2}$, and $C = \PCHOICE{C_1}{p}{C_2}$ are immediate by the induction hypothesis.

\paragraph{The case $C = \WHILEDO{\varphi}{C_1}$}
Let $\charfun{\rtc+\rta}$ be the $\ertsymbol$-characteristic function of $C$ w.r.t.\ $\rtc+\rta$ and let $\charfun{\rta}$ be the $\ertsymbol$-characteristic function of $C$ w.r.t.\ $\rta$. Using the induction hypothesis, we show by means of an inner induction that for all $n\in \Nats$, $\charfun{\rtc+\rta}^n(0) \leq \rtc+ \charfun{\rta}^n(0)$. Since $\ertsymbol$ is $\omega$-continuous, we then get
 \begin{align*}
    & \ert{\WHILEDO{\varphi}{C_1}}{\rtc+\rta} \\
    \eeq & \sup_{n\in\Nats} \charfun{\rtc+\rta}^n(0) \\
     \ppreceq & \sup_{n\in\Nats} \rtc+\charfun{\rta}^n(0) \\
     \eeq& \rtc+\ert{\WHILEDO{\varphi}{C_1}}{\rta} ~.
 \end{align*}
 
 It is left to show $\charfun{\rtc+\rta}^n(0) \leq \rtc+ \charfun{\rta}^n(0)$. The base case $n=0$ is immediate. For the induction step, consider the following: 
 \begin{align*}
      &\charfun{\rtc+\rta}^{n+1}(0)  \\
      \eeq & \iverson{\varphi} \cdot \ert{C'}{\charfun{\rtc+\rta}^{n}(0)} + \iverson{\neg\varphi} \cdot (\rtc + \rta ) \\
      \ppreceq& \iverson{\varphi} \cdot \ert{C'}{\rtc+ \charfun{\rta}^{n}(0)} + \iverson{\neg\varphi}\cdot (\rtc + \rta ) 
      \tag{by inner I.H. and monotonicity of $\ertsymbol$} \\
      \ppreceq& \iverson{\varphi} \cdot \left(\rtc+\ert{C'}{ \charfun{\rta}^{n}(0)}\right)) + \iverson{\neg\varphi}\cdot (\rtc + \rta ) 
      \tag{by outer I.H. } \\
       \eeq&\rtc+ \iverson{\varphi} \cdot \ert{C'}{ \charfun{\rta}^{n}(0)} + \iverson{\neg\varphi}\cdot  \rta 
      \\
      \eeq &\rtc+ \charfun{\rta}^{n+1}(0)~.
 \end{align*}
\end{proof}

\subsection{Proof of \Cref{thm:ert_frame}}
\label{proof:ert_frame}
For every $C \in \hpgcl$ and runtimes $\rta,\rtb$ with $\modC{C}\cap\Vars(\rtb)=\emptyset$, we have
	\begin{align*}
	\ert{C}{\rta \sepadd \rtb} \ppreceq \ert{C}{\rta} \sepadd \rtb 
	\end{align*}

\begin{proof}
	It suffices to prove the frame rule for the case $C = \ALLOC{x}{e}$ since \citet{haslbeckPhd} has proven the frame rule for the remaining program statements by induction on the structure of $\hpgcl$. We reason as follows:
	\begin{align*}
	   &\ert{\ALLOC{x}{e}}{\rta \sepadd \rtb}(\sk,\hh) \\
	   \eeq & (\Sup\, v \colon \bigl(\bigoplus_{i=1}^{e} \isingleton{v+i-1}{0}\bigr) \sepmon (\rta \sepadd \rtb)\subst{x}{v})(\sk,\hh)
	   \tag{by \Cref{table:ert}} \\
	   \eeq & \sup \setcomp{(\bigl(\bigoplus_{i=1}^{e} \isingleton{v+i-1}{0}\bigr)\sepmon (\rta \sepadd \rtb)\subst{x}{v})(\sk, \hh)}{v \in \Vals}
	    \\
	    \eeq & \sup \setcomp{\setcomp{(\rta \sepadd\rtb)\subst{x}{v})(\sk, \hh \sepcon\hh')
	    	 \monus
	    	 \bigl(\bigoplus_{i=1}^{e} \isingleton{v+i-1}{0}\bigr)(\sk, \hh')}{\hh' \disjoint \hh}
    	 }{v \in \Vals}
     \tag{by \Cref{def:sepmon}}
	    \\
	    \eeq & \sup \setcomp{\setcomp{(\rta\subst{x}{v} \sepadd\rtb))(\sk, \hh \sepcon\hh')
	    		\monus
	    		\bigl(\bigoplus_{i=1}^{e} \isingleton{v+i-1}{0}\bigr)(\sk, \hh')}{\hh' \disjoint \hh}
	    }{v \in \Vals}
    \tag{$x \not \in \Vars(\rtb)$}
	    \\
	    \eeq & \sup \setcomp{\setcomp{(\rta\subst{x}{v} \sepadd\rtb))(\sk, \hh \sepcon\hh')
	    		}{\hh \disjoint \hh' =\singleheap{v}{0}\sepcon\ldots\sepcon\singleheap{v + \sk(e) -1}{0}}
	    }{v \in \Vals}
	    \\
	    \eeq & \sup_v \sup_{\substack{\hh' \eeq \ldots\\\hh'\disjoint \hh}} (\rta\subst{x}{v} \sepadd\rtb))(\sk, \hh \sepcon\hh')
	    \tag{rewrite} \\
	    \eeq & \sup_v \sup_{\substack{\hh' \eeq \ldots\\\hh'\disjoint \hh}}
	    \min_{\hh_1 \sepcon\hh_2 = \hh \sepcon \hh'}
	     \rta\subst{x}{v}(\sk,\hh_1) + \rtb(\sk, \hh_2)
    \tag{by \Cref{def:sepadd}} \\
    \lleq & \sup_v \sup_{\substack{\hh' \eeq \ldots\\\hh'\disjoint \hh}}
    \min_{\hh_1 \sepcon\hh_2 = \hh }
    \rta\subst{x}{v}(\sk,\hh_1 \sepcon \hh') + \rtb(\sk, \hh_2)
    \tag{restrict minimum} \\
    \lleq & \min_{\hh_1 \sepcon\hh_2 = \hh } \sup_v \sup_{\substack{\hh' \eeq \ldots\\\hh'\disjoint \hh}}
    \rta\subst{x}{v}(\sk,\hh_1 \sepcon \hh') + \rtb(\sk, \hh_2)
    \tag{$\sup \inf \leq \inf \sup$} \\
    \eeq & \min_{\hh_1 \sepcon\hh_2 = \hh } (\Sup\, v \colon \bigl(\bigoplus_{i=1}^{e} \isingleton{v+i-1}{0}\bigr) \sepmon \rta\subst{x}{v})(\sk, \hh_1) + \rtb(\sk, \hh_2)
    \tag{analogous to above reasoning} \\
    \eeq&  \big((\Sup\, v \colon \bigl(\bigoplus_{i=1}^{e} \isingleton{v+i-1}{0}\bigr) \sepmon \rta\subst{x}{v}) \sepadd \rtb\big)(\sk,\hh)
    \tag{by \Cref{def:sepadd}} \\
    \eeq&  \big(\ert{\ALLOC{x}{e}}{\rta} \sepadd \rtb\big)(\sk,\hh)
    \tag{by \Cref{table:ert}} 
	\end{align*}
\end{proof}

\subsection{Proof of \Cref{thm:local_rules}}
\label{proof:local_rules}

Let $C \in \hpgcl$. Then:

\begin{enumerate}
	\item \textnormal{(mut)}: $\ert{\HASSIGN{e}{e'}}{\isingleton{e}{e'}} \preceq \ivalidpointer{e}$
	\item \textnormal{(lkp)}: 
	\[
	\ert{\ASSIGNH{x}{e}}{\pureemp{x=z} \sepadd \isingleton{e\subst{x}{y}}{z}} \ppreceq \pureemp{x=y} \sepadd \isingleton{e}{z}~,
	\]
	%
	\item \textnormal{(alc)}:  if $x$ does not occur in $e$, then
	\[\ert{\ALLOC{x}{e}}{\bigoplus_{i=1}^{e} \isingleton{x+i-1}{0}} \ppreceq \iemp~.
	\]
	\item \textnormal{(aux)}: For all $\rta, \rtb \in \T$ and all $y \in \Vars$ not occurring in $C$, 
	\begin{align*}
	\ert{C}{\rta} \lleq \rtb \qquad \textnormal{implies} \qquad \ert{C}{\Inf y\colon \rta} \ppreceq \Inf y \colon \rtb~.
	\end{align*}
\end{enumerate}

Let $(\sk,\hh)\in \States$.

\subsubsection{The rule (mut)}
For the rule (mut), i.e.,
\[
    \ert{\HASSIGN{e}{e'}}{\isingleton{e}{e'}} \ppreceq \ivalidpointer{e}~,
\]
consider the following: If $\sk(e)\not\in\dom{\hh}$, then
\begin{align*}
     &\ert{\HASSIGN{e}{e'}}{\isingleton{e}{e'}}(\sk,\hh) \\
     \eeq& \infty \tag{by \Cref{lem:mut}} \\
     \lleq &\ivalidpointer{e}(\sk,\hh)~. \tag{$\sk,\hh \not \models \slvalidpointer{e}$, Def. of \iiverson{.}}
\end{align*}
If $\sk(e)\in\dom{\hh}$, then
\begin{align*}
&\ert{\HASSIGN{e}{e'}}{\isingleton{e}{e'}}(\sk,\hh) \\
\eeq& \isingleton{e}{e'}(\sk,\hh\statesubst{\sk(e)}{\sk(e')}) \tag{by \Cref{lem:mut}} \\
\lleq & \ivalidpointer{e}(\sk,\hh) ~. \tag{$\dom{h} = \dom{\hh\statesubst{\sk(e)}{\sk(e')}}$} \\
\end{align*}
\qed

\subsubsection{The rule (lkp)}
For the rule (lkp), i.e., 
\[
    \ert{\ASSIGNH{x}{e}}{\pureemp{x=z} \sepadd \isingleton{e\subst{x}{y}}{z}} \ppreceq \pureemp{x=y} \sepadd \isingleton{e}{z}~,
\]
consider the following:
If $\hh \neq \{ \sk(e) \mapsto \sk(z) \}$ or $\sk(x) \neq \sk(y)$, then the right-hans side of the above inequality evaluates to $\infty$ and there is nothing to show.

%
If $\hh = \{ \sk(e) \mapsto \sk(z) \}$ and $\sk(x) = \sk(y)$, then
\begin{align*}
&\ert{\ASSIGNH{x}{e}}{\pureemp{x=z} \sepadd \isingleton{e\subst{x}{y}}{z}}(\sk,\hh) \\
\eeq& (\pureemp{x=z} \sepadd \isingleton{e\subst{x}{y}}{z})(\sk\statesubst{x}{\hh(\sk(e))}, \hh)  \tag{by \Cref{lem:lookup}}\\
\eeq& \isingleton{e\subst{x}{y}}{z}(\sk\statesubst{x}{\hh(\sk(e))}, \hh) \tag{$\sk(z) = \hh(\sk(e))$} \\
\eeq& 0 \tag{$\sk\statesubst{x}{\hh(\sk(e))}, \hh \models \slsingleton{e\subst{x}{y}}{z}$ using $\sk(x) = \sk(y)$} \\
\lleq& (\pureemp{x=y} \sepadd \isingleton{e}{z})(\sk, \hh)  
       \tag{assumption, Def. of \iiverson{.}}
\end{align*}
\qed

\subsubsection{The rule (alc)}
For (alc), i.e., if $x$ does not occur in $e$, then
\[
   \ert{\ALLOC{x}{e}}{\bigoplus_{i=1}^{e} \isingleton{x+i-1}{0}} 
   \ppreceq 
   \iemp~,
\]
consider the following:
\begin{align*}
& 
\ert{\ALLOC{x}{e}}{\bigoplus_{i=1}^{e} \isingleton{x+i-1}{0}}(\sk,\hh)
\\
\eeq & 
\left(\Sup\, v \colon \bigl(\bigoplus_{i=1}^{e} \isingleton{v+i-1}{0}\bigr) \sepmon \left(\bigoplus_{i=1}^{e} \isingleton{x+i-1}{0}\right)\subst{x}{v}\right)(\sk,\hh)
\tag{\Cref{table:ert}} \\
\eeq & 
\left(\Sup\, v \colon \bigl(\bigoplus_{i=1}^{e} \isingleton{v+i-1}{0}\bigr) \sepmon \bigoplus_{i=1}^{e} \isingleton{v+i-1}{0}\right)(\sk,\hh)
\tag{$x$ does not occur in $e$}
\\
\eeq & \left(\sup_{v} \sup \left\{ \left(\bigoplus_{i=1}^{e}\isingleton{v+i-1}{0}\right)(\sk,\hh \sepcon \hh') \monus \left(\bigoplus_{i=1}^{e}\isingleton{v+i-1}{0}\right)(\sk,\hh') ~|~ \hh \disjoint \hh' \right\}\right)(\sk,\hh)
\end{align*}
We distinguish two cases: $\hh = \emptyheap$ and $\hh \neq \emptyheap$.
First, assume $\hh = \emptyheap$. Then:
\begin{align*}
     & \ldots \\
\eeq & \sup_{v} \sup \left\{ \left(\bigoplus_{i=1}^{e}\isingleton{v+i-1}{0}\right)(\sk,\hh') \monus \left(\bigoplus_{i=1}^{e}\isingleton{v+i-1}{0}\right)(\sk,\hh') ~|~ \emptyheap \disjoint \hh' \right\}
\\
\eeq & \sup_{v} \lambda(\sk,\hh)\mydot \sup \left\{ 0 ~|~ \emptyheap \disjoint \hh' \right\}
\\
\eeq & 0 \\
\ppreceq & \iemp(\sk,\emptyheap).
\end{align*}
Second, assume $\hh \neq \emptyheap$. Then:
\begin{align*}
     & \ldots \\
\eeq & 
\left(\sup_{v} \sup \left\{ \left(\bigoplus_{i=1}^{e}\isingleton{v+i-1}{0}\right)(\sk,\hh \sepcon \hh') \monus \left(\bigoplus_{i=1}^{e}\isingleton{v+i-1}{0}\right)(\sk,\hh') ~|~ \hh \disjoint \hh' \right\}\right)(\sk,\hh)
\\
\eeq & \infty 
\tag{choose $v$ and $\hh'\disjoint \hh$ such that $(\bigoplus_{i=1}^{e}\isingleton{v+i-1}{0})(\sk,\hh') = 0$} \\
\eeq &  \iemp(\sk,\hh) 
\tag{by assumption}\\
\ppreceq & \iemp(\sk,\hh)~.
\end{align*}
\qed
%

\subsubsection{The rule (aux)}
For (aux), i.e., for all $\rta, \rtb \in \T$ and all $y \in \Vars$ not occurring in $C$, 
\[
\ert{C}{\rta} \preceq \rtb \quad \textnormal{implies} \quad \ert{C}{\Inf y\colon \rta} \preceq \Inf y \colon \rtb~,
\]
we show that if $y$ does not occur in $C$, then
\begin{align}
\label{eqn:proof:local_rules_1}
    \ert{C}{\Inf y \colon \rta} \ppreceq \Inf y \colon \ert{C}{\rta}~.
\end{align}
This gives us
\begin{align*}
   & \ert{C}{\rta} \ppreceq \rtb \\
   \text{implies} \quad & \Inf y \colon \ert{C}{\rta}  \ppreceq \Inf y \colon \rtb \\
   \text{implies} \quad & \ert{C}{\Inf y \colon \rta} \ppreceq  \Inf y \colon \rtb~.
   \tag{by Inequality \ref{eqn:proof:local_rules_1}}
\end{align*}
It remains to prove Inequality \ref{eqn:proof:local_rules_1} by induction on $C$.

\emph{The case $C = \TICK{e}$.}
\begin{align*}
& \ert{\TICK{e}}{\Inf y \colon \rta} 
\\
\eeq & 
\emprun{e} \sepadd \Inf y \colon \rta 
\tag{\Cref{table:ert}} \\
\eeq & \Inf y \colon \emprun{e} \sepadd \rta 
\tag{\Cref{lem:inf-prenex}}
\\
\ppreceq & 
\Inf y \colon \ert{\TICK{e}}{\rta}.
\tag{\Cref{table:ert}}
\end{align*}

\emph{The case $C = \ASSIGN{x}{e}$.}
\begin{align*}
& \ert{\ASSIGN{x}{e}}{\Inf y \colon \rta} 
\\
\eeq & 
(\Inf y \colon \rta)\subst{x}{e}
\tag{\Cref{table:ert}} \\
\eeq & \Inf y \colon \rta\subst{x}{e}
\tag{$x \notin \Vars(C)$}
\\
\lleq & 
\Inf y \colon \ert{\ASSIGN{x}{e}}{\rta}.
\tag{\Cref{table:ert}}
\end{align*}

\emph{The case $C = \ASSIGNH{x}{e}$.}
Let $(\sk,\hh) \in \States$. If $\sk(e) \in \dom{\hh}$, consider the following:
\begin{align*}
& \ert{\ASSIGNH{x}{e}}{\Inf y \colon \rta}(\sk,\hh) 
\\
\eeq & 
\left(\Inf y \colon \rta\right)(\sk\statesubst{x}{\hh(\sk(e))}, \hh)
\tag{\Cref{lem:lookup}}
\\
\eeq &
\inf_v \rta(\sk\statesubst{x}{\hh(\sk(e))}\statesubst{y}{v}, \hh)
\\
\eeq &
\inf_v \ert{\ASSIGNH{x}{e}}{\rta}(\sk\statesubst{y}{v},\hh) 
\tag{\Cref{lem:lookup}}
\\
\eeq & 
\Inf y\colon \ert{\ASSIGNH{x}{e}}{\rta}(\sk,\hh).
\end{align*}
If $\sk(e) \notin \dom{\hh}$, consider the following:
\begin{align*}
& \ert{\ASSIGNH{x}{e}}{\Inf y \colon \rta}(\sk,\hh) 
\\
\eeq & 
\infty
\tag{\Cref{lem:lookup}}
\\
\eeq &
\inf_v \infty 
\\
\eeq & 
\inf_v \ert{\ASSIGNH{x}{e}}{\rta}(\sk\statesubst{y}{v},\hh)
\tag{\Cref{lem:lookup}}
\\
\eeq & 
\Inf y\colon \ert{\ASSIGNH{x}{e}}{\rta}(\sk,\hh).
\end{align*}

\emph{The case $C = \HASSIGN{e}{e'}$.}
Let $(\sk,\hh) \in \States$. If $\sk(e) \in \dom{\hh}$, consider the following:
\begin{align*}
& \ert{\HASSIGN{e}{e'}}{\Inf y \colon \rta}(\sk,\hh) 
\\
\eeq & 
\left(\Inf y \colon \rta\right)(\sk, \hh\statesubst{\sk(e)}{\sk(e')})
\tag{\Cref{lem:mut}}
\\
\eeq &
\inf_v \rta(\sk\statesubst{y}{v}, \hh\statesubst{\sk(e)}{\sk(e')})
\\
\eeq &
\inf_v \ert{\HASSIGN{e}{e'}}{\rta}(\sk\statesubst{y}{v},\hh) 
\tag{\Cref{lem:mut}}
\\
\eeq & 
\Inf y\colon \ert{\HASSIGN{e}{e'}}{\rta}(\sk,\hh).
\end{align*}
If $\sk(e) \notin \dom{\hh}$, consider the following:
\begin{align*}
& \ert{\HASSIGN{e}{e'}}{\Inf y \colon \rta}(\sk,\hh) 
\\
\eeq & 
\infty
\tag{\Cref{lem:mut}}
\\
\eeq &
\inf_v \infty 
\\
\eeq & 
\inf_v \ert{\HASSIGN{e}{e'}}{\rta}(\sk\statesubst{y}{v},\hh)
\tag{\Cref{lem:mut}}
\\
\eeq & 
\Inf y\colon \ert{\HASSIGN{e}{e'}}{\rta}(\sk,\hh).
\end{align*}

\emph{The case $C = \FREE{e}$.}
\begin{align*}
& 
\ert{\FREE{e}}{\Inf y \colon \rta}
\\
\eeq &
\ivalidpointer{e} \sepadd \Inf y \colon \rta
\tag{\Cref{table:ert}} \\
\eeq &
\Inf y\colon \ivalidpointer{e} \sepadd \rta
\tag{\Cref{lem:inf-prenex}} \\
\eeq &
\Inf y \colon \ert{\FREE{e}}{\rta}.
\tag{\Cref{table:ert}}
\end{align*}

\emph{The case $C = \ALLOC{x}{e}$.}
\begin{align*}
& 
\ert{\ALLOC{x}{e}}{\Inf y \colon \rta}
\\
\eeq &	
\Sup v\colon \left(\bigoplus_{i=1}^{e} \isingleton{v+i-1}{0}\right) \sepmon \Inf y \colon \rta\subst{x}{v}
\tag{\Cref{table:ert}} \\
\eeq &
\lambda(\sk,\hh) \mydot \sup_v \sup \left\{ (\Inf y \colon \rta\subst{x}{v})(\sk,\hh \sepcon \hh') \monus \left(\bigoplus_{i=1}^{e} \isingleton{v+i-1}{0}\right)(\sk,\hh') ~\middle|~ \hh \disjoint \hh' \right\}
\tag{Def. of $\sepmon$} \\
\eeq &
\lambda(\sk,\hh) \mydot \sup_v \sup \left\{ \inf_u \rta\subst{x}{v}(\sk\statesubst{y}{u},\hh \sepcon \hh') \monus \left(\bigoplus_{i=1}^{e} \isingleton{v+i-1}{0}\right)(\sk\statesubst{y}{u},\hh') ~\middle|~ \hh \disjoint \hh' \right\}
\tag{$y \notin \Vars(C)$, $y \neq i$} \\
\ppreceq &
\lambda(\sk,\hh) \mydot \inf_u \sup_v \sup \left\{  \rta\subst{x}{v}(\sk\statesubst{y}{u},\hh \sepcon \hh') \monus \left(\bigoplus_{i=1}^{e} \isingleton{v+i-1}{0}\right)(\sk\statesubst{y}{u},\hh') ~\middle|~ \hh \disjoint \hh' \right\}
\tag{$\sup \sup \inf \ldots \leq \inf \sup \sup \ldots$} \\
\eeq &
\lambda(\sk,\hh) \mydot \inf_u \sup_v \left(\left(\bigoplus_{i=1}^{e} \isingleton{v+i-1}{0}\right) \sepmon \rta\subst{x}{v}\right)(\sk\statesubst{y}{u},\hh)
\tag{Def. of $\sepmon$} \\
\eeq &
\Inf y~ \Sup v\colon \left(\bigoplus_{i=1}^{e} \isingleton{v+i-1}{0}\right) \sepmon \rta\subst{x}{v}
\\
\eeq &
\Inf y\colon \ert{\ALLOC{x}{e}}{\rta}~.
\tag{\Cref{table:ert}} 
\end{align*}

\emph{The case $C = \COMPOSE{C_1}{C_2}$.}
\begin{align*}
& \ert{\COMPOSE{C_1}{C_2}}{\Inf y \colon \rta} 
\\
\eeq & \ert{C_1}{\ert{C_2}{\Inf y \colon \rta}}
\tag{\Cref{table:ert}}
\\
\ppreceq & \ert{C_1}{\Inf y \colon \ert{C_2}{\rta}}
\tag{I.H.}
\\
\ppreceq & \Inf y \colon \ert{C_1}{\ert{C_2}{\rta}}
\tag{I.H.}
\\
\eeq & \Inf y \colon \ert{\COMPOSE{C_1}{C_2}}{\rta}.
\tag{\Cref{table:ert}}
\end{align*}

\emph{The case $C = \ITE{\varphi}{C_1}{C_2}$.}
\begin{align*}
& \ert{\ITE{\varphi}{C_1}{C_2}}{\Inf y \colon \rta} 
\\
\eeq & \iverson{\varphi} \cdot \ert{C_1}{\Inf y \colon \rta}
\pplus \iverson{\neg \varphi} \cdot \ert{C_2}{\Inf y \colon \rta}
\tag{\Cref{table:ert}}
\\
\ppreceq & \iverson{\varphi} \cdot \Inf y \colon \ert{C_1}{\rta}
\pplus \iverson{\neg \varphi} \cdot \Inf y \colon  \ert{C_2}{\rta}
\tag{I.H.}
\\
\eeq & \Inf y \colon \iverson{\varphi} \cdot \ert{C_1}{\rta}
\pplus \Inf y \colon \iverson{\neg \varphi} \cdot \ert{C_2}{\rta} 
\tag{$y \notin \Vars(\varphi)$}
\\
\ppreceq &  \Inf y \colon (\iverson{\varphi} \cdot \ert{C_1}{\rta}
         \pplus \iverson{\neg \varphi} \cdot \ert{C_2}{\rta}) \\
\eeq & \Inf y \colon \ert{\ITE{\varphi}{C_1}{C_2}}{\rta}.
\tag{\Cref{table:ert}}
\end{align*}

\emph{The case $C = \PCHOICE{C_1}{p}{C_2}$.}
\begin{align*}
& \ert{\PCHOICE{C_1}{p}{C_2}}{\Inf y \colon \rta} 
\\
\eeq & p \cdot \ert{C_1}{\Inf y \colon \rta}
\pplus (1-p) \cdot \ert{C_2}{\Inf y \colon \rta}
\tag{\Cref{table:ert}}
\\
\ppreceq & p \cdot \Inf y \colon \ert{C_1}{\rta}
\pplus (1-p) \cdot \Inf y \colon  \ert{C_2}{\rta}
\tag{I.H.}
\\
\eeq & \Inf y \colon p \cdot \ert{C_1}{\rta}
\pplus \Inf y \colon (1-p) \cdot \ert{C_2}{\rta} 
\tag{$y \notin \Vars(p)$}
\\
\ppreceq &  \Inf y \colon (p \cdot \ert{C_1}{\rta}
         \pplus (1-p) \cdot \ert{C_2}{\rta}) \\
\eeq & \Inf y \colon \ert{\PCHOICE{C_1}{p}{C_2}}{\rta}.
\tag{\Cref{table:ert}}
\end{align*}

\emph{The case $C = \WHILEDO{\varphi}{C_1}$.}
Let $I = \Inf y\colon \ert{\WHILEDO{\varphi}{C_1}}{\rta}$.
Then, consider the following:
\begin{align*}
& \iverson{\varphi} \cdot (\Inf y\colon \rta) \pplus \iverson{\neg \varphi} \cdot \ert{C_1}{I} 
\\
\ppreceq &
\iverson{\varphi} \cdot (\Inf y\colon \rta) \pplus \iverson{\neg \varphi} \cdot \Inf y\colon \ert{C_1}{\ert{\WHILEDO{\varphi}{C_1}}{\rta}} 
\tag{Def. of $I$, I.H.} \\
\eeq & \Inf y\colon \left(\iverson{\varphi} \cdot \rta \pplus \iverson{\neg \varphi} \cdot \ert{C_1}{\ert{\WHILEDO{\varphi}{C_1}}{\rta}} \right) 
\tag{$y \notin \Vars(C)$} \\
\eeq & \Inf y\colon \ert{\WHILEDO{\varphi}{C_1}}{\rta}
\tag{\Cref{table:ert}} \\
\eeq & I.
\end{align*}
Hence, by Park induction, we have
\[
  \ert{\WHILEDO{\varphi}{C_1}}{\Inf y\colon \rta}
  \ppreceq I \eeq \Inf y\colon \ert{\WHILEDO{\varphi}{C_1}}{\rta}~.
\]
\qed

\subsection{Program Annotations}
\label{app:annotations}

We use program annotations to apply the $\ertsymbol$ calculus on source-code level. See \Cref{fig:ert-annotations_monotonicity} for annotations representing the rules given in \Cref{table:ert} and for exploiting monotonicity, which correpsonds to the rule of consequence in classical Hoare logic. For annotations using the frame rule (\Cref{thm:ert_frame}), see \Cref{fig:ert-annotations_framing}. For annotations for loops using invariants, see \Cref{fig:ert-loop-annotations}. Due to the backward-moving nature of the $\ertsymbol$-calculus, it is more inituitive to read these annotations from bottom to top.

\begin{figure}[t]
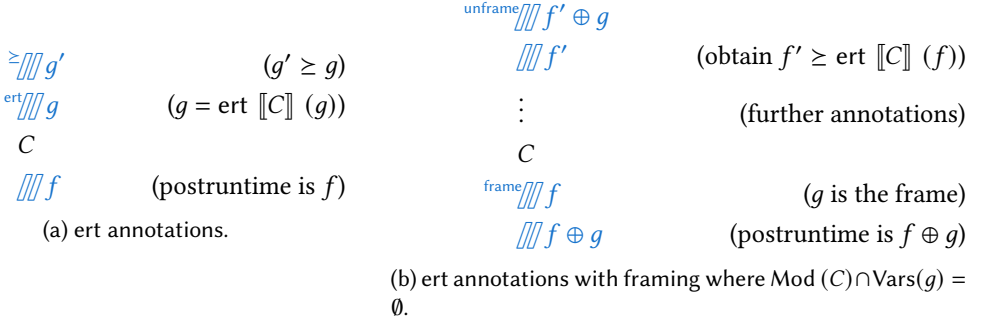

	\begin{subfigure}[t]{0.4\linewidth}
		\begin{minipage}{1\linewidth}%
			\abovedisplayskip=-0em%
			\belowdisplayskip=0pt%
			\begin{align*}
			&\succeqannotate{\rtb'} \tag{$\rtb' \succeq \rtb$}\\
			&\ertannotate{\rtb} \tag{$\rtb = \ert{C}{\rtb}$}\\
			&C \\
			&\annotate{\rta} \tag{postruntime is $\rta$}
			\end{align*}%
			\normalsize%
		\end{minipage}%
		\subcaption{$\ertsymbol$ annotations.}
		\label{fig:ert-annotations_monotonicity}
	\end{subfigure}
	\quad\hfill
	\begin{subfigure}[t]{0.55\linewidth}
		\begin{minipage}{1\linewidth}%
			\abovedisplayskip=-0em%
			\belowdisplayskip=0pt%
			\begin{align*}
			&\unframeannotate{\rta' \sepadd \rtb}{} \\
			& \annotate{\rta'} \tag{obtain $\rta' \succeq \ert{C}{\rta}$} \\
			& \vdots \tag{further annotations} \\
			&C \\
			&\frameannotate{\rta}{} \tag{$\rtb$ is the frame} \\
			&\annotate{\rta \sepadd \rtb} \tag{postruntime is $\rta \sepadd \rtb$}
			\end{align*}%
			\normalsize%
		\end{minipage}%
		\subcaption{$\ertsymbol$ annotations with framing where $\modC{C}\cap \Vars(\rtb) = \emptyset$.}
		\label{fig:ert-annotations_framing}
	\end{subfigure}
	\vspace{-1em}
	\caption{Program annotation style for exploiting monotonicity and framing. It is more intuitive to read these annotations from bottom to top.}
	\label{fig:annotatio_ert}
	\vspace{-1em}
\end{figure}%

\begin{figure}[t]
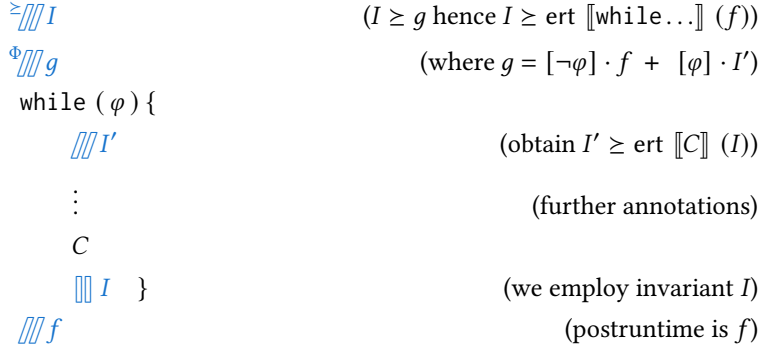

	\begin{minipage}{1\linewidth}
		\begin{align*}
		&\succeqannotate{I} \tag{$I \succeq \rtb$ hence $I \succeq \ert{\WHILESYMBOL \ldots}{f}$}\\
		&\phiannotate{\rtb} \tag{where $\rtb = \iverson{\neg \guard} \cdot \rta \pplus \iverson{\guard} \cdot I'$}\\
		&\WHILE{\guard} \\
		&\qquad \annotate{I'} \tag{obtain $I' \succeq \ert{C}{I}$}\\
		&\qquad \vdots \tag{further annotations} \\
		&\qquad C \\
		&\qquad \starannotate{I} \quad \} \tag{we employ invariant $I$}\\
		&\annotate{\rta} \tag{postruntime is $\rta$}
		\end{align*}
	\end{minipage}
	\caption{Annotation style for loops using invariants (cf. \Cref{thm:ert_invariants})}
	\label{fig:ert-loop-annotations}
\end{figure}

\subsection{Verification of the Lagging List Traversal}
\label{app:lagging_list}

\begin{figure}[!p]
	\begin{align*}
	&\succeqannotate{2 \cdot \Inf \lh \colon \elist{\lh}{\varlistelem} \sepadd \esize{\varlistelem}{0}} \\
	& \phiannotate{\iverson{\varlistelem=0} \cdot 0 \pplus \iverson{\varlistelem\neq 0} \cdot\big( 2 \cdot \Inf \lh \colon \elist{\lh}{\varlistelem} \sepadd \esize{\varlistelem}{0} \sepadd \pureemp{\varlistelem\neq 0}\big)} \\
	%
	%
	&\WHILE{\varlistelem  \neq 0 } \\
	&\qquad \succeqannotate{2 \cdot \Inf \lh \colon \elist{\lh}{\varlistelem} \sepadd \esize{\varlistelem}{0} \sepadd \pureemp{\varlistelem\neq 0}} \\
	&\qquad \succeqannotate{\Inf \lh \colon \elist{\lh}{\varlistelem} \sepadd \esize{\varlistelem}{0} \sepadd \pureemp{\varlistelem\neq 0}} \\
	& \qquad \quad \phantannotate{{}+{} \Inf \lh \colon \elist{\lh}{\varlistelem} \sepadd \esize{\varlistelem}{0}}  \\
	&\qquad \ertannotate{\emprun{1} \sepadd \big(\Inf \lh, v, v'\colon \elist{\lh}{v'} \sepadd \esize{v}{0} \sepadd \pureemp{\varlistelem=v'} \sepadd \isingleton{\varlistelem}{v}} \\
	& \qquad \quad \phantannotate{{}+{} \Inf \lh \colon \elist{\lh}{\varlistelem} \sepadd \esize{\varlistelem}{0}\big)}  \\
	& \qquad \TICK{1}\fatsemi \\
	&\qquad \succeqannotate{\Inf \lh, v, v' \colon \elist{\lh}{v'} \sepadd \esize{v}{0} \sepadd \pureemp{\varlistelem=v'} \sepadd \isingleton{\varlistelem}{v}} \\
	& \qquad \quad \phantannotate{{}+{} \Inf \lh \colon \elist{\lh}{\varlistelem} \sepadd \esize{\varlistelem}{0}} \\
	&\qquad \ertannotate{0.5 \cdot (\Inf \lh, v, v'\colon \elist{\lh}{v'} \sepadd 2 \cdot \esize{v}{0} \sepadd \pureemp{\varlistelem=v'} \sepadd \isingleton{\varlistelem}{v})} \\
     & \qquad \quad \phantannotate{{}+{} 0.5 \cdot (2 \cdot \Inf \lh \colon \elist{\lh}{\varlistelem} \sepadd \esize{\varlistelem}{0})} \\
	& \qquad \{ \\
	& \qquad \qquad \auxvarannotate{\Inf \lh, v, v'\colon  \elist{\lh}{v'} \sepadd 2 \cdot \esize{v}{0} \sepadd \pureemp{\varlistelem=v'} \sepadd \isingleton{\varlistelem}{v}}{} \\
	& \qquad \qquad \unframeannotate{\elist{\lh}{v'} \sepadd 2 \cdot \esize{v}{0} \sepadd \pureemp{\varlistelem=v'} \sepadd \isingleton{\varlistelem}{v}}{} \\
	& \qquad \qquad \lkpannotate{\pureemp{\varlistelem=v'} \sepadd \isingleton{\varlistelem}{v}} \\
	&\qquad \qquad  \ASSIGNH{\varlistelem}{\varlistelem} \\
	& \qquad \qquad \frameannotate{\pureemp{\varlistelem=v} \sepadd \isingleton{v'}{\varlistelem}}{} \\
	& \qquad \qquad \auxvarannotate{\elist{\lh}{v'} \sepadd 2 \cdot \esize{v}{0} \sepadd \pureemp{\varlistelem=v} \sepadd \isingleton{v'}{\varlistelem}} \\
	& \qquad \qquad \succeqannotate{\Inf \lh, v, v'  \colon  \elist{\lh}{v'} \sepadd 2 \cdot \esize{v}{0} \sepadd \pureemp{\varlistelem=v} \sepadd \isingleton{v'}{\varlistelem}} \\
	& \qquad \qquad \annotate{2 \cdot \Inf \lh \colon \elist{\lh}{\varlistelem} \sepadd \esize{\varlistelem}{0}} \\
	& \qquad \} [0.5] \{ \\
	& \qquad \qquad \ertannotate{2 \cdot \Inf \lh \colon \elist{\lh}{\varlistelem} \sepadd \esize{\varlistelem}{0}} \\
	& \qquad \qquad \SKIP \\
	& \qquad \qquad \annotate{2 \cdot \Inf \lh \colon \elist{\lh}{\varlistelem} \sepadd \esize{\varlistelem}{0}} \\
	& \qquad \} \\
	& \qquad \starannotate{2 \cdot \Inf \lh \colon \elist{\lh}{\varlistelem} \sepadd \esize{\varlistelem}{0}} \\
	&\} \\
	&\annotate{0}
	\end{align*}

	\caption{Lagging List Traversal Runtime Verification. The annotations are more intuitive to read from bottom to top. See \Cref{app:annotations} for explanations on these annotations.}\label{fig:proof_lagging_list}
\end{figure}

See \Cref{fig:proof_lagging_list} for detailed source-code annotations for the lagging list traversal case study.

\subsection{Proof of \Cref{thm:op:eert-bellman}}\label{a:op:eert-bellman}

\begin{proof}

By construction of $\eert$ and by the definition of bellman compliance, it suffices
to show that, for all $C \in \hpgcl$, $\rta \in \T$ and program states $(\sk,\hh)$, we have
\[
    \underbrace{~\ert{C}{\rta}(\sk,\hh)~}_{\eeq \eertt{C}{\rta}(\sk,\hh)} 
    \qeq \MROP(\conf{C}{\sk}{\hh}) \pplus \sup_{\ma \in \MA} 
    \sum_{(\conf{C}{\sk}{\hh}) \mstep{p}{\ma} (\conf{C'}{\sk'}{\hh'})} \eertt{C'}{\rta}(\sk',\hh')~.
\]
We proceed by induction on the structure of rules of our execution relation for $\hpgcl$ programs (i.e., we do not have to consider rules involving only $\term$, $\fail$, and $\sink$, cf. \Cref{fig:op:rules}).
All rules except those for sequential composition are base cases; we consider them grouped by $\hpgcl$ statement.


\paragraph{The case $C = \TICK{e}$.}
\begin{align*}
       & \ert{\TICK{e}}{\rta}(\sk,\hh) \\
  \eeq & \emprun{\sk(e)}(\sk,\hh) ~\sepadd~ \rta(\sk,\hh) \tag{Def. of $\ertsymbol$} \\
  \eeq & \sk(e) + \rta(\sk,\hh) \tag{Def. of $\emprun{\sk(e)}$ and $\sepadd$} \\
  \eeq & \sk(e) \pplus \eertt{\term}{\rta}(\sk,\hh) \\
  \eeq & \MROP(\conf{\TICK{e}}{\sk}{\hh}) \pplus
         \sup_{\ma \in \MA}
         \sum_{(\conf{\TICK{e}}{\sk}{\hh}) \mstep{p}{\ma} (\conf{C'}{\sk'}{\hh'})}
 	     p \cdot \eertt{C'}{\rta}(\sk',\hh')~.
 	     \tag{Def. of $\MROP$, \Cref{fig:op:rules}}
\end{align*}

\paragraph{The case $C = \ASSIGN{x}{e}$.}
\begin{align*}
       & \ert{\ASSIGN{x}{e}}{\rta}(\sk,\hh) \\
  \eeq & \rta\subst{x}{e}(\sk,\hh) \tag{Def. of $\ertsymbol$} \\
  \eeq & \rta(\sk\subst{x}{e}, \hh) \tag{standard substitution lemma} \\
  \eeq & \eertt{\term}{\rta}(\sk\subst{x}{\sk(e)},\hh) \\
  \eeq & \MROP(\conf{\ASSIGN{x}{e}}{\sk}{\hh}) \pplus
         \sup_{\ma \in \MA}
         \sum_{(\conf{\ASSIGN{x}{e}}{\sk}{\hh}) \mstep{p}{\ma} (\conf{C'}{\sk'}{\hh'})}
 	     p \cdot \eertt{C'}{\rta}(\sk',\hh')~.
 	     \tag{Def. of $\MROP$, \Cref{fig:op:rules}}
\end{align*}

\paragraph{The case $C = \ALLOC{x}{e}$.}
We distinguish two cases.

First, assume $\sk(e) = n > 0$ and $\hh'' = \singleheap{\val}{0} \sepadd \ldots \sepadd \singleheap{\val+n-1}{0}$. 
Then:
\begin{align*}
       & \ert{\ALLOC{x}{e}}{\rta}(\sk,\hh) \\
  \eeq & \left(\Sup\, v \colon \bigl(\bigoplus_{i=1}^{e } \isingleton{e+i-1}{0}\bigr) \sepmon \rta\subst{x}{v}\right)(\sk,\hh ) \tag{Def. of $\ertsymbol$} \\
  \eeq & \sup_{v \in \Nats} \rta(\sk\subst{x}{v},\hh \sepcon \hh'' \tag{Def. of $\sepmon$, assumption} \\
  \eeq & \sup_{v \in \Nats} \eertt{\term}{\rta}(\sk\subst{x}{v},\hh \sepcon \hh'') \\
  \eeq & \MROP(\conf{\ALLOC{x}{e}}{\sk}{\hh}) \pplus
         \sup_{\ma \in \MA}
         \sum_{(\conf{\ALLOC{x}{e}}{\sk}{\hh}) \mstep{p}{\ma} (\conf{C'}{\sk'}{\hh'})}
 	     p \cdot \eertt{C'}{\rta}(\sk',\hh')~.
 	     \tag{$\MA = \Nats$, Def. of $\MROP$, \Cref{fig:op:rules}}
\end{align*}

Second, assume $\sk(e) = 0$. Then:
\begin{align*}
       & \ert{\ALLOC{x}{e}}{\rta}(\sk,\hh) \\
  \eeq & \left(\Sup\, v \colon \bigl(\bigoplus_{i=1}^{e} \isingleton{e+i-1}{0}\bigr) \sepmon \rta\subst{x}{v}\right)(\sk,\hh ) \tag{Def. of $\ertsymbol$} \\
  \eeq & \left(\Sup\, v \colon \bigl(\iemp \sepmon \rta\subst{x}{v}\bigr)\right)(\sk,\hh ) \tag{assumption} \\
  \eeq & \sup_{v \in \Nats} \rta(\sk\subst{x}{v},\hh \tag{Def. of $\sepmon$} \\
  \eeq & \sup_{v \in \Nats} \eertt{\term}{\rta}(\sk\subst{x}{v},\hh) \\
  \eeq & \MROP(\conf{\ALLOC{x}{e}}{\sk}{\hh}) \pplus
         \sup_{\ma \in \MA}
         \sum_{(\conf{\ALLOC{x}{e}}{\sk}{\hh}) \mstep{p}{\ma} (\conf{C'}{\sk'}{\hh'})}
 	     p \cdot \eertt{C'}{\rta}(\sk',\hh')~.
 	     \tag{$\MA = \Nats$, Def. of $\MROP$, \Cref{fig:op:rules}}
\end{align*}


\paragraph{The case $C = \HASSIGN{e}{e'}$.}
Let $v = \sk(e)$. We distinguish two cases.

If $v \in \dom{h}$, then
\begin{align*}
       & \ert{\HASSIGN{e}{e'}}{\rta}(\sk,\hh) \\
  \eeq & \left(\ivalidpointer{e} \sepadd \bigl(\isingleton{e}{e'} \sepmon \rta \bigr)\right)(\sk,\hh) \tag{Def. of $\ertsymbol$} \\
  \eeq & \rta(\sk,\hh\subst{v}{\sk(e')}) \tag{$v \in \dom{h}$, Def. of $\sepmon$} \\
  \eeq & \eertt{\term}{\rta}(\sk,\hh\subst{v}{\sk(e')}) \\
  \eeq & \MROP(\conf{\HASSIGN{e}{e'}}{\sk}{\hh}) \pplus
         \sup_{\ma \in \MA}
         \sum_{(\conf{\HASSIGN{e}{e'}}{\sk}{\hh}) \mstep{p}{\ma} (\conf{C'}{\sk'}{\hh'})}
 	     p \cdot \eertt{C'}{\rta}(\sk',\hh')~.
 	     \tag{$v \in \dom{h}$, Def. of $\MROP$, \Cref{fig:op:rules}}
\end{align*}
If $v \notin \dom{h}$, then
\begin{align*}
       & \ert{\HASSIGN{e}{e'}}{\rta}(\sk,\hh) \\
  \eeq & \left(\ivalidpointer{e} \sepadd \bigl(\isingleton{x}{e'} \sepmon \rta \bigr)\right)(\sk,\hh) \tag{Def. of $\ertsymbol$} \\
  \eeq & \infty \tag{$v \notin \dom{h}$, Def. of $\sepmon$} \\
  \eeq & \eertt{\fail}{\rta}(\sk,\hh) \\
    \eeq & \MROP(\conf{\HASSIGN{e}{e'}}{\sk}{\hh}) \pplus
         \sup_{\ma \in \MA}
         \sum_{(\conf{\HASSIGN{e}{e'}}{\sk}{\hh}) \mstep{p}{\ma} (\conf{C'}{\sk'}{\hh'})}
 	     p \cdot \eertt{C'}{\rta}(\sk',\hh')~.
 	     \tag{Def. of $\MROP$, by \Cref{fig:op:rules}: $\conf{\HASSIGN{e}{e'}}{\sk}{\hh} \mstep{}{} \conf{\fail}{\sk}{\hh}$}
\end{align*}

\paragraph{The case $C = \ASSIGNH{x}{e}$.}
Let $v = \sk(e)$.  We distinguish two cases.

If $v \in \dom{h}$, then
\begin{align*}
       & \ert{\ASSIGNH{x}{e}}{\rta}(\sk,\hh) \\
  \eeq & \left(\Inf\, v \colon \isingleton{e}{v} \sepadd \bigl(\isingleton{e}{v} \sepmon \rta\subst{x}{v}\bigr)\right)(\sk,\hh) \tag{Def. of $\ertsymbol$} \\
  \eeq & \rta(\sk\subst{x}{\hh(v)},\hh) \tag{$\sk(e) = v \in \dom{h}$, Def. of $\sepmon$} \\
  \eeq & \eertt{\term}{\rta}(\sk,\hh\subst{v}{\sk(e')}) \\
    \eeq & \MROP(\conf{\ASSIGNH{x}{e}}{\sk}{\hh}) \pplus
         \sup_{\ma \in \MA}
         \sum_{(\conf{\HASSIGN{x}{e}}{\sk}{\hh}) \mstep{p}{\ma} (\conf{C'}{\sk'}{\hh'})}
 	     p \cdot \eertt{C'}{\rta}(\sk',\hh')~.
 	     \tag{$v \in \dom{h}$, Def. of $\MROP$, \Cref{fig:op:rules}}
\end{align*}
If $v \notin \dom{h}$, then
\begin{align*}
       & \ert{\ASSIGNH{x}{e}}{\rta}(\sk,\hh) \\
  \eeq & \left(\Inf\, v \colon \isingleton{e}{v} \sepadd \bigl(\isingleton{e}{v} \sepmon \rta\subst{x}{v}\bigr)\right)(\sk,\hh) \tag{Def. of $\ertsymbol$} \\
  \eeq & \infty \tag{$v \notin \dom{h}$, Def. of $\sepmon$} \\
  \eeq & \eertt{\fail}{\rta}(\sk,\hh) \\
  \eeq & \MROP(\conf{\ASSIGNH{x}{e}}{\sk}{\hh}) \pplus
         \sup_{\ma \in \MA}
         \sum_{(\conf{\ASSIGNH{x}{e}}{\sk}{\hh}) \mstep{p}{\ma} (\conf{C'}{\sk'}{\hh'})}
 	     p \cdot \eertt{C'}{\rta}(\sk',\hh')~.
 	     \tag{$v \notin \dom{h}$, Def. of $\MROP$, \Cref{fig:op:rules}}
\end{align*}

\paragraph{The case $C = \FREE{e}$.}
We distinguish two cases.

First, assume there exists a heap $\hh''$ and a value $v$ such that $\hh = \hh'' \sepcon \{ \sk(e) \mapsto v \}$. Then:
\begin{align*}
       & \ert{\FREE{e}}{\rta}(\sk,\hh) \\
  \eeq & \left(\ivalidpointer{e} \sepadd \rta\right)(\sk,\hh) \tag{Def. of $\ertsymbol$} \\
  \eeq & f(\sk, \hh'') \tag{Def. of $\sepadd$, assumption} \\
  \eeq & \eertt{\term}{\rta}(\sk,\hh'') \\
  \eeq & \MROP(\conf{\FREE{e}}{\sk}{\hh}) \pplus
         \sup_{\ma \in \MA}
         \sum_{(\conf{\FREE{e}}{\sk}{\hh}) \mstep{p}{\ma} (\conf{C'}{\sk'}{\hh'})}
 	     p \cdot \eertt{C'}{\rta}(\sk',\hh')~.
 	     \tag{Def. of $\MROP$, \Cref{fig:op:rules}}
\end{align*}
Second, no such $\hh''$ exists, \ie, $\sk(e) \notin \dom{\hh}$. Then:
First, assume there exists a heap $\hh'$ and a value $v$ such that $\hh = \{ \sk(e) \mapsto v \}$. Then:
\begin{align*}
       & \ert{\FREE{e}}{\rta}(\sk,\hh) \\
  \eeq & \left(\ivalidpointer{e} \sepadd \rta\right)(\sk,\hh) \tag{Def. of $\ertsymbol$} \\
  \eeq & \infty \tag{Def. of $\sepadd$, assumption} \\
  \eeq & \eertt{\fail}{\rta}(\sk,\hh) \\
  \eeq & \MROP(\conf{\FREE{e}}{\sk}{\hh}) \pplus
         \sup_{\ma \in \MA}
         \sum_{(\conf{\FREE{e}}{\sk}{\hh}) \mstep{p}{\ma} (\conf{C'}{\sk'}{\hh'})}
 	     p \cdot \eertt{C'}{\rta}(\sk',\hh')~.
 	     \tag{Def. of $\MROP$, \Cref{fig:op:rules}}
\end{align*}

\paragraph{The case $C = \PCHOICE{C_1}{p}{C_2}$.}
\begin{align*}
       & \ert{\PCHOICE{C_1}{p}{C_2}}{\rta}(\sk,\hh) \\
  \eeq & p(\sk,\hh) \cdot \ert{C_1}{\rta}(\sk,\hh) \pplus (1-p(\sk,\hh)) \cdot \ert{C_2}{\rta}(\sk,\hh) \tag{Def. of $\ertsymbol$} \\
  \eeq & \MROP(\conf{\PCHOICE{C_1}{p}{C_2}}{\sk}{\hh}) \pplus \sup_{\ma \in \MA} p(\sk,\hh) \cdot \eertt{C_1}{\rta}(\sk,\hh) \pplus (1-p(\sk,\hh)) \cdot \eertt{C_2}{\rta}(\sk,\hh) \tag{Def. of $\MROP$ and $\eert$} \\
  \eeq & \MROP(\conf{\PCHOICE{C_1}{p}{C_2}}{\sk}{\hh}) \pplus
         \sup_{\ma \in \MA}
         \sum_{(\conf{\ASSIGN{x}{e}}{\sk}{\hh}) \mstep{q}{\ma} (\conf{C'}{\sk'}{\hh'})}
 	     q \cdot \eertt{C'}{\rta}(\sk',\hh')~.
 	     \tag{\Cref{fig:op:rules}}
\end{align*}

\paragraph{The case $C = \ITE{\varphi}{C_1}{C_2}$.}
We distinguish two cases: $\sk \models \varphi$ and $\sk \not\models \varphi$.

First, assume $\sk \models \varphi$. Then:
\begin{align*}
       & \ert{\ITE{\varphi}{C_1}{C_2}}{\rta}(\sk,\hh) \\
  \eeq & \ert{C_1}{\rta}(\sk,\hh) \tag{Def. of $\ertsymbol$, assumption} \\
  \eeq & \eertt{C_1}{\rta}(\sk,\hh) \\
  \eeq & \MROP(\conf{\ITE{\varphi}{C_1}{C_2}}{\sk}{\hh}) \pplus 
         \tag{assumption, Def. of $\MROP$, \Cref{fig:op:rules}} \\
       & \qquad
         \sup_{\ma \in \MA}
         \sum_{(\conf{\ITE{\varphi}{C_1}{C_2}}{\sk}{\hh}) \mstep{p}{\ma} (\conf{C'}{\sk'}{\hh'})}
 	     p \cdot \eertt{C'}{\rta}(\sk',\hh')~.
\end{align*}
Second, assume $\sk \not\models \varphi$. Then:
\begin{align*}
       & \ert{\ITE{\varphi}{C_1}{C_2}}{\rta}(\sk,\hh) \\
  \eeq & \ert{C_2}{\rta}(\sk,\hh) \tag{Def. of $\ertsymbol$, assumption} \\
  \eeq & \eertt{C_2}{\rta}(\sk,\hh) \\
  \eeq & \MROP(\conf{\ITE{\varphi}{C_1}{C_2}}{\sk}{\hh}) \pplus 
         \tag{assumption, Def. of $\MROP$, \Cref{fig:op:rules}} \\
       & \qquad
         \sup_{\ma \in \MA}
         \sum_{(\conf{\ITE{\varphi}{C_1}{C_2}}{\sk}{\hh}) \mstep{p}{\ma} (\conf{C'}{\sk'}{\hh'})}
 	     p \cdot \eertt{C'}{\rta}(\sk',\hh')~.
\end{align*}

\paragraph{The case $C = \WHILEDO{\varphi}{C_1}$.}
Observe that $\WHILEDO{\varphi}{C_1}$ is equivalent to 
\[ \ITE{\varphi}{\COMPOSE{C_1}{\WHILEDO{\varphi}{C_1}}}{\SKIP}~. \]
(It is straightforward to check that the $\ertsymbol$ of both programs is identical.)

We distinguish two cases: $\sk \models \varphi$ and $\sk \not\models \varphi$.

First, assume $\sk \models \varphi$. Then:
\begin{align*}
       & \ert{\WHILEDO{\varphi}{C_1}}{\rta}(\sk,\hh) \\
  \eeq & \ert{\COMPOSE{C_1}{\WHILEDO{\varphi}{C_1}}}{\rta}(\sk,\hh) \tag{assumption and observation} \\
  \eeq & \eertt{\COMPOSE{C_1}{\WHILEDO{\varphi}{C_1}}}{\rta}(\sk,\hh) \\
  \eeq & \MROP(\conf{\WHILEDO{\varphi}{C_1}}{\sk}{\hh}) \pplus 
         \tag{assumption, Def. of $\MROP$, \Cref{fig:op:rules}} \\
       & \qquad
         \sup_{\ma \in \MA}
         \sum_{(\conf{\WHILEDO{\varphi}{C_1}}{\sk}{\hh}) \mstep{p}{\ma} (\conf{C'}{\sk'}{\hh'})}
 	     p \cdot \eertt{C'}{\rta}(\sk',\hh')~.
\end{align*}
Second, assume $\sk \not\models \varphi$. Then:
\begin{align*}
       & \ert{\WHILEDO{\varphi}{C_1}}{\rta}(\sk,\hh) \\
  \eeq & \rta(\sk,\hh) \tag{Def. of $\ertsymbol$, assumption} \\
  \eeq & \eertt{\term}{\rta}(\sk,\hh) \\
  \eeq & \MROP(\conf{\WHILEDO{\varphi}{C_1}}{\sk}{\hh}) \pplus 
         \tag{assumption, Def. of $\MROP$, \Cref{fig:op:rules}} \\
       & \qquad
         \sup_{\ma \in \MA}
         \sum_{(\conf{\WHILEDO{\varphi}{C_1}}{\sk}{\hh}) \mstep{p}{\ma} (\conf{C'}{\sk'}{\hh'})}
 	     p \cdot \eertt{C'}{\rta}(\sk',\hh')~.
\end{align*}

\paragraph{The case $C = \COMPOSE{C_1}{C_2}$.}

By construction of the execution relation $\mstep{}{}$ (cf., \Cref{fig:op:rules}), a single execution step starting in a configuration $\conf{C_1}{\sk}{\hh}$ either proceeds execution, terminates successfully in $\term$ or fails with an error by moving to $\fail$; it is, for example, never possible to terminate and fail via two distinct \emph{single} execution steps.
Formally, we distinguish the following three cases:
\begin{enumerate}
	\item $\conf{C_1}{\sk}{\hh} \mstep{p}{\ma} \conf{C_1'}{\sk'}{\hh'}$,
	\item $\conf{C_1}{\sk}{\hh} \mstep{p}{\ma} \conf{\term}{\sk'}{\hh'}$, and
	\item $\conf{C_1}{\sk}{\hh} \mstep{p}{\ma} \conf{\fail}{\sk}{\hh}$.
\end{enumerate}

First, assume all executions steps starting in $\conf{C_1}{\sk}{\hh}$ are of the form $\conf{C_1}{\sk}{\hh} \mstep{p}{\ma} \conf{C_1'}{\sk'}{\hh'}$.
Then:
\begin{align*}
       & \ert{\COMPOSE{C_1}{C_2}}{\rta}(\sk,\hh) \\
  \eeq & \ert{C_1}{\ert{C_2}{\rta}}(\sk,\hh) \tag{Def. of $\ertsymbol$} \\
  \eeq & \MROP(\conf{C_1}{\sk}{\hh}) \pplus 
         \sup_{\ma \in \MA}
         \sum_{(\conf{C_1}{\sk}{\hh}) \mstep{p}{\ma} (\conf{C_1'}{\sk'}{\hh'})}
 	     p \cdot \eertt{C_1'}{\ert{C_2}{\rta}}(\sk',\hh')
 	     \tag{I.H., assumption} \\
 \eeq &	 \MROP(\conf{\COMPOSE{C_1}{C_2}}{\sk}{\hh}) \pplus 
         \sup_{\ma \in \MA}
         \sum_{(\conf{C_1}{\sk}{\hh}) \mstep{p}{\ma} (\conf{C_1'}{\sk'}{\hh'})}
 	     p \cdot \eertt{\COMPOSE{C_1'}{C_2}}{\rta}(\sk',\hh')
 	     \tag{Def. of $\eert$, $\MROP$} \\
 \eeq &	 \MROP(\conf{\COMPOSE{C_1}{C_2}}{\sk}{\hh}) \pplus 
         \sup_{\ma \in \MA}
         \sum_{(\conf{\COMPOSE{C_1}{C_2}}{\sk}{\hh}) \mstep{p}{\ma} (\conf{C'}{\sk'}{\hh'})}
 	     p \cdot \eertt{C'}{\rta}(\sk',\hh')~.
 	     \tag{assumption, \Cref{fig:op:rules}} 
\end{align*}

Second, assume all executions steps starting in $\conf{C_1}{\sk}{\hh}$ are of the form $\conf{C_1}{\sk}{\hh} \mstep{p}{\ma} \conf{\term}{\sk'}{\hh'}$.
Then:
\begin{align*}
       & \ert{\COMPOSE{C_1}{C_2}}{\rta}(\sk,\hh) \\
  \eeq & \ert{C_1}{\ert{C_2}{\rta}}(\sk,\hh) \tag{Def. of $\ertsymbol$} \\
  \eeq & \MROP(\conf{C_1}{\sk}{\hh}) \pplus 
         \sup_{\ma \in \MA}
         \sum_{(\conf{C_1}{\sk}{\hh}) \mstep{p}{\ma} (\conf{\term}{\sk'}{\hh'})}
 	     p \cdot \eertt{\term}{\ert{C_2}{\rta}}(\sk',\hh')
 	     \tag{I.H., assumption} \\
 \eeq & \MROP(\conf{C_1}{\sk}{\hh}) \pplus 
         \sup_{\ma \in \MA}
         \sum_{(\conf{C_1}{\sk}{\hh}) \mstep{p}{\ma} (\conf{\term}{\sk'}{\hh'})}
 	     p \cdot \eertt{C_2}{\rta}(\sk',\hh')
 	     \tag{Def. of $\eert$} \\
 \eeq &	 \MROP(\conf{\COMPOSE{C_1}{C_2}}{\sk}{\hh}) \pplus 
         \sup_{\ma \in \MA}
         \sum_{(\conf{\COMPOSE{C_1}{C_2}}{\sk}{\hh}) \mstep{p}{\ma} (\conf{C'}{\sk'}{\hh'})}
 	     p \cdot \eertt{C'}{\rta}(\sk',\hh')~.
 	     \tag{Def. of $\MROP$, assumption \Cref{fig:op:rules}}
\end{align*}

Third, assume all executions steps starting in $\conf{C_1}{\sk}{\hh}$ are of the form $\conf{C_1}{\sk}{\hh} \mstep{p}{\ma} \conf{\fail}{\sk}{\hh}$.
Then:
\begin{align*}
       & \ert{\COMPOSE{C_1}{C_2}}{\rta}(\sk,\hh) \\
  \eeq & \ert{C_1}{\ert{C_2}{\rta}}(\sk,\hh) \tag{Def. of $\ertsymbol$} \\
  \eeq & \MROP(\conf{C_1}{\sk}{\hh}) \pplus 
         \sup_{\ma \in \MA}
         \sum_{(\conf{C_1}{\sk}{\hh}) \mstep{p}{\ma} (\conf{\fail}{\sk}{\hh})}
 	     p \cdot \eertt{\fail}{\ert{C_2}{\rta}}(\sk,\hh)
 	     \tag{I.H., assumption} \\
 \eeq & \MROP(\conf{\COMPOSE{C_1}{C_2}}{\sk}{\hh}) \pplus 
         \sup_{\ma \in \MA}
         \sum_{(\conf{C_1}{\sk}{\hh}) \mstep{p}{\ma} (\conf{\fail}{\sk}{\hh})}
 	     p \cdot \infty
 	     \tag{Def. of $\MROP$, $\eert$} \\
 \eeq & \MROP(\conf{\COMPOSE{C_1}{C_2}}{\sk}{\hh}) \pplus 
         \sup_{\ma \in \MA}
         \sum_{(\conf{C_1}{\sk}{\hh}) \mstep{p}{\ma} (\conf{\fail}{\sk}{\hh})}
 	     p \cdot \eertt{\fail}{\rta}(\sk,\hh)
 	     \tag{Def. of $\eert$} \\
 \eeq &	 \MROP(\conf{\COMPOSE{C_1}{C_2}}{\sk}{\hh}) \pplus 
         \sup_{\ma \in \MA}
         \sum_{(\conf{\COMPOSE{C_1}{C_2}}{\sk}{\hh}) \mstep{p}{\ma} (\conf{C'}{\sk'}{\hh'})}
 	     p \cdot \eertt{C'}{\rta}(\sk',\hh')~.
 	     \tag{assumption, \Cref{fig:op:rules}}
\end{align*}
	
Hence, $\eert$ is Bellman compliant.
\end{proof}

\subsection{Proof of \Cref{thm:op:eert-leq-oprt}}\label{a:op:eert-leq-oprt}

By construction of $\eert$ and the definition of bellman compliance, it suffices
to show that, for all $C \in \hpgcl$, $\rta \in \T$ and program states $(\sk,\hh)$, we have
 	    \[
 	        \eertt{C}{\rta}(\sk,\hh) 
 	        \qeq \ert{C}{\rta}(\sk,\hh) 
 	        \quad\leq\quad\quad
 	        \oprtt{C}{\rta}(\sk,\hh)~.
 	    \]
 	    
\begin{proof}

We show the above inequality by induction on the structure of $\hpgcl$ programs.

For the base cases, let $C$ be one of the $\hpgcl$ programs
$\SKIP$, $\TICK{e}$, $\ASSIGN{x}{e}$, $\HASSIGN{e}{e'}$, $\ASSIGNH{x}{e}$, $\ALLOC{x}{e}$, or $\FREE{e}$.
As can be seen in \Cref{fig:op:rules}, we need to distinguish two kinds of execution steps: those that terminate, \ie, move to $\term$, and those that fail, \ie, move to $\fail$.

First, assume that an execution step starting in $\conf{C}{\sk}{\hh}$ terminates, \ie is of the form
\[ \conf{C}{\sk}{\hh} \qmstep{p}{\ma} \conf{\term}{\sk'}{\hh'}~. \]
Then, consider the following:

\begin{align*}
     & \eertt{C}{\rta}(\sk,\hh) \\
\eeq & \MROP(\conf{C}{\sk}{\hh}) \pplus 
         \sup_{\ma \in \MA}
         \sum_{(\conf{C}{\sk}{\hh}) \mstep{p}{\ma} (\conf{C'}{\sk'}{\hh'})}
 	     p \cdot \eertt{C'}{\rta}(\sk',\hh')
 	     \tag{\Cref{thm:op:eert-bellman}} \\
\eeq & \MROP(\conf{C}{\sk}{\hh}) \pplus 
         \sup_{\ma \in \MA}
         \sum_{(\conf{C}{\sk}{\hh}) \mstep{p}{\ma} (\conf{\term}{\sk'}{\hh'})}
 	     p \cdot \eertt{\term}{\rta}(\sk',\hh')
 	     \tag{assumption} \\ 	   
\eeq & \MROP(\conf{C}{\sk}{\hh}) \pplus 
         \sup_{\ma \in \MA}
         \sum_{(\conf{C}{\sk}{\hh}) \mstep{p}{\ma} (\conf{\term}{\sk'}{\hh'})}
 	     p \cdot \rta(\sk',\hh')
 	     \tag{Def. of $\eert$} \\ 	
\eeq & \MROP(\conf{C}{\sk}{\hh}) \pplus 
         \sup_{\ma \in \MA}
         \sum_{(\conf{C}{\sk}{\hh}) \mstep{p}{\ma} (\conf{\term}{\sk'}{\hh'})}
 	     p \cdot \oprtt{\term}{\rta}(\sk',\hh')
 	     \tag{\Cref{thm:op:bellman}} \\ 	 
\eeq & \oprtt{C}{\rta}(\sk,\hh)~.
 	     \tag{assumption, \Cref{thm:op:bellman}} \\   
\end{align*}

Second, assume that an execution step starting in $\conf{C}{\sk}{\hh}$ leads to an error, \ie is of the form
\[ \conf{C}{\sk}{\hh} \qmstep{p}{\ma} \conf{\fail}{\sk}{\hh}~. \]
Then, consider the following:
\begin{align*}
     & \eertt{C}{\rta}(\sk,\hh) \\
\eeq & \MROP(\conf{C}{\sk}{\hh}) \pplus 
         \sup_{\ma \in \MA}
         \sum_{(\conf{C}{\sk}{\hh}) \mstep{p}{\ma} (\conf{C'}{\sk'}{\hh'})}
 	     p \cdot \eertt{C'}{\rta}(\sk',\hh')
 	     \tag{\Cref{thm:op:eert-bellman}} \\
\eeq & \MROP(\conf{C}{\sk}{\hh}) \pplus 
         \sup_{\ma \in \MA}
         \sum_{(\conf{C}{\sk}{\hh}) \mstep{p}{\ma} (\conf{\fail}{\sk}{\hh})}
 	     p \cdot \eertt{\fail}{\rta}(\sk,\hh)
 	     \tag{assumption} \\ 	   
\eeq & \MROP(\conf{C}{\sk}{\hh}) \pplus 
         \sup_{\ma \in \MA}
         \sum_{(\conf{C}{\sk}{\hh}) \mstep{p}{\ma} (\conf{\fail}{\sk}{\hh})}
 	     p \cdot \infty
 	     \tag{Def. of $\eert$} \\ 	
\eeq & \MROP(\conf{C}{\sk}{\hh}) \pplus 
         \sup_{\ma \in \MA}
         \sum_{(\conf{C}{\sk}{\hh}) \mstep{p}{\ma} (\conf{\fail}{\sk}{\hh})}
 	     p \cdot \oprtt{\fail}{\rta}(\sk,\hh)
 	     \tag{\Cref{thm:op:bellman}} \\ 	 
\eeq & \oprtt{C}{\rta}(\sk,\hh)~.
 	     \tag{assumption, \Cref{thm:op:bellman}}
\end{align*}

For the induction hypothesis, assume that, for all $\hpgcl$ programs $C$, runtimes $\rta \in \T$, and program states $(\sk,\hh)$, we have
\begin{align*}
  \eertt{C}{\rta}(\sk,\hh) \quad\leq\quad \oprtt{C}{\rta}(\sk,\hh)~.	\tag{I.H.}
\end{align*}

In remains to prove the composite cases.

\paragraph{The case $C = \ITE{\varphi}{C_1}{C_2}$.}
\begin{align*}
     & \eertt{\ITE{\varphi}{C_1}{C_2}}{\rta}(\sk,\hh) \\
\eeq & \MROP(\conf{\ITE{\varphi}{C_1}{C_2}}{\sk}{\hh}) \pplus 
       \tag{\Cref{thm:op:eert-bellman}} \\
     &  \sup_{\ma \in \MA}
        \sum_{(\conf{\ITE{\varphi}{C_1}{C_2}}{\sk}{\hh}) \mstep{p}{\ma} (\conf{C'}{\sk'}{\hh'})}
 	     p \cdot \eertt{C'}{\rta}(\sk',\hh')
 	     \\
\lleq & \MROP(\conf{\ITE{\varphi}{C_1}{C_2}}{\sk}{\hh}) \pplus 
       \tag{I.H.} \\
     &  \sup_{\ma \in \MA}
        \sum_{(\conf{\ITE{\varphi}{C_1}{C_2}}{\sk}{\hh}) \mstep{p}{\ma} (\conf{C'}{\sk'}{\hh'})}
 	     p \cdot \oprtt{C'}{\rta}(\sk',\hh')
 	     \\
\eeq & \oprtt{\ITE{\varphi}{C_1}{C_2}}{\rta}(\sk,\hh)~.
 	     \tag{\Cref{thm:op:bellman}} 
\end{align*}

\paragraph{The case $C = \PCHOICE{C_1}{p}{C_2}$.}
\begin{align*}
     & \eertt{\PCHOICE{C_1}{p}{C_2}}{\rta}(\sk,\hh) \\
\eeq & \MROP(\conf{\PCHOICE{C_1}{p}{C_2}}{\sk}{\hh}) \pplus 
       \tag{\Cref{thm:op:eert-bellman}} \\
     &  \sup_{\ma \in \MA}
        \sum_{(\conf{\PCHOICE{C_1}{p}{C_2}}{\sk}{\hh}) \mstep{p}{\ma} (\conf{C'}{\sk'}{\hh'})}
 	     p \cdot \eertt{C'}{\rta}(\sk',\hh')
 	     \\
\lleq & \MROP(\conf{\PCHOICE{C_1}{p}{C_2}}{\sk}{\hh}) \pplus 
       \tag{I.H.} \\
     &  \sup_{\ma \in \MA}
        \sum_{(\conf{\PCHOICE{C_1}{p}{C_2}}{\sk}{\hh}) \mstep{q}{\ma} (\conf{C'}{\sk'}{\hh'})}
 	     p \cdot \oprtt{C'}{\rta}(\sk',\hh')
 	     \\
\eeq & \oprtt{\PCHOICE{C_1}{p}{C_2}}{\rta}(\sk,\hh)~.
 	     \tag{\Cref{thm:op:bellman}} 
\end{align*}

\paragraph{The case $C = \COMPOSE{C_1}{C_2}$.}
\begin{align*}
     & \eertt{\COMPOSE{C_1}{C_2}}{\rta} \\
\eeq & \eertt{C1}{\eertt{C_2}{\rta}} \tag{Def. of $\eert$} \\
\ppreceq & \eertt{C1}{\oprtt{C_2}{\rta}} \tag{I.H.} \\
\ppreceq & \oprtt{C1}{\oprtt{C_2}{\rta}} \tag{I.H.} \\
\ppreceq & \oprtt{\COMPOSE{C_1}{C_2}}{\rta}~. \tag{by \Cref{thm:op:seq}; see further below}
\end{align*}

\paragraph{The case $C = \WHILEDO{\varphi}{C_1}$.}
Recall that
\[
\eertt{\WHILEDO{\varphi}{C_1}}{\rta} 
\eeq \ert{\WHILEDO{\varphi}{C_1}}{\rta}
\eeq \lfp \rtb\mydot ~ \underbrace{\iverson{\neg \varphi} \cdot \rta \pplus \iverson{\varphi} \cdot \ert{C_1}{\rtb}}_{\eeq \Phi(\rtb)}~.
\]
Moreover, let $I = \oprtt{\WHILEDO{\varphi}{C_1}}{\rta}$. 
Then, consider the following:
\begin{align*}
     & \Phi(I) \\
\eeq & \iverson{\neg \varphi} \cdot \rta \pplus \iverson{\varphi} \cdot \ert{C_1}{\rtb} \tag{Def. of $\Phi$} \\
\eeq & \iverson{\neg \varphi} \cdot \rta \pplus \iverson{\varphi} \cdot \eertt{C_1}{\rtb} \tag{Def. of $\eert$} \\
\ppreceq & \iverson{\neg \varphi} \cdot \rta \pplus \iverson{\varphi} \cdot \oprtt{C_1}{\rtb} \tag{I.H.} \\ 
\eeq & I \eeq \oprtt{\WHILEDO{\varphi}{C_1}}{\rta}~. \tag{by \Cref{thm:op:while}; see further below}
\end{align*}
Hence, $I$ is a prefixed point of $\Phi(\rtb)$ and thus 
\begin{align*}
 & \eertt{\WHILEDO{\varphi}{C_1}}{\rta} \\
\eeq & \lfp \rtb\mydot ~ \underbrace{\iverson{\neg \varphi} \cdot \rta \pplus \iverson{\varphi} \cdot \ert{C_1}{\rtb}}_{\eeq \Phi(\rtb)} \\
\ppreceq &  I \eeq \oprtt{\WHILEDO{\varphi}{C_1}}{\rta}~.
\end{align*}
	
\end{proof}

\subsubsection{Auxiliary Lemmas}


\begin{lemma}\label{thm:op:seq}
	$\oprtt{C_1}{\oprtt{C_2}{\rta}} \ppreceq \oprtt{\COMPOSE{C_1}{C_2}}{\rta}$.
\end{lemma}

\newcommand{\auxpsi}[3]{\Psi_{#1}^{#2}\llbracket #3 \rrbracket}

\begin{proof}
We first define an auxiliary transformer $\auxpsi{n}{\rtb}{C}$ that unrolls the Bellman equations at most $n$ times. To this end, let $\auxpsi{0}{\rtb}{C} = 0$.
Moreover, we define 
    \begin{align*}
      \auxpsi{n+1}{\rtb}{C}(\sk,\hh)
      \eeq \begin{cases}
        \rtb(\sk,\hh) & \text{ if } C = \term \\
        \infty & \text{ if } C = \fail \\
      	\MROP(\conf{C}{\sk}{\hh}) \pplus 
        \sup_{\ma \in \MA}
        \sum_{(\conf{C}{\sk}{\hh}) \mstep{p}{\ma} (\conf{C'}{\sk'}{\hh'})} p \cdot \auxpsi{n}{\rtb}{C'}(\sk',\hh')
        & \text{ if } C \in \hpgcl~.
      \end{cases}
    \end{align*}      
We make two observations:
\begin{enumerate}
	\item for all $C \in \hpgcl$, $\sup_{n \in \Nats} \auxpsi{n}{\rtb}{C}$ is Bellman compliant; and
	\item $\oprtt{C}{\rtb} \ppreceq \sup_{n \in \Nats} \auxpsi{n}{\rtb}{C}$.
\end{enumerate}
For observation (1), let $C \in \hpgcl$ and consider the following calculations:
\begin{align*}
	& \sup_{n \in \Nats} \auxpsi{n}{\rtb}{C}(\sk,\hh) \\
	\eeq & \sup_{n \in \Nats} \auxpsi{n+1}{\rtb}{C}(\sk, \hh)
	\tag{by def. $\auxpsi{0}{\rtb}{C} = 0$} \\
	& \sup_{n \in \Nats} 
	\MROP(\conf{C}{\sk}{\hh}) \pplus 
        \sup_{\ma \in \MA}
        \sum_{(\conf{C}{\sk}{\hh}) \mstep{p}{\ma} (\conf{C'}{\sk'}{\hh'})} p \cdot \auxpsi{n}{\rtb}{C'}(\sk',\hh')
        \tag{by def. of $\auxpsi{n+1}{\rtb}{C}$ for $C \in \hpgcl$} \\
    \eeq & 
	\MROP(\conf{C}{\sk}{\hh}) \pplus 
        \sup_{\ma \in \MA}
        \sum_{(\conf{C}{\sk}{\hh}) \mstep{p}{\ma} (\conf{C'}{\sk'}{\hh'})} p \cdot \sup_{n \in \Nats} \auxpsi{n}{\rtb}{C'}(\sk',\hh')~.
        \tag{continuity of $+, \sup, \sum, \cdot$}
\end{align*}
Hence, $\sup_{n \in \Nats} \auxpsi{n}{\rtb}{C}$ is Bellman compliant.
For observation (2), note that observation (1) and \Cref{thm:op:oprt-least} imply 
 \[ \oprtt{C}{\rtb} \ppreceq \sup_{n \in \Nats} \auxpsi{n}{\rtb}{C}~. \]

Assume, for the moment, that we already know the following:    
\begin{align*}
  \forall n \in \Nats \quad 
  \forall C_1, C_2 \in \hpgcl \colon \qquad\qquad 
  \auxpsi{n}{\oprtt{C_2}{\rta}}{C_1}(\sk,\hh) \ppreceq
  \oprtt{\COMPOSE{C_1}{C_2}}{\rta}(\sk,\hh)~. \tag{$\dag$}
\end{align*}
Then we can immediate prove the claim:
\begin{align*}
  & \oprtt{C_1}{\oprtt{C_2}{\rta}}(\sk,\hh)	\\
\ppreceq & \sup_{n \in \Nats} \auxpsi{n}{\oprtt{C_2}{\rta}}{C_1}(\sk,\hh)
	           \tag{by observation (2)} \\
\ppreceq & \oprtt{\COMPOSE{C_1}{C_2}}{\rta}(\sk,\hh)~.  \tag{by ($\dag$)}
\end{align*}

To complete the proof, it remains to show ($\dag$). We proceed by induction on $n$.

\emph{Induction base.} For $n = 0$, consider the following:
\begin{align*}
  & \auxpsi{0}{\oprtt{C_2}{\rta}}{C_1}(\sk,\hh) \\
  \eeq & 0 \tag{by definition} \\
  \ppreceq &
  \oprtt{\COMPOSE{C_1}{C_2}}{\rta}(\sk,\hh)~.
\end{align*}

\emph{Induction hypothesis.} Assume for an arbitrary, but fixed, $n \in \Nats$ that for all $C_1,C_2 \in \hpgcl$ and all program states $(\sk,\hh)$, we have
\begin{align*}
  \auxpsi{n}{\oprtt{C_2}{\rta}}{C_1}(\sk,\hh) \ppreceq 
  \oprtt{\COMPOSE{C_1}{C_2}}{\rta}(\sk,\hh)~. \tag{I.H.}
\end{align*}

\emph{Induction step.} 
Let $C_1, C_2 \in \hpgcl$ and fix some program state $(\sk, \hh)$.
By construction of the execution relation $\mstep{}{}$ (cf., \Cref{fig:op:rules}), 
all $C'$ in configurations ($\conf{C'}{\sk'}{\hh'}$) reached from ($\conf{C_1}{\sk}{\hh}$)
via a \emph{single} execution step are either (1) all $\hpgcl$ programs, (2) all $\term$, or (3) all $\fail$.
Hence, it suffices to show the claim for the following three cases:
\begin{enumerate}
	\item for every step $\conf{C_1}{\sk}{\hh} \mstep{p}{\ma} \conf{C'}{\sk'}{\hh'}$, we have $C' \in \hpgcl$;
	\item for every step $\conf{C_1}{\sk}{\hh} \mstep{p}{\ma} \conf{C'}{\sk'}{\hh'}$, we have  $C' = \term$; and
	\item for every step $\conf{C_1}{\sk}{\hh} \mstep{p}{\ma} \conf{C'}{\sk'}{\hh'}$, we have $C' = \fail$.
\end{enumerate}
For case (1), consider the following:
\begin{align*}
  & \auxpsi{n+1}{\oprtt{C_2}{\rta}}{C_1}(\sk,\hh) \\
  \eeq & 
  \MROP(\conf{C}{\sk}{\hh}) \pplus 
        \sup_{\ma \in \MA}
        \sum_{(\conf{C_1}{\sk}{\hh}) \mstep{p}{\ma} (\conf{C'}{\sk'}{\hh'})} p \cdot \auxpsi{n}{\oprtt{C_2}{\rta}}{C'}(\sk',\hh') 
        \tag{by definition and $C_1 \in \hpgcl$} \\
  \ppreceq &
  \MROP(\conf{C}{\sk}{\hh}) \pplus 
        \sup_{\ma \in \MA}
        \sum_{(\conf{C_1}{\sk}{\hh}) \mstep{p}{\ma} (\conf{C'}{\sk'}{\hh'})} p \cdot 
        	\oprtt{\COMPOSE{C'}{C_2}}{\rta}(\sk',\hh')
        \tag{by I.H.; applicable since $C' \in \hpgcl$} \\
  \eeq &
  \MROP(\conf{C}{\sk}{\hh}) \pplus 
        \sup_{\ma \in \MA}
        \sum_{(\conf{\COMPOSE{C_1}{C_2}}{\sk}{\hh}) \mstep{p}{\ma} (\conf{\COMPOSE{C'}{C_2}}{\sk'}{\hh'})} p \cdot 
        	\oprtt{\COMPOSE{C'}{C_2}}{\rta}(\sk',\hh') \\
  \eeq & 
  \oprtt{\COMPOSE{C_1}{C_2}}{\rta}(\sk,\hh)~. \tag{by \Cref{thm:op:bellman}}
\end{align*}
For case (2), consider the following:
\begin{align*}
  & \auxpsi{n+1}{\oprtt{C_2}{\rta}}{C_1}(\sk,\hh) \\
  \eeq & 
  \MROP(\conf{C}{\sk}{\hh}) \pplus 
        \sup_{\ma \in \MA}
        \sum_{(\conf{C_1}{\sk}{\hh}) \mstep{p}{\ma} (\conf{C'}{\sk'}{\hh'})} p \cdot \auxpsi{n}{\oprtt{C_2}{\rta}}{C'}(\sk',\hh') 
        \tag{by definition and $C_1 \in \hpgcl$} \\
   \eeq &
  \MROP(\conf{C}{\sk}{\hh}) \pplus 
        \sup_{\ma \in \MA}
        \sum_{(\conf{C_1}{\sk}{\hh}) \mstep{p}{\ma} (\conf{\term}{\sk'}{\hh'})} p \cdot \auxpsi{n}{\oprtt{C_2}{\rta}}{\term}(\sk',\hh') 
        \tag{case (2)} \\
   \ppreceq &
  \MROP(\conf{C}{\sk}{\hh}) \pplus 
        \sup_{\ma \in \MA}
        \sum_{(\conf{C_1}{\sk}{\hh}) \mstep{p}{\ma} (\conf{\term}{\sk'}{\hh'})} p \cdot \oprtt{C_2}{\rta}(\sk',\hh') 
        \tag{by definition} \\
        \eeq &
  \MROP(\conf{C}{\sk}{\hh}) \pplus 
        \sup_{\ma \in \MA}
        \sum_{(\conf{\COMPOSE{C_1}{C_2}}{\sk}{\hh}) \mstep{p}{\ma} (\conf{C_2}{\sk'}{\hh'})} p \cdot \oprtt{C_2}{\rta}(\sk',\hh') 
        \tag{by definition} \\
  \eeq & 
  \oprtt{\COMPOSE{C_1}{C_2}}{\rta}(\sk,\hh)~. \tag{by \Cref{thm:op:bellman}}
\end{align*}
For case (3), consider the following:
\begin{align*}
  & \auxpsi{n+1}{\oprtt{C_2}{\rta}}{C_1}(\sk,\hh) \\
  \eeq & 
  \MROP(\conf{C}{\sk}{\hh}) \pplus 
        \sup_{\ma \in \MA}
        \sum_{(\conf{C_1}{\sk}{\hh}) \mstep{p}{\ma} (\conf{C'}{\sk'}{\hh'})} p \cdot \auxpsi{n}{\oprtt{C_2}{\rta}}{C'}(\sk',\hh') 
        \tag{by definition and $C_1 \in \hpgcl$} \\
   \eeq &
  \MROP(\conf{C}{\sk}{\hh}) \pplus 
        \sup_{\ma \in \MA}
        \sum_{(\conf{C_1}{\sk}{\hh}) \mstep{p}{\ma} (\conf{\fail}{\sk'}{\hh'})} p \cdot \auxpsi{n}{\oprtt{C_2}{\rta}}{\fail}(\sk',\hh') 
        \tag{case (3)} \\
   \ppreceq &
  \MROP(\conf{C}{\sk}{\hh}) \pplus 
        \sup_{\ma \in \MA}
        \sum_{(\conf{C_1}{\sk}{\hh}) \mstep{p}{\ma} (\conf{\fail}{\sk'}{\hh'})} p \cdot \oprtt{\fail}{\rta}(\sk',\hh') 
       \tag{by definition} \\
  \eeq & 
  \MROP(\conf{C}{\sk}{\hh}) \pplus 
        \sup_{\ma \in \MA}
        \sum_{(\conf{\COMPOSE{C_1}{C_2}}{\sk}{\hh}) \mstep{p}{\ma} (\conf{\fail}{\sk'}{\hh'})} p \cdot \oprtt{\fail}{\rta}(\sk',\hh') \\
  \eeq & 
  \oprtt{\COMPOSE{C_1}{C_2}}{\rta}(\sk,\hh)~. \tag{by \Cref{thm:op:bellman}}
\end{align*}
In all three cases, we thus have $\auxpsi{n+1}{\oprtt{C_2}{\rta}}{C_1}(\sk,\hh) \ppreceq \oprtt{\COMPOSE{C_1}{C_2}}{\rta}(\sk,\hh)$.
\end{proof}

\begin{lemma}\label{thm:op:while}
	$\oprtt{\WHILEDO{\varphi}{C}}{\rta} \eeq \iverson{\neg \varphi} \cdot \rta \pplus \iverson{\varphi} \cdot \oprtt{\COMPOSE{C}{\WHILEDO{\varphi}{C}}}{\rta}$.
\end{lemma}
\begin{proof}
  We prove the claim separately for each program state $(\sk,\hh)$.
  Two cases arise: $\sk \models \varphi$ and $\sk \not\models \varphi$.
  First, assume $\sk \models \varphi$. Then, consider the following:
  \begin{align*}
  	     & \oprtt{\WHILEDO{\varphi}{C}}{\rta}(\sk,\hh) \\
  	\eeq & \MROP(\conf{\WHILEDO{\varphi}{C}}{\sk}{\hh}) \pplus \\
         &  \sum_{(\conf{\WHILEDO{\varphi}{C}}{\sk}{\hh}) \mstep{}{} (\conf{\COMPOSE{C}{\WHILEDO{\varphi}{C}}}{\sk}{\hh})}
 	     \oprtt{\COMPOSE{C}{\WHILEDO{\varphi}{C}}}{\rta}(\sk,\hh)
 	     \tag{\Cref{thm:op:bellman}, \Cref{fig:op:rules}, assumption} \\
  	\eeq & \oprtt{\COMPOSE{C}{\WHILEDO{\varphi}{C}}}{\rta}(\sk,\hh)
 	       \tag{Def. of $\MROP$, algebra} \\ 
 	\eeq & \left(\iverson{\neg \varphi} \cdot \rta \pplus \iverson{\varphi} \cdot \oprtt{\COMPOSE{C}{\WHILEDO{\varphi}{C}}}{\rta}\right)(\sk,\hh)~.
 	        \tag{algebra, assumption}
  \end{align*}
  Second, assume $\sk \not\models \varphi$. Then, consider the following:
  \begin{align*}
  	     & \oprtt{\WHILEDO{\varphi}{C}}{\rta}(\sk,\hh) \\
  	\eeq & \MROP(\conf{\WHILEDO{\varphi}{C}}{\sk}{\hh}) \pplus \\
         &  \sum_{(\conf{\WHILEDO{\varphi}{C}}{\sk}{\hh}) \mstep{}{} (\conf{\term}{\sk}{\hh})}
 	     \oprtt{\term}{\rta}(\sk,\hh)
 	     \tag{\Cref{thm:op:bellman}, \Cref{fig:op:rules}, assumption} \\
  	\eeq & \rta(\sk,\hh) \tag{Def. of $\MROP$, \Cref{thm:op:bellman}} \\ 
 	\eeq & \left(\iverson{\neg \varphi} \cdot \rta \pplus \iverson{\varphi} \cdot \oprtt{\COMPOSE{C}{\WHILEDO{\varphi}{C}}}{\rta}\right)(\sk,\hh)~.
 	        \tag{algebra, assumption}
  \end{align*}
  
\end{proof}

\subsection{Appendix to Section~\ref{s:aert}}
\subsection{Proof of \Cref{thm:aert_decomp}}
For all $C\in\hpgcl$ all $\arta\in\Api$, and all $\rta \in \T$ we have
\[
	\aert{\potent}{C}{\arta} = \ert{C}{\arta + \potent} - \potent~,
\]
\begin{proof}
	By induction on $C$. \\ 

	\noindent
	For the base cases, the claim holds by definition. 
	
	\paragraph{The case $C=\ITE{\varphi}{C_1}{C_2}$.}
	We have 
	\begin{align*}
	   & \aert{\potent}{\ITE{\varphi}{C_1}{C_2}}{\arta} \\
	   \eeq& \iverson{\varphi} \cdot \aert{\potent}{C_1}{\arta} 
	               + \iverson{\neg \varphi} \cdot \aert{\potent}{C_2}{\arta}  \\
	   \tag{by \Cref{table:aert}} \\
	    \eeq& \iverson{\varphi} \cdot (\ert{C_1}{\arta+ \potent} - \potent) 
	    + \iverson{\neg \varphi} \cdot (\ert{C_2}{\arta + \potent} - \potent )
	    \tag{by I.H.} \\
	    \eeq& (\iverson{\varphi} \cdot \ert{C_1}{\arta+ \potent} 
	    + \iverson{\neg \varphi} \cdot \ert{C_2}{\arta + \potent}) - \potent \\
	     \eeq& \ert{\ITE{\varphi}{C_1}{C_2}}{\arta + \potent} - \potent 
	     \tag{by \Cref{table:ert}}
	\end{align*}
	
	\paragraph{The case $C=\PCHOICE{C_1}{p}{C_2}$.}
	
	We have 
	\begin{align*}
	& \aert{\potent}{\PCHOICE{C_1}{p}{C_2}}{\arta} \\
	\eeq& p \cdot \aert{\potent}{C_1}{\arta} 
	+ (1-p) \cdot \aert{\potent}{C_2}{\arta}  \\
	\tag{by \Cref{table:aert}} \\
	\eeq& p \cdot (\ert{C_1}{\arta+ \potent} - \potent) 
	+ (1-p) \cdot (\ert{C_2}{\arta + \potent} - \potent )
	\tag{by I.H.} \\
	\eeq& (p \cdot \ert{C_1}{\arta+ \potent} 
	+ (1-p) \cdot \ert{C_2}{\arta + \potent}) - \potent \\
	\eeq& \ert{\PCHOICE{C_1}{p}{C_2}}{\arta + \potent} - \potent 
	\tag{by \Cref{table:ert}}
	\end{align*}
	
	\paragraph{The case $C=\WHILEDO{\varphi}{C_1}$.}
	
	First notice that 
	\begin{align*}
	   &\aert{\potent}{C}{\arta} \eeq \ert{C}{\arta + \potent} - \potent \\
	   \text{iff}\quad & \aert{\potent}{C}{\arta} + \potent \eeq \ert{C}{\arta + \potent}\\
	   \text{iff}\quad & \potent + \lfp \aertcharfun{\arta} \eeq \lfp \charfun{\arta+\potent}~,
	\end{align*}
	where $\aertcharfun{\arta}$ and $\charfun{\arta+\potent}$ is the $\aertsymbol$- and the $\ertsymbol$-characteristic functional of the loop $C$, respectively. 
	For the latter statement, it suffices to show that (1) if $\artb= \aertcharfun{\arta}(\artb)$, then $\artb+\potent = \charfun{\arta+\potent}(\artb + \potent)$ (which proves the $\geq$-direction), and (2) if $\rta = \charfun{\arta + \potent}(\rta)$, then $\rta - \potent = \aertcharfun{\arta}(\rta - \potent)$ (which proves the $\leq$-direction).
	
	For (1), consider the following:
	\begin{align*}
	   & \artb \eeq \aertcharfun{\arta}(\artb) \\
	   \text{implies} \quad & \artb \eeq \iverson{\varphi}\cdot \aert{\potent}{C_1}{\artb} + \iverson{\neg\varphi}\cdot \arta
	   \tag{by definition} \\
	   \text{implies} \quad & \artb + \potent \eeq \iverson{\varphi}\cdot \aert{\potent}{C_1}{\artb} + \iverson{\neg\varphi}\cdot \arta + \potent \\
	   \text{implies} \quad & \artb + \potent \eeq \iverson{\varphi}\cdot (\aert{\potent}{C_1}{\artb}+ \potent) + \iverson{\neg\varphi}\cdot (\arta + \potent)  \\
	   \text{implies} \quad & \artb + \potent \eeq \iverson{\varphi}\cdot (\ert{C_1}{\artb +\potent} + \iverson{\neg\varphi}\cdot (\arta + \potent) 
	   \tag{by I.H.} \\
	   \text{implies} \quad & \artb + \potent \eeq \charfun{\arta + \potent}(\artb)
	   \tag{by definition}
	\end{align*}
	
	For (2), consider the following:
	\begin{align*}
	& \rta \eeq \charfun{\arta + \potent}(\rta) \\
	\text{implies} \quad & \rta \eeq \iverson{\varphi}\cdot \ert{C_1}{\rta} + \iverson{\neg\varphi}\cdot (\arta + \potent)
	\tag{by definition} \\
	\text{implies} \quad & \rta - \potent \eeq \iverson{\varphi}\cdot \ert{C_1}{\rta} + \iverson{\neg\varphi}\cdot (\arta + \potent) - \potent
	\tag{by definition} \\
	\text{implies} \quad & \rta - \potent \eeq \iverson{\varphi}\cdot (\ert{C_1}{\rta} - \potent) + \iverson{\neg\varphi}\cdot \arta \\
	\text{implies} \quad & \rta - \potent \eeq \iverson{\varphi}\cdot (\aert{\potent}{C_1}{\rta - \potent} + \iverson{\neg\varphi}\cdot \arta 
	\tag{by I.H.} \\
	\text{implies} \quad & \rta - \potent \eeq \aertcharfun{\arta}(\rta - \potent)~.
	\tag{by definition}
	\end{align*}
	
	This completes the proof.
	
\end{proof}

\subsection{Appendix to the Insert-Detelte-FindAny Case Study}
\label{app:findany}

\begin{figure}[t]

		\begin{align*}
			&\unframeannotate{(\rta' + \potent) \sepadd \rtb ~~~-~~~\potent}{} \\
			& \annotate{\rta'} \tag{obtain $\rta' \succeq \aert{\potent}{C}{\rta-\potent}$} \\
			& \vdots \tag{further annotations} \\
			&C \\
			&\frameannotate{\rta - \potent}{} \tag{$\rtb$ is the frame} \\
			&\annotate{\rta \sepadd \rtb ~~~-~~~ \potent} \tag{post $\potent$-runtime is $\rta \sepadd \rtb - \potent$}
		\end{align*}%

	\caption{$\aertsymbol$ annotations with framing where $\modC{C}\cap \Vars(\rtb) = \emptyset$.}
	\label{fig:aert-annotations_framing}
\end{figure}

Throughout, we use a source-code annotation style analogous to the annotations described in \Cref{s:aert} and \Cref{app:annotations}. Framing, i.e.\ \Cref{thm:aert_frame}, is annotated as shown in \Cref{fig:aert-annotations_framing}. To verify the amortized expected runtimes of $\procinsert{\varadd}$ and $\procremove{\varremove}$, we proceed as follows:
We first obtain an upper bound on $\aert{\potent}{\procrank}{0}$ in \Cref{fig:procrank_1} and 
\Cref{fig:procrank_2}. We then obtain an upper bound on 
\[\aert{\potent}{\procsample}{\aert{\potent}{\procrank}{0}} = \aert{\potent}{\COMPOSE{\procsample}{\procrank}}{0}
\]
 in \Cref{fig:procsample_1,fig:procsample_2}. Then, denoting this upper bound by $\arta$, we are in a position to verify $\procdelete{\varremove}$ and $\procinsert{\varadd}$ in \Cref{fig:findany_delete,fig:findany_insert}, where in \Cref{fig:findany_delete}, we make use of the fact that
 \begin{align}
 	\label{eqn:findany1}
 	& \potent + \arta  \\
 	\eeq & \potent  -\potent + \pureemp{\ls \geq 1} \sepadd \emprun{\ls} \sepadd \nicefrac{1}{\ls}\cdot \sum_{i=1}^{\ls} \Inf \ls_2,\lend,z,w \colon  \dll{\lh}{z}{0}{i}{w}   \\
     &\quad \sepadd \dll{w}{\lend}{z}{\ls_2}{0} \sepadd \potent\subst{\varany}{z} \sepadd \pureemp{i\geq 1 \wedge  \ls_2 + i =\ls} 
     \notag\\
     \eeq & \pureemp{\ls \geq 1} \sepadd \emprun{\ls} \sepadd \nicefrac{1}{\ls}\cdot \sum_{i=1}^{\ls} \Inf \ls_2,\lend,z,w \colon \dll{\lh}{z}{0}{i}{w}   \\
     &\quad \sepadd \dll{w}{\lend}{z}{\ls_2}{0} \sepadd \potent\subst{\varany}{z} \sepadd \pureemp{i\geq 1 \wedge  \ls_2 + i =\ls}
     \notag \\
     \ppreceq & \pureemp{\ls \geq 1} \sepadd \emprun{\ls} \sepadd \nicefrac{1}{\ls}\cdot \sum_{i=1}^{\ls} \Inf \ls_2,\lend,z,w \colon \dll{\lh}{z}{0}{i}{w}   \\
     &\quad \sepadd \dll{w}{\lend}{z}{\ls_2}{0} \sepadd (\ls \cdot (1 + \iverson{z = \varremove}))\sepadd \pureemp{i\geq 1 \wedge  \ls_2 + i =\ls} 
     \tag{definition of $\potent$}\\
     \ppreceq & \pureemp{\ls \geq 1} \sepadd \emprun{\ls}
     \sepadd \emprun{\nicefrac{1}{\ls} \cdot ((\ls-1)\cdot \ls + 2\cdot\ls) }
      \sepadd \nicefrac{1}{\ls}\cdot \sum_{i=1}^{\ls} \Inf \ls_2,\lend,z,w \colon \dll{\lh}{z}{0}{i}{w}   \\
     &\quad \sepadd \dll{w}{\lend}{z}{\ls_2}{0} \sepadd \pureemp{i\geq 1 \wedge  \ls_2 + i =\ls} 
     \tag{$z = \varremove$ holds for at most one summand}\\
 	\ppreceq&2+ \pureemp{\ls \geq 1} \sepadd \emprun{2\cdot \ls} \sepadd \Inf \lend \colon
 	\dll{\lh}{\lend}{0}{\ls}{0} 
 	\notag
 \end{align}
as well as
\begin{align}
	\label{eqn:findany2}
	& \iverson{\varremove = \varany} \cdot (-\potent + 2+ \pureemp{\ls \geq 1} \sepadd \emprun{2\cdot \ls} \sepadd \Inf \lend \colon
	\dll{\lh}{\lend}{0}{\ls}{0} )\\
	\eeq& \iverson{\varremove = \varany} \cdot (-(\ls \cdot (1+\iverson{\varremove = \varany}) ) + 2+ \pureemp{\ls \geq 1} \sepadd \emprun{2\cdot \ls} \sepadd \Inf \lend \colon
	\dll{\lh}{\lend}{0}{\ls}{0} ) 
	\notag\\
	\eeq& \iverson{\varremove = \varany} \cdot (-2\cdot\ls + 2+ \pureemp{\ls \geq 1} \sepadd \emprun{2\cdot \ls} \sepadd \Inf \lend \colon
	\dll{\lh}{\lend}{0}{\ls}{0} ) 
	\notag\\
	\eeq& \iverson{\varremove = \varany} \cdot ( 2+ \pureemp{\ls \geq 1} \sepadd \Inf \lend \colon
	\dll{\lh}{\lend}{0}{\ls}{0} )
	\notag
\end{align}

We moreover define
 \[
    \ert{\UNIFASSIGN{x}{e}{e'}}{\rta} 
    \eeq 
    \iverson{e \leq e'} \cdot \nicefrac{1}{e'-e+1} \cdot \sum\limits_{i=e}^{e'} \rta\subst{x}{i}~,
 \]
 where we require that $\sk(e),\sk(e') \in \Nats$ for all stacks $\ss$, which can be shown to be syntactic sugar \cite{QSLpopl}. Using \Cref{thm:aert_decomp}, we obtain $\aertsymbol$ versions of the local rules from \Cref{thm:local_rules}:
 
 \begin{theorem}
 	Let $C \in \hpgcl$. Then:
 	
 	\begin{enumerate}
 		\item \textnormal{(mut):} \qquad $\aert{\potent}{\HASSIGN{e}{e'}}{\isingleton{e}{e'} - \potent} \ppreceq \ivalidpointer{e} - \potent$
 		\item \textnormal{(lkp):} \qquad 
 		%
 		$\aert{{\potent}}{\ASSIGNH{x}{e}}{\pureemp{x=z} \sepadd \isingleton{e\subst{x}{y}}{z} - \potent} \ppreceq \pureemp{x=y} \sepadd \isingleton{e}{z} - \potent$,
 		%
 		\item \textnormal{(alc):} if $x$ does not occur in $e$, then
 		\[\aert{{\potent}}{\ALLOC{x}{e}}{\bigoplus_{i=1}^{e} \isingleton{x+i-1}{0} - {\potent}} \ppreceq \iemp - {\potent}~.
 		\]
 		\item \textnormal{(aux):} For all $\rta, \rtb \in \T$ and all $y \in \Vars$ not occurring in $C$, 
 		\begin{align*}
 		\aert{\potent}{C}{\rta - \potent} \ppreceq \rtb - \potent \qqimplies \aert{\potent}{C}{-\potent + \Inf y\colon \rta} \ppreceq -\potent + \Inf y \colon \rtb~.
 		\end{align*}
 	\end{enumerate}
 \end{theorem}

\begin{figure}[!p]
	\begin{align*}
	%
	%
	&\asucceqannotate{ \isingleton{\varany}{-,-,-}\sepadd 0+ \Inf \lend \colon  \dll{\lh}{\lend}{0}{\ls}{0} \sepadd \emprun{\ls}} \\
	&\aaertannotate{ \isingleton{\varany}{-,-,-}\sepadd 0+ \Inf \lend, \ls_2 \colon  \dll{\lh}{\lend}{0}{\ls_2}{0} \sepadd \emprun{\ls_2}} \\
	& \ASSIGN{\varrank}{1}\fatsemi \\
	&\asucceqannotate{\isingleton{\varany}{-,-,-}\sepadd 0+ \Inf \lend, \ls_2 \colon  \dll{\lh}{\lend}{0}{\ls_2}{0} \sepadd \emprun{\ls_2}} \\
	&\asucceqannotate{\ivalidpointer{\varany+2}\sepadd 0 + \Inf \lend, \ls_2 \colon  \dll{\lh}{\lend}{0}{\ls_2}{0} \sepadd \emprun{\ls_2}} \\
	%
	%
	&\ASSIGNH{\varvalany}{\varany +2} \fatsemi \\
	%
	%
	&\asucceqannotate{\Inf \lend, \ls_2 \colon  \dll{\lh}{\lend}{0}{\ls_2}{0} \sepadd \emprun{\ls_2}} \\
	&\aaertannotate{\Inf v,\lend, \ls_1,\ls_2 \colon \dll{\lh}{v}{0}{\ls_1}{\lh} \sepadd \dll{\lh}{\lend}{v}{\ls_2}{0} \sepadd \emprun{\ls_2}} \\
	&\ASSIGN{\varcur}{\lh}\fatsemi \\
	&\asucceqannotate{\Inf v,\lend, \ls_1,\ls_2 \colon \dll{\lh}{v}{0}{\ls_1}{\varcur} \sepadd \dll{\varcur}{\lend}{v}{\ls_2}{0} \sepadd \emprun{\ls_2}} \\
	&\psiannotate{\ldots}\\
	&\WHILE{\varcur \neq 0} \\
	&\qquad \asucceqannotate{\Inf v,\lend,\ls_1,\ls_2 \colon  \dll{\lh}{v}{0}{\ls_1}{\varcur} \sepadd  \dll{\varcur}{\lend}{v}{\ls_2}{0} \sepadd \emprun{\ls_2}} \\
	&\qquad \aauxvarannotate{\Inf w,v',v,\lend,\ls_1,\ls_2 \colon \emprun{1} \sepadd \dll{\lh}{w}{0}{\ls_1}{v'} \sepadd \isingleton{\varcur}{v,w,-} \sepadd \pureemp{\varcur=v'}} \\
	&\qquad\aphantannotate{\quad {}\sepadd \dll{v}{\lend}{v'}{\ls_2}{0} \sepadd \emprun{\ls_2} \sepadd\emprun{\potent} - \potent}{} \\
	&\qquad \aaertannotate{\emprun{1} \sepadd \dll{\lh}{w}{0}{\ls_1}{v'} \sepadd \isingleton{\varcur}{v,w,-} \sepadd \pureemp{\varcur=v'}\sepadd \dll{v}{\lend}{v'}{\ls_2}{0} \sepadd \emprun{\ls_2}  \sepadd\emprun{\potent} - \potent}{} \\
	&\qquad \TICK{1}\fatsemi \\
	&\qquad \aannotate{\ldots \textnormal{cf.\ \Cref{fig:procrank_2}}}
	\end{align*}
	
	\caption{Program $\procrank$ (part 1).}
	\label{fig:procrank_1}
\end{figure}

\begin{figure}[!p]
	\begin{align*}
       &\qquad \ldots \\
	   & \qquad\aannotate{\ldots \textnormal{cf.\ \Cref{fig:procrank_1}}} \\
	  &\qquad \TICK{1}\fatsemi \\
	   &\qquad \aannotate{\dll{\lh}{w}{0}{\ls_1}{v'} \sepadd \isingleton{\varcur}{v,w,-} \sepadd \pureemp{\varcur=v'}\sepadd \dll{v}{\lend}{v'}{\ls_2}{0}}{} \\
	   &\qquad \aphantannotate{\quad{}\sepadd \dll{v}{\lend}{v'}{\ls_2}{0} \sepadd \emprun{\ls_2} \sepadd\emprun{\potent} - \potent} \\
	  &\qquad \ASSIGNH{\varvalcur}{\varcur + 2}\fatsemi \\
	  &\qquad \aaertannotate{\dll{\lh}{w}{0}{\ls_1}{v'} \sepadd \isingleton{\varcur}{v,w,-} \sepadd \pureemp{\varcur=v'}\sepadd \dll{v}{\lend}{v'}{\ls_2}{0}}{} \\
	  &\qquad \aphantannotate{\quad{}\sepadd \dll{v}{\lend}{v'}{\ls_2}{0} \sepadd \emprun{\ls_2} \sepadd\emprun{\potent} - \potent} \\
	  &\qquad \IF{\varvalcur < \varvalany} \\
	  & \qquad \qquad \ASSIGN{\varrank}{\varrank+1} \\
	  &\qquad \} \\
	   &\qquad \aunframeannotate{\dll{\lh}{w}{0}{\ls_1}{v'} \sepadd \isingleton{\varcur}{v,w,-} \sepadd \pureemp{\varcur=v'}\sepadd \dll{v}{\lend}{v'}{\ls_2}{0}}{} \\
	   &\qquad \aphantannotate{\quad{}\sepadd \dll{v}{\lend}{v'}{\ls_2}{0}\sepadd \emprun{\ls_2} \sepadd\emprun{\potent} - \potent} \\
	   &\qquad \alkpannotate{\isingleton{\varcur}{v} \sepadd \pureemp{\varcur=v'} - \potent}{} \\
	   &\qquad\ASSIGNH{\varcur}{\varcur} \\
	   &\qquad \aframeannotate{\isingleton{v'}{v} \sepadd \pureemp{\varcur=v} - \potent}{} \\
	   &\qquad \aauxvarannotate{\dll{\lh}{w}{0}{\ls_1}{v'} \sepadd \isingleton{v'}{v,w,-} \sepadd \pureemp{\varcur=v}\sepadd \dll{v}{\lend}{v'}{\ls_2}{0} \sepadd \emprun{\ls_2} \sepadd\emprun{\potent} - \potent} \\
	   &\qquad \asucceqannotate{\Inf w,v',v,\lend,\ls_1,\ls_2 \colon \dll{\lh}{w}{0}{\ls_1}{v'} \sepadd \isingleton{v'}{v,w,-} \sepadd \pureemp{\varcur=v} \sepadd \emprun{\ls_2} }\\
	   &\qquad \aphantannotate{\quad {}\sepadd \dll{v}{\lend}{v'}{\ls_2}{0}\sepadd\emprun{\potent} - \potent} \\
	   & \qquad \astarannotate{ \Inf v,\lend, \ls_1,\ls_2 \colon \dll{\lh}{v}{0}{\ls_1}{\varcur} \sepadd \dll{\varcur}{\lend}{v}{\ls_2}{0} \sepadd \emprun{\ls_2} } \\
	   &\}\\
	   &\aannotate{0} 
	\end{align*}
	
	\caption{Program $\procrank$ (part 2).}
	\label{fig:procrank_2}
\end{figure}

\begin{figure}[!p]
	\begin{align*}
		%
		&\aunframeannotate{-\potent + \pureemp{\ls \geq 1} \sepadd \emprun{\ls} \sepadd \nicefrac{1}{\ls}\cdot \sum_{i=1}^{\ls} \Inf \ls_2,\lend,z,w \colon \dll{\lh}{z}{0}{i}{w}}{}   \\
		&\aphantannotate{\quad {} \sepadd \dll{w}{\lend}{z}{\ls_2}{0} \sepadd \potent\subst{\varany}{z} \sepadd \pureemp{i\geq 1 \wedge  \ls_2 + i =\ls}}  \\
		& \aaertannotate{-\potent + \iverson{\ls \geq 1} \cdot \nicefrac{1}{\ls}\cdot \sum_{i=1}^{\ls} \Inf \ls_2,\lend,z,w \colon \sepadd \dll{\lh}{z}{0}{i}{w}}   \\
		&\aphantannotate{\quad {} \sepadd \dll{w}{\lend}{z}{\ls_2}{0} \sepadd \potent\subst{\varany}{z} \sepadd \pureemp{i\geq 1 \wedge  \ls_2 + i =\ls}}  \\
		& \UNIFASSIGN{\varsamplen}{1}{\ls} \fatsemi \\
		& \aaertannotate{\Inf \ls_2,\lend,z,w \colon \sepadd \dll{\lh}{z}{0}{\varsamplen}{w}}   \\
		&\aphantannotate{\quad {} \sepadd \dll{w}{\lend}{z}{\ls_2}{0} \sepadd \potent\subst{\varany}{z} \sepadd \pureemp{\varsamplen\geq 1 \wedge  \ls_2 + \varsamplen =\ls}- \potent}  \\
		&\ASSIGN{\varcur}{\varsamplen} \fatsemi \\
		& \asucceqannotate{\Inf \ls_2,\lend,z,w \colon \sepadd \dll{\lh}{z}{0}{\varcur}{w}}   \\
		&\aphantannotate{\quad {} \sepadd \dll{w}{\lend}{z}{\ls_2}{0} \sepadd \potent\subst{\varany}{z} \sepadd \pureemp{\varcur \geq 1 \wedge  \ls_2 + \varcur =\ls}- \potent}  \\
		& \aaertannotate{\Inf \ls_1,\ls_2,\lend,z,z',w \colon  \dll{\lh}{z}{0}{\ls_1}{\lh}} \\
		&\aphantannotate{\quad {}\sepadd \dll{\lh}{z'}{z}{\varcur}{w} \sepadd \dll{w}{\lend}{z}{\ls_2}{0} \sepadd \potent\subst{\varany}{z'}} \\
		& \aphantannotate{\quad {}\sepadd \pureemp{\varcur \geq 1 \wedge \ls_1 + \ls_2 + \varcur =\ls}- \potent}  \\
		&\ASSIGN{\varany}{\lh}\fatsemi \\
		& \asucceqannotate{\Inf \ls_1,\ls_2,\lend,z,z',w \colon  \dll{\lh}{z}{0}{\ls_1}{\varany}} \\
		&\aphantannotate{\quad {}\sepadd \dll{\varany}{z'}{z}{\varcur}{w} \sepadd \dll{w}{\lend}{z}{\ls_2}{0} \sepadd \potent\subst{\varany}{z'}} \\
		& \aphantannotate{\quad {}\sepadd \pureemp{\varcur \geq 1 \wedge \ls_1 + \ls_2 + \varcur =\ls} - \potent}  \\
		&\psiannotate{\ldots} \\
		&\WHILE{\varcur > 1} \\
		%
		&\qquad \asucceqannotate{\Inf \ls_1,\ls_2,\lend,z,z',w \colon  \dll{\lh}{z}{0}{\ls_1}{\varany}}{} \\
		&\qquad \aphantannotate{\quad {}\sepadd \dll{\varany}{z'}{z}{\varcur }{w} \sepadd \dll{w}{\lend}{z}{\ls_2}{0} \sepadd \potent\subst{\varany}{z'} - \potent} \\
		&\qquad \aannotate{\ldots \textnormal{cf.\ \Cref{fig:procsample_2}}}  \\
		&\qquad \ASSIGNH{\varany}{\varany} \fatsemi \\
		&\qquad \ASSIGN{\varcur}{\varcur -1} \\
		&\} \\
	\end{align*}
	
	\caption{Program $\procsample$ (part 1).}
	\label{fig:procsample_1}
\end{figure}

\begin{figure}[!p]
	\begin{align*}
	%
	& \ldots \\
	&\WHILE{\varcur > 1} \\
	%
    &\qquad \aannotate{\ldots \textnormal{cf.\ \Cref{fig:procsample_1}}} \\
   %
    &\qquad \aauxvarannotate{\Inf \ls_1,\ls_2,\lend,z,z',w,v,v' \colon  \dll{\lh}{z}{0}{\ls_1}{v'} \sepadd \pureemp{v'=\varany} \sepadd \isingleton{\varany}{v,z,-}}{} \\
    &\qquad \aphantannotate{\quad {}\sepadd \dll{v}{z'}{v'}{\varcur - 1}{w} \sepadd \dll{w}{\lend}{z}{\ls_2}{0} \sepadd \potent\subst{\varany}{z'}} \\
    &\qquad \aphantannotate{\quad {}\sepadd \pureemp{\varcur \geq 2 \wedge \ls_1 + \ls_2 + \varcur  =\ls}- \potent}  \\
    &\qquad \aunframeannotate{ \dll{\lh}{z}{0}{\ls_1}{v'} \sepadd \pureemp{v'=\varany} \sepadd \isingleton{\varany}{v,z,-}}{} \\
    &\qquad \aphantannotate{\quad {}\sepadd \dll{v}{z'}{v'}{\varcur - 1}{w} \sepadd \dll{w}{\lend}{z}{\ls_2}{0} \sepadd \potent\subst{\varany}{z'}} \\
    &\qquad \aphantannotate{\quad {}\sepadd \pureemp{\varcur \geq 2 \wedge \ls_1 + \ls_2 + \varcur  =\ls}- \potent}  \\
	& \qquad \alkpannotate{\pureemp{v'=\varany} \sepadd \isingleton{\varany}{v} - \potent}{} \\
	&\qquad \ASSIGNH{\varany}{\varany} \fatsemi \\
	%
	%
	& \qquad \aframeannotate{\pureemp{v=\varany} \sepadd \isingleton{v'}{v} - \potent}{} \\
	& \qquad \aauxvarannotate{ \dll{\lh}{z}{0}{\ls_1}{v'} \sepadd \pureemp{v=\varany} \sepadd \isingleton{v'}{v,z,-}} \\
	&\qquad \aphantannotate{\quad {}\sepadd \dll{v}{z'}{v'}{\varcur - 1}{w} \sepadd \dll{w}{\lend}{z}{\ls_2}{0} \sepadd \potent\subst{\varany}{z'}} \\
	&\qquad \aphantannotate{\quad {}\sepadd \pureemp{\varcur \geq 2 \wedge \ls_1 + \ls_2 + \varcur  =\ls} - \potent}  \\
	& \qquad \asucceqannotate{\Inf \ls_1,\ls_2,\lend,z,z',w,v,v' \colon  \dll{\lh}{z}{0}{\ls_1}{v'} \sepadd \pureemp{v=\varany} \sepadd \isingleton{v'}{v,z,-}} \\
	&\qquad \aphantannotate{\quad {}\sepadd \dll{v}{z'}{v'}{\varcur - 1}{w} \sepadd \dll{w}{\lend}{z}{\ls_2}{0} \sepadd \potent\subst{\varany}{z'}} \\
	&\qquad \aphantannotate{\quad {}\sepadd \pureemp{\varcur \geq 2 \wedge \ls_1 + \ls_2 + \varcur  =\ls}- \potent}  \\
	& \qquad \aaertannotate{\Inf \ls_1,\ls_2,\lend,z,z',w \colon  \dll{\lh}{z}{0}{\ls_1}{\varany}} \\
	&\qquad \aphantannotate{\quad {}\sepadd \dll{\varany}{z'}{z}{\varcur - 1}{w} \sepadd \dll{w}{\lend}{z}{\ls_2}{0} \sepadd \potent\subst{\varany}{z'}} \\
	&\qquad \aphantannotate{\quad {}\sepadd \pureemp{\varcur \geq 2 \wedge \ls_1 + \ls_2 + \varcur - 1 =\ls}- \potent}  \\
	&\qquad \ASSIGN{\varcur}{\varcur -1} \\
	& \qquad \astarannotate{\Inf \ls_1,\ls_2,\lend,z,z',w \colon  \dll{\lh}{z}{0}{\ls_1}{\varany}} \\
	&\qquad \aphantannotate{\quad {}\sepadd \dll{\varany}{z'}{z}{\varcur}{w} \sepadd \dll{w}{\lend}{z}{\ls_2}{0} \sepadd \potent\subst{\varany}{z'}} \\
	&\qquad \aphantannotate{\quad {}\sepadd \pureemp{\varcur \geq 1 \wedge \ls_1 + \ls_2 + \varcur =\ls} - \potent}  \\
	%
	&\} \\
	&\aframeannotate{\Inf \lend,z,\ls_1,\ls_2,w,w' \colon \pureemp{\ls_1+\ls_2+1 = \ls}\sepadd \dll{\lh}{z}{0}{\ls_1}{\varany}}{} \\
	&\aphantannotate{\quad {}\sepadd \isingleton{\varany}{w,w',-}\sepadd \dll{w}{\lend}{\varany}{\ls_2}{0} } \\
	&\asucceqannotate{\emprun{\ls} \sepadd \Inf \lend,z,\ls_1,\ls_2,w,w' \colon \pureemp{\ls_1+\ls_2+1 = \ls}\sepadd \dll{\lh}{z}{0}{\ls_1}{\varany}} \\
	&\aphantannotate{\quad {}\sepadd \isingleton{\varany}{w,w',-}\sepadd \dll{w}{\lend}{\varany}{\ls_2}{0}  \sepadd\emprun{\potent} - \potent} \\
	&\aannotate{ \isingleton{\varany}{-,-,-}\sepadd 0+ \Inf \lend \colon  \dll{\lh}{\lend}{0}{\ls}{0} \sepadd \emprun{\ls}}
	\end{align*}
	
	\caption{Program $\procsample$ (part 2).}
	\label{fig:procsample_2}
\end{figure}

\begin{figure}[!p]
		\begin{align*}
		&\aaertannotate{ 1 + \Inf \lend \colon \dll{\lh}{\lend}{0}{\ls}{0} 
			+ \isingleton{\varremove}{-,-,-}\sepadd 0 + \isingleton{\varany}{-,-,-}\sepadd 0} \\
		& \procremove{x}\fatsemi \\
		%
		%
		&\asucceqannotate{2 + \iverson{\ls=0}\cdot\isingleton{\varremove}{-,-,-}  +\iverson{\ls \geq 1}\cdot ( \iverson{\varremove= \varany} \cdot (\Inf \lend \colon \dll{\lh}{\lend}{0}{\ls}{0} \sepadd \isingleton{\varremove}{-,-,-})} \\
		&\aphantannotate{\quad{}+\iverson{\varremove \neq \varany}\cdot (\isingleton{\varremove}{-,-,-} \sepadd\isingleton{\varany}{-,-,-} \sepadd 0))} \\
		& \IF{\ls = 0} \\
		&\qquad \asucceqannotate{\isingleton{\varremove}{-,-,-} \sepadd 0 + 2\cdot\ls}  \\
		&\qquad \aaertannotate{\isingleton{\varremove}{-,-,-} \sepadd 0 + \potent\subst{\varany}{0} - \potent}  \\
		&\qquad \COMPOSE{\ASSIGN{\varany}{0}}{\ASSIGN{\varrank}{0}}  \\
		&\qquad \aannotate{\isingleton{\varremove}{-,-,-} \sepadd 0 } \\
		&\ELSE  \\
		& \qquad \aaertannotate{\iverson{\varremove = \varany} \cdot \big(2+ \pureemp{\ls \geq 1} \sepadd \Inf \lend \colon \dll{\lh}{\lend}{0}{\ls}{0} \sepadd \isingleton{\varremove}{-,-,-} \big)}  \\
		&\qquad\aphantannotate{\quad {}+ \iverson{\varremove \neq \varany} \cdot \big( \emprun{1} \sepadd \isingleton{\varremove}{-,-,-} \sepadd{}\isingleton{\varany}{-,-,-} \sepadd 0 \big)}
		\tag{\Cref{eqn:findany2}} \\
		& \qquad \IF{\varremove = \varany} \\
		& \qquad \qquad \asucceqannotate{-\potent + 2+\pureemp{\ls \geq 1} \sepadd \emprun{2\cdot \ls} \sepadd \Inf \lend \colon
			\dll{\lh}{\lend}{0}{\ls}{0} \sepadd \isingleton{\varremove}{-,-,-}}  
		\tag{\Cref{eqn:findany1}}\\
		%
		%
		&\qquad \qquad \aunframeannotate{(\arta+\potent) \sepadd\isingleton{\varremove}{-,-,-} - \potent}\\
		&\qquad \qquad \COMPOSE{\procsample}{\procrank} \\
		%
		%
		&\qquad \qquad \asucceqannotate{0} \\
		&\qquad \qquad \aframeannotate{0\sepadd \emprun{\potent} - \potent} \\
		&\qquad \qquad\asucceqannotate{\isingleton{\varremove}{-,-,-} \sepadd 0 \sepadd \emprun{\potent} - \potent}\\
		&\qquad \qquad \aannotate{\isingleton{\varremove}{-,-,-} \sepadd 0 } \\
		&\qquad \ELSE \\
		& \qquad \qquad \asucceqannotate{\emprun{1} \sepadd \isingleton{\varremove}{-,-,-} \sepadd\isingleton{\varany}{-,-,-} \sepadd 0 }  \\
		& \qquad \qquad \COMPOSE{\COMPOSE{\ASSIGNH{\varvalrem}{\varremove+2}}{\ASSIGNH{\varvalany}{\varany+2}}}{\TICK{1}}\fatsemi \\
		&\qquad \qquad \IF{\varvalrem < \varvalany} \\
		&\qquad \qquad \qquad \ASSIGN{\varrank}{\varrank \monus 1}~\} \\
		&\qquad \qquad \aannotate{\isingleton{\varremove}{-,-,-} \sepadd 0 } \\
		&\qquad \} \\
		&\qquad \aannotate{\isingleton{\varremove}{-,-,-} \sepadd 0 } \\
		&\} \\
		&\aaertannotate{\isingleton{\varremove}{-,-,-} \sepadd 0}\\
		&\FREE{\varremove,\varremove+1, \varremove+2}\\
		&\aannotate{0}
		\end{align*}
	
	\caption{$\aertsymbol$-annotations for $\procdelete{\varremove}$. Here $\rta$ is the $\potent$-runtime obtained in \Cref{fig:procsample_1}.}
	\label{fig:findany_delete}
\end{figure}

\begin{figure}[!p]
	
		\begin{align*}
			&\aaertannotate{4 + \Inf \lend \colon \dll{\lh}{\lend}{0}{\ls}{0}+ \iverson{\ls \geq 1} \cdot(\isingleton{\varany}{-,-,-}\sepadd 0) } \\
			&\procadd{\varadd} \fatsemi \\
			&\asucceqannotate{2 + \iiverson{\ls \geq 1} \sepadd \Inf \lend \colon \dll{\lh}{\lend}{0}{\ls}{0} +\iverson{\ls \geq 2}\cdot(\isingleton{\varany+2}{-}\sepadd 0)}  \\
			&\aaertannotate{\nicefrac{1}{\ls+1} \cdot (2\cdot\ls  + \iiverson{\ls \geq 1}
			\sepadd \Inf \lend \colon \dll{\lh}{\lend}{0}{\ls}{0} )} \\
		    &\quad \aphantannotate{{}+ (1-\nicefrac{1}{\ls+1}) \cdot (1 + \iverson{\ls \geq 2}\cdot(\isingleton{\varany+2}{-}\sepadd 0))} \\
			&\{ \\
			&\qquad \asucceqannotate{2\cdot \ls  + \iiverson{\ls \geq 1}
				\sepadd \Inf \lend \colon \dll{\lh}{\lend}{0}{\ls}{0} } \\
			%
			%
			%
			&\qquad \aaertannotate{\isingleton{\lh}{-,-,-}\sepadd 0+ (\Inf \lend \colon \dll{\lh}{\lend}{0}{\ls}{0} \sepadd \emprun{\ls})} \\
			&\qquad\aphantannotate{\quad {}+ \ls \cdot (1 + \iverson{\lh = \varremove}) - \ls \cdot (1 + \iverson{\varany = \varremove})} \\
			&\qquad \ASSIGN{\varany}{\lh} \fatsemi \\
			%
			%
			&\qquad \asucceqannotate{\isingleton{\varany}{-,-,-}\sepadd 0+ \Inf \lend \colon  \dll{\lh}{\lend}{0}{\ls}{0} \sepadd \emprun{\ls}}\\
			&\qquad \procrank \\
			&\qquad \aannotate{0} \\
			&\}~[\nicefrac{1}{\ls+1}]~\{ \\
			&\qquad \asucceqannotate{1 + \iverson{\ls \geq 2}\cdot(\isingleton{\varany+2}{-}\sepadd 0)} \\
			&\qquad \IF{\ls \geq 2} \\
			&\qquad \qquad \asucceqannotate{1 + \isingleton{\varany+2}{-} \sepadd 0} \\
			&\qquad\qquad \ASSIGNH{\varvalany }{\varany+2}\fatsemi \\
			&\qquad \qquad \aaertannotate{1} \\
			&\qquad \qquad \TICK{1}\fatsemi \\
			&\qquad \qquad \aaertannotate{0} \\
			&\qquad\qquad \IF{\varadd < \varvalany} \\
			&\qquad\qquad\qquad \ASSIGN{\varrank}{\varrank +1} \\
			&\qquad \qquad \} \\
			&\qquad \qquad \aaertannotate{0} \\
			&\qquad \ELSE \\
			&\qquad\qquad  \asucceqannotate{\ls} \\
			&\qquad \qquad \aaertannotate{\ls \cdot (1 + \iverson{\lh = \varremove}) - \ls \cdot (1 + \iverson{\varany = \varremove})} \\
			&\qquad \qquad \COMPOSE{\ASSIGN{\varany}{\lh}}{\ASSIGN{\varrank}{1}} \\
			&\qquad \qquad \aaertannotate{0} \\
			&\qquad \} \\
			&\qquad \aaertannotate{0} \\
			&\} \\
			&\aaertannotate{0}
		\end{align*}
	
	\caption{$\aertsymbol$-annotations for $\procinsert{\varadd}$.}
	\label{fig:findany_insert}
\end{figure}

\subsection{Appendix to the Randomized Dynamic List Case Study}
\label{app:dynlist}

See \Cref{fig:dynamic_list_aert_proof_1} and \ref{fig:dynamic_list_aert_proof_2} for detailed source-code annotations, where we define 
\[
\earray{\ahead}{\asize} \eeq \bigoplus_{i=1}^{\asize} \ivalidpointer{\ahead + i -1} ~,
\]
for the sake of readability and assume the following specifications of the involved subprograms as described in \Cref{sec:dynlist}:

\begin{itemize}
	\item $\aert{\potent}{\ALLOC{\ahead'}{\asize'}}{\earray{\ahead'}{\asize'} -\potent } \preceq \iemp - \potent$
	\item \begin{align*}
	&\aert{\potent}{\arraycopy{\ahead}{\asize}{\ahead'}}{\earray{\ahead}{\asize} \sepadd \earray{\ahead'}{\asize'}-\potent} \\  
	{}\preceq \qquad &  \emprun{\asize} \sepadd \pureemp{\asize'\geq \asize} \sepadd \earray{\ahead}{\asize} \sepadd \earray{\ahead'}{\asize'} -\potent
	\end{align*}
	\item $\aert{\potent}{\arraydelete{\ahead}{\asize}-\potent}{\iemp} \preceq \earray{\ahead}{\asize} - \potent$
\end{itemize}

\begin{figure}
	
	\begin{align*}
	&\asucceqannotate{\emprun{3}  \sepadd \pureemp{\aoff\leq \asize \wedge \asize \geq 1} \sepadd \earray{\ahead}{\asize}} \\
	& \asucceqannotate{\pureemp{\aoff\leq \asize \wedge \asize \geq 1} \sepadd \emprun{1} \sepadd \earray{\ahead}{\asize} \sepadd ( \iverson{\aoff  =\asize} \cdot (\nicefrac{1}{\asize +1} \cdot(\emprun{\asize+1}) + \nicefrac{\asize}{\asize +1} \cdot 2) + \iverson{\aoff  \neq\asize} \cdot 2)} \\
	&\aaertannotate{\iverson{\aoff  =\asize} \cdot (\nicefrac{1}{\asize +1} \cdot(\emprun{1} \sepadd \pureemp{\aoff \leq \asize+1 \wedge \asize+1 \geq 1} \sepadd  \emprun{\asize} \sepadd \earray{\ahead}{\asize} \sepadd \emprun{\potent\subst{\aoff}{\aoff +1}\subst{\asize}{\asize+1}} - \potent)} \\
	&\aphantannotate{\quad {}+ \nicefrac{\asize}{\asize +1} \cdot(\emprun{1} \sepadd \pureemp{\aoff \leq 2\cdot \asize \wedge 2\cdot \asize \geq 1} \sepadd  \emprun{\asize} \sepadd \earray{\ahead}{\asize} \sepadd \emprun{\potent\subst{\aoff}{\aoff +1}\subst{\asize}{2\cdot \asize}}) - \potent)} \\
	&\aphantannotate{{}+ \iverson{\aoff  \neq\asize} \cdot (\emprun{1} \sepadd 0 \sepadd \ivalidpointer{\ahead + \aoff} \sepadd \emprun{\potent\subst{\aoff}{\aoff +1}} - \potent)} \\
	&\IF{\aoff  = \asize} \\
	&\qquad \aaertannotate{\nicefrac{1}{\asize +1} \cdot(\emprun{1} \sepadd \pureemp{\aoff \leq \asize+1 \wedge \asize+1 \geq 1} \sepadd  \emprun{\asize} \sepadd \earray{\ahead}{\asize} \sepadd \emprun{\potent\subst{\aoff}{\aoff +1}\subst{\asize}{\asize+1}} - \potent)} \\
	&\qquad \aphantannotate{\quad {}+ \nicefrac{\asize}{\asize +1} \cdot(\emprun{1} \sepadd \pureemp{\aoff \leq 2\cdot \asize \wedge 2\cdot \asize \geq 1} \sepadd  \emprun{\asize} \sepadd \earray{\ahead}{\asize} \sepadd \emprun{\potent\subst{\aoff}{\aoff +1}\subst{\asize}{2\cdot \asize}}- \potent)} \\
	&\qquad  \PCHOICE{\ASSIGN{\asize'}{\asize+1}}{\nicefrac{1}{\asize +1}}{\ASSIGN{\asize'}{2\cdot \asize}} \fatsemi \\
	&\qquad \aunframeannotate{\emprun{1} \sepadd \pureemp{\aoff \leq \asize' \wedge \asize' \geq 1} \sepadd  \emprun{\asize} \sepadd \earray{\ahead}{\asize} \sepadd \emprun{\potent\subst{\aoff}{\aoff +1}\subst{\asize}{\asize'}} - \potent}{} \\
	&\qquad \aannotate{\iemp - \potent} \\
	& \qquad  \ALLOC{\ahead'}{\asize'} \fatsemi \\
	&\qquad \aframeannotate{\earray{\ahead'}{\asize'} - \potent}{} \\
	&\qquad \asucceqannotate{\emprun{1} \sepadd \pureemp{\aoff \leq \asize' \wedge \asize' \geq 1} \sepadd  \emprun{\asize} \sepadd \earray{\ahead}{\asize} \sepadd \earray{\ahead'}{\asize'} \sepadd \emprun{\potent\subst{\aoff}{\aoff +1}\subst{\asize}{\asize'}} - \potent}{} \\
	&\qquad \aunframeannotate{\emprun{1} \sepadd \pureemp{\aoff \leq \asize' \wedge \asize' \geq 1} \sepadd \pureemp{\asize' \geq \asize } \sepadd \emprun{\asize} \sepadd \earray{\ahead}{\asize} \sepadd \earray{\ahead'}{\asize'} \sepadd \emprun{\potent\subst{\aoff}{\aoff +1}\subst{\asize}{\asize'}} - \potent}{} \\
	&\qquad \aannotate{\pureemp{\asize' \geq \asize } \sepadd \emprun{\asize} \sepadd \earray{\ahead}{\asize} \sepadd \earray{\ahead'}{\asize'} - \potent} \\
	& \qquad \arraycopy{\ahead}{\asize}{\ahead'} \fatsemi \\
	&\qquad \aframeannotate{\earray{\ahead}{\asize} \sepadd \earray{\ahead'}{\asize'} - \potent}{} \\
	&\qquad \asucceqannotate{\emprun{1} \sepadd \pureemp{\aoff \leq \asize' \wedge \asize' \geq 1} \sepadd \earray{\ahead}{\asize} \sepadd \earray{\ahead'}{\asize'} 
		\sepadd \emprun{\potent\subst{\aoff}{\aoff +1}\subst{\asize}{\asize'}}}{} \\
	&\qquad \aunframeannotate{\emprun{1} \sepadd 0 \sepadd \ivalidpointer{\ahead' + \aoff}  \sepadd \earray{\ahead}{\asize}
		\sepadd \emprun{\potent\subst{\aoff}{\aoff +1}\subst{\asize}{\asize'}} -\potent}{} \\
	&\qquad \aannotate{\earray{\ahead}{\asize} - \potent} \\
	&\qquad \arraydelete{\ahead}{\asize} \fatsemi \\
	& \qquad \aframeannotate{\iemp - \potent}{} \\
	& \qquad \aaertannotate{\emprun{1} \sepadd 0 \sepadd \ivalidpointer{\ahead' + \aoff} 
		\sepadd \emprun{\potent\subst{\aoff}{\aoff +1}\subst{\asize}{\asize'}} - \potent} \\
	& \qquad \ASSIGN{\ahead}{\ahead'}\fatsemi \\
	& \qquad \aaertannotate{\emprun{1} \sepadd 0 \sepadd \ivalidpointer{\ahead + \aoff} 
		\sepadd \emprun{\potent\subst{\aoff}{\aoff +1}\subst{\asize}{\asize'}} - \potent} \\
	& \qquad \ASSIGN{\asize}{\asize'}\\
	& \qquad \aannotate{\emprun{1} \sepadd 0 \sepadd \ivalidpointer{\ahead + \aoff} \sepadd \emprun{\potent\subst{\aoff}{\aoff +1}} - \potent} \\
	& \} \\
	&\aannotate{\ldots \text{see \Cref{fig:dynamic_list_aert_proof_2}}} \\
	&\ldots
	\end{align*}
	
	\caption{Randomized Dynamic Table Amortized Expected Runtime Verification (part 1).}
	\label{fig:dynamic_list_aert_proof_1}
\end{figure}

\begin{figure}
	
	\begin{align*}
	& \ldots \text{see \Cref{fig:dynamic_list_aert_proof_1}} \\
	&\aunframeannotate{\emprun{1} \sepadd 0 \sepadd \ivalidpointer{\ahead + \aoff} \sepadd \emprun{\potent\subst{\aoff}{\aoff +1}} -\potent}{} \\
	&\amutannotate{\ivalidpointer{\ahead + \aoff} -\potent} \\
	& \HASSIGN{\ahead + \aoff}{y}\fatsemi \\
	&\aframeannotate{\isingleton{\ahead + \aoff}{y} -\potent}{} \\
	&\asucceqannotate{\emprun{1} \sepadd 0 \sepadd \isingleton{\ahead + \aoff}{y} \sepadd \emprun{\potent\subst{\aoff}{\aoff +1}} - \potent} \\
	&\aaertannotate{1+ \potent\subst{\aoff}{\aoff +1} - \potent} \\
	&\ASSIGN{\aoff}{\aoff +1}\fatsemi \\
	%
	&\aaertannotate{1} \\
	& \TICK{1} \\
	& \aannotate{0} \\
	\end{align*}
	
	\caption{Randomized Dynamic Table Amortized Expected Runtime Verification (part 2).}
	\label{fig:dynamic_list_aert_proof_2}
\end{figure}

\end{document}